\documentclass[ejs]{imsart}
\RequirePackage{amsthm,amsmath,amssymb,color}
{
	\theoremstyle{plain}
	
}
\RequirePackage{natbib}
\RequirePackage[colorlinks,citecolor=blue,urlcolor=blue]{hyperref}
\usepackage{longtable}
\usepackage{epsfig,graphicx,float}
\usepackage[bottom]{footmisc}
\usepackage{indentfirst}
\usepackage{endnotes}
\usepackage{verbatim}
\usepackage{threeparttable}
\usepackage{comment}
\usepackage{multirow}
\usepackage[titletoc]{appendix}

\usepackage{algorithm}
\usepackage{algorithmic}

\doi{10.1214/154957804100000000}
\pubyear{0000}
\volume{0}
\firstpage{1}
\lastpage{1}

\startlocaldefs

\newtheoremstyle{remark}
{12pt}
{12pt}
{\normalfont}
{0pt}
{\bfseries}
{}
{10pt}
{\thmname{#1}\thmnumber{ #2}: \bfseries\thmnote{#3}}

\newtheoremstyle{theorems}
{12pt}
{12pt}
{\itshape}
{17pt}
{\bfseries}
{}
{10pt}
{\thmname{#1}\thmnumber{ #2}. \bfseries\thmnote{#3}}

\newtheoremstyle{model}
{8pt}
{8pt}
{\itshape}
{17pt}
{\bfseries}
{}
{10pt}
{\thmname{#1}\thmnumber{ #2}: \bfseries\thmnote{#3}}

\theoremstyle{theorems}

\newtheorem{lemma}{Lemma}
\newtheorem{theorem}{Theorem}

\theoremstyle{remark}

\theoremstyle{model}

\newcommand{\indep}{\;\, \rule[0em]{.03em}{.6em} \hspace{-.25em}      
	\rule[0em]{.65em}{.03em} \hspace{-.25em} \rule[0em]{.03em}{.6em}\;\,}

\newcommand{\E}{\mathrm{E}}                

\newcommand{\sign}[1]{\mathrm{sgn}(#1)}

\def\cs{{\mathcal S}\lo {Y | X}}

\newcommand\cms {{\mathcal S}\lo {E(Y | X)}}

\def\nano{\scriptscriptstyle}
\def\trans{^{\mbox{\tiny{\sf T}}}}
\newcommand\spn[1]{{\mathcal S} ({#1})}

\newcommand\udex[1]{^{\raisebox{1.2pt}{\mbox{$\nano #1$}}}}
\newcommand\lo[1]{_{\nano #1}}
\def\inv{\udex {-1}}
\def\tsum{\textstyle\sum}

\def\var{\mathrm{var}}

\def\real{\mathbb {R}}
\def\eop
{\hfill $\Box$
}

\def\bop{O\lo P}
\def\sop{o\lo P}

\def\one{\mbox{I}}
\def\two{\mbox{I\!I}}

\def\vecc{{\mathrm {vec}}}
\def\prob{\mathrm {prob}}
\def\I{\mathrm I}
\def\diag{{\mathrm {diag}}}
\def\F{{\mathcal {F}}}

\def\msir{\mathrm{SIR}\lo{\mathrm {M}}}

\def\rmsir{\mathrm{SIR}\lo{\mathrm {RM}}}
\def\smsir{\mathrm{SIR}\lo{\mathrm {M}}\udex {\mathrm {S}}}
\def\srmsir{\mathrm{SIR}\lo{\mathrm {RM}}\udex {\mathrm {S}}}
\def\msave{\mathrm{SAVE}\lo{\mathrm {M}}}
\def\rmsave{\mathrm{SAVE}\lo{\mathrm {RM}}}
\def\go{{\mathrm {G}}\lo 1}
\def\hgo{\widehat {\mathrm {G}}\lo 1}
\def\G{{\mathrm {G}}}

\endlocaldefs

\begin{document}

\begin{frontmatter}

\title{Adjusting inverse regression for predictors with clustered distribution}
\runtitle{Inverse regression for predictors with clustered distribution}


\author{\fnms{Wei} \snm{Luo}\ead[label=e1]{weiluo@zju.edu.cn}}
\and
\author{\fnms{Yan} \snm{Guo}\ead[label=e2]{yan\_guo@zju.edu.cn}}

\address{Center for Data Science, Zhejiang University, China 310058\\
\printead{e1,e2}}

\runauthor{W. Luo}

\begin{abstract}
A major family of sufficient dimension reduction (SDR) methods, called inverse regression, commonly require the distribution of the predictor $X$ to have a linear $E(X|\beta\trans X)$ and a degenerate $\var (X | \beta\trans X)$ for the desired reduced predictor $\beta\trans X$. In this paper, we adjust the first- and second-order inverse regression methods by modeling $E(X|\beta\trans X)$ and $\var (X | \beta\trans X)$ under the mixture model assumption on $X$, which allows these terms to convey more complex patterns and is most suitable when $X$ has a clustered sample distribution. The proposed SDR methods build a natural path between inverse regression and the localized SDR methods, and in particular inherit the advantages of both; that is, they are $n\udex {1/2}$-consistent, efficiently implementable, directly adjustable under the high-dimensional settings, and fully recovering the desired reduced predictor. These findings are illustrated by simulation studies and a real data example at the end, which also suggest the effectiveness of the proposed methods for non-clustered data.
\end{abstract}




\end{frontmatter}

\section{Introduction}\label{sec:Intro}

Sufficient dimension reduction (SDR) has attracted intensive research interest in the recent years due to the need of high-dimensional statistical data analysis. Usually, SDR assumes that the predictor $X$ affects the response $Y$ through a lower-dimensional $\beta\trans X$, that is,
\begin{align}\label{assume: sdr cs}
	Y \indep X \mid \beta\trans X
\end{align}
where $\indep$ means independence; and the research interest of SDR is to find the qualified $\beta\trans X$ as the subsequent predictor in a model-free manner, i.e. without parametric assumptions on $Y|\beta\trans X$. For identifiability, \cite{cook1998} introduced the central subspace, denoted by $\cs$, as the parameter of interest, whose arbitrary basis matrix satisfies (\ref{assume: sdr cs}) with minimal number of columns. This space is unique under fairly general conditions \citep{yin2008}.

When the research interest is focused on regression analysis, SDR is adjusted to find $\beta\trans X$ that only preserves the information about regression, i.e. $E(Y | X)$ being measurable with respect to $\beta\trans X$. Similar to $\cs$, the unique parameter of interest in this scenario is called the central mean subspace and denoted by $\cms$ \citep{cook2002}. For ease of presentation, we do not distinguish the two parameters $\cms$ and $\cs$ unless otherwise specified, and call both the central SDR subspace uniformly. Let $d$ be the dimension of the central SDR subspace and $\beta\lo 0 \in \real\udex {p\times d}$ be its arbitrary basis matrix.

\def\M{{\Omega}}

In the literature, a major family of estimators for the central SDR subspace is called inverse regression, which includes the ordinary least squares \citep[OLS;][]{li1989} and principal Hessian directions \citep[pHd;][]{li1992} that estimate $\cms$, and sliced inverse regression \citep[SIR;][]{li1991}, sliced average variance estimator \citep[SAVE;][]{cook1991}, and directional regression \citep[DR;][]{li2007} that estimate $\cs$, etc. A common feature of the inverse regression methods is that they first use the moments of $X | Y$ to construct a matrix-valued parameter $\M$, whose columns fall in the central SDR subspace under certain parametric assumptions of $X$, and then use an estimate of $\M$ to recover the central SDR subspace. Suppose $X$ has zero mean and covariance matrix $\Sigma\lo X$. The forms of $\M$ for OLS, SIR, and SAVE are, respectively,
\begin{align} \label{eq: M}
	\Sigma\lo X\inv E(XY), \quad \Sigma\lo X\inv E\{E\udex {\otimes 2}(X|Y)\} \Sigma\lo X\inv, \quad  \Sigma\lo X\inv E[E\{\Sigma\lo X - \var (X | Y)\}\udex {\otimes 2}] \Sigma\lo X\inv,
\end{align}
where $v\udex {\otimes 2}$ denotes $vv\trans$ for any matrix $v$. Clearly, $\M$ is readily estimable when $Y$ is discrete with limited number of outcomes. To ease the estimation for continuous $Y$ (except for OLS), the slicing technic is commonly applied to replace a univariate $Y$ in $\M$ with $Y\lo S = \tsum\lo {i=1}\udex H i I(q\lo {i-1} < Y \leq q\lo i)$, $H$ being a fixed integer and $q\lo i$ being the $(i/H)$th quantile of $Y$; the case of multivariate $Y$ follows similarly. Because the moments of $X|Y\lo S$ can be estimated by the sample moments within each slice $(q\lo i, q\lo {i+1})$, a consistent estimator of $\M$, denoted by $\widehat \M$, can be readily constructed. The central SDR subspace is then estimated by the linear span of the leading left singular vectors of $\widehat \M$.

Despite their popularity in applications, the inverse regression methods are also commonly criticized for the aforementioned parametric assumptions on $X$. For example, the first-order inverse regression methods, i.e.  OLS and SIR that only involve $E(X |Y)$ or its average $E(XY)$, require the linearity condition
\begin{align}\label{assume: linearity}
	E(X | \beta\lo 0\trans X) = P\trans (\Sigma\lo X, \beta\lo 0) X
\end{align}
where $P(\Sigma\lo X, \beta) = \beta (\beta\trans \Sigma\lo X\beta)\inv \beta\trans \Sigma\lo X$ is the projection matrix onto the column space of $\beta$ under the inner product $\langle u, v\rangle = u\trans \Sigma\lo X v$. The second-order inverse regression methods that additionally involve $\var(X| Y)$, e.g. pHd, SAVE, and directional regression, require both the linearity condition (\ref{assume: linearity}) and the constant variance condition
\begin{align}\label{assume: const var}
	\var(X | \beta\lo 0\trans X) = \Sigma\lo X - \Sigma\lo X P (\Sigma\lo X, \beta\lo 0).
\end{align}
Since both conditions adopt simple parametric forms on the moments of $X|\beta\lo 0\trans X$, they can be violated in practice. Moreover, as the central SDR subspace is unknown, both (\ref{assume: linearity}) and (\ref{assume: const var}) are often strengthened to that they hold for any $\beta \in \real\udex {p \times d}$. The strengthened linearity condition requires $X$ to be elliptically distributed \citep{li1991}, and, together with the strengthened constant variance condition, $X$ must have a multivariate normal distribution \citep{cook1991}. \cite{hall1993} justified the approximate satisfaction of (\ref{assume: linearity}) for general $X$ as $p$ grows, but their theory is developed for diverging $p$, and, more importantly, it is built in a Bayesian sense that $\beta\lo 0$ follows a continuous prior distribution on $\real\udex {p\times d}$. In the literature of high-dimensional SDR \citep{chen2010, lin2016sparse,zeng2022}, sparsity is commonly assumed on the central SDR subspace so that $\beta\lo 0$ only has a few nonzero rows. Thus, \citeauthor{hall1993}'s result is often inapplicable.


To enhance the applicability of inverse regression, multiple methods have been proposed to relax the linearity condition (\ref{assume: linearity}). \cite{li2009} and \cite{dong2010} allow $E(X | \beta\lo 0\trans X)$ to lie in a general finite-dimensional functional space, and they reformulate the inverse regression methods accordingly as minimizing certain objective functions. However, the non-convexity of the objective functions may complicate the implementation, and, unlike the connection between the linearity condition and the ellipticity of $X$, the relaxed linearity condition is lack of interpretability. Another effort specified for SIR can be found in \cite{guan2017}, which assumes $X$ to follow a mixture of skew-elliptical distributions with identical shape parameters up to multiplicative scalars. \citeauthor{guan2017}'s method can be easily implemented once the mixture model of $X$ is consistently fitted.

In a certain sense, the minimum average variance estimator \citep[MAVE;][]{xia2002}, a state-of-the-art SDR method that estimates the central mean subspace $\cms$ by local linear regression, can also be regarded as relaxing the linearity condition (\ref{assume: linearity}) for OLS. Recall that OLS is commonly used to estimate the coefficients, i.e. the one-dimensional $\cms$, of the linear model of $E(Y | \beta\lo 0\trans X)$. Thus, the linearity condition (\ref{assume: linearity}) on $X$ can be regarded as a surrogate to the assumption of linear $E(Y|\beta\lo 0\trans X)$ for SDR estimation. As MAVE adopts and fits a linear $E(Y|\beta\lo 0\trans X)$ in each local neighborhood of $X$, it equivalently adopts the linearity condition in these neighborhoods, which only requires the first-order Taylor expansion of $E(X|\beta\lo 0\trans X)$ and thus is fairly general.

Following the same spirit as MAVE, the aggregate dimension reduction \citep[ADR;][]{wang2020} conducts SIR in the local neighborhoods of $X$ and thus also relaxes the linearity condition. Recently, the idea of localization was also studied in \cite{fertl2022ensemble}, who proposed both the cumulative covariance estimator (CVE) to recover $\cms$ and the enemble CVE (eCVE) to recover $\cs$. For convenience, we call MAVE, ADR, CVE, eCVE, and other SDR methods based on localized estimation uniformly the MAVE-type methods. By the nature of localized estimation, these methods share the common limitation that they quickly lose consistency as $p$ grows.

Generally, the conjugate between the parametric assumptions on $X|\beta\lo 0\trans X$, i.e. (\ref{assume: linearity}) and (\ref{assume: const var}), and the parametric assumptions on $Y|\beta\lo 0\trans X$ can be formulated under the unified SDR framework proposed in \cite{ma2012}. In the population level, each SDR method in this framework solves the estimating equation
\begin{align}\label{eq: ee general}
	E  [\{\alpha (X) - \mu\lo \alpha (\beta\trans X) \} \{ g (Y, \beta\trans X) - \mu\lo g (\beta\trans X) \}] = 0
\end{align}
for certain $\alpha (X)$ and $g (Y, \beta\trans X)$, where $\mu\lo \alpha (\beta\trans X)$ and $\mu\lo g (\beta\trans X)$ are the working models for $E\{\alpha (X) | \beta\trans X\}$ and $E\{g (Y, \beta\trans X) | \beta\trans X\}$, respectively. \cite{ma2012} justified a double-robust property that, as long as either $\mu\lo \alpha (\beta\lo 0\trans X)$ or $\mu\lo g (\beta\lo 0\trans X)$ truly specifies the corresponding regression function when $\beta$ is fixed at $\beta\lo 0$, any minimal solution to (\ref{eq: ee general}) must be included in the central SDR subspace. Up to asymptotically negligible errors, (\ref{eq: ee general}) incorporates a large family of SDR methods. For example, with $\alpha(X) = X$ and the linearity condition (\ref{assume: linearity}), one can take linear $\mu\lo \alpha ( \beta\trans X )$ and force $\mu \lo g(\beta\trans X)$ to be zero, and solving (\ref{eq: ee general}) delivers OLS if $g(Y,\beta\trans X)$ is $Y$ and delivers SIR if it is $(I(Y\lo S=1),\ldots, I(Y\lo S = H))$.

The double-robustness property for (\ref{eq: ee general}) illuminates the tradeoff for general SDR methods: if we choose to avoid parametric assumptions or nonparametric estimation on $Y|\beta\lo 0\trans X$, then we must adopt those on $X|\beta\lo 0\trans X$ instead. In principle, like the rich literature for modeling $Y|\beta\lo 0\trans X$, numerous assumptions can be adopted on $X | \beta\lo 0\trans X$ to represent different balances between parsimoniousness and flexibility. Nonetheless, under the unified framework (\ref{eq: ee general}), only the two extremes have been widely studied in the SDR literature: the most aggressive parametric models (\ref{assume: linearity}) and (\ref{assume: const var}) in inverse regression, and the most conservative nonparametric model in MAVE. Although a compromise between the two was discussed by \cite{li2009} and \cite{dong2010} as mentioned above, there is lack of parametric learning of $X | \beta\lo 0\trans X$ that specifies easily interpretable conditional moments and is often satisfied in practice.

In this paper, we adjust both the first- and second-order inverse regression methods by approximating the distribution of $X$ with mixture models, where we allow each mixture component to differ in either center or shape but it must satisfy the linearity condition (\ref{assume: linearity}) and the constant variance condition (\ref{assume: const var}). The mixture model approximation has been widely used in applications to fit generally unknown distributions \citep{lindsay1995mixture,everitt2013finite}, especially those with clustered sample support, and it permits estimating each of $E(X | \beta\lo 0\trans X)$ and $\var (X | \beta\lo 0\trans X)$ by one of multiple commonly seen parametric models chosen in a data-driven manner. As supported by numerical studies, the proposed methods also have a consistent sample performance under more complex settings, for example, if $X$ has a unimodal but curved sample support. Compared with the MAVE-type methods, our approach is more robust to the dimensionality and in particular effective under the high-dimensional settings. Compared with \cite{guan2017}, it allows the mixture components to have different shapes and is applicable for all the second-order inverse regression methods. It also alleviates the issue of non-exhaustive recovery of the central SDR subspace that persists in OLS and SIR, in a similar fashion to MAVE and ADR that capture the local patterns of data.

The rest of the paper is organized as follows. Throughout the theoretical development, we assume $X$ to be continuous and follow a mixture model, under which the parametric forms of the moments of $X | \beta\lo 0\trans X$ are studied in Section $2$. We adjust SIR in Section $3$, and further modify it towards sparsity under the high-dimensional settings in Section 4. Section 5 is devoted for the adjustment of SAVE. The simulation studies are presented in Section 6, where we evaluate the performance of the proposed methods under various settings, including the cases where $X$ does not follow a mixture model. A real data example is investigated in Section 7. The adjustment for OLS is omitted for its similarity to SIR. The adjustments of other second-order inverse regression methods, as well as the detailed implementation and complementery simulation studies, are deferred to the Appendix. R code for implementing the proposed methods can be downloaded at \url{https://github.com/Yan-Guo1120/IRMN}. For convenience, we focus on continuous $Y$ whose sample slices have an equal size, and we regard the dimension $d$ of the central SDR subspace as known a priori. The generality of the latter is briefly justified in Section 2, given the consistent estimation of $d$ discussed in Sections 3 and 5. 

\section{Preliminaries on mixture models}\label{sec:mixnorm}

We first clarify the notations. Suppose $X= (X\lo 1,\ldots, X\lo p)\trans$. For any matrix $\beta \in \real\udex {p\times d}$, let $\spn \beta$ be the column space of $\beta$ and let $\vecc(\beta)$ be the vector that stacks the columns of $\beta$ together. Like $\beta\udex {\otimes 2}$ in (\ref{eq: M}), sometimes we also denote $\beta$ itself by $\beta\udex {\otimes 1}$. We expand $P(\Sigma\lo X, \beta)$ in (\ref{assume: linearity}) to $P(\Lambda, \beta)$ for any positive definite matrix $\Lambda$, and, with $\I\lo p$ being the identity matrix, we denote $\I\lo p - P(\Lambda, \beta)$ by $Q(\Lambda, \beta)$. For a set of matrices $\{\M\lo i: i \in \{1,\ldots, k\}\}$ with equal number of rows, we denote the ensemble matrix $(\M\lo 1,\ldots, \M\lo k)$ by $(\M\lo i)\lo {i\in\{1,\ldots, k\}}$.

In the literature, the mixture model has been widely employed to approximate unknown distributions \citep{lindsay1995mixture,everitt2013finite}, both clustered and non-clustered, as it provides a wide range of choices to balance between the estimation efficiency and the modeling flexibility, especially for the local patterns of data. Application fields include, for example, astronomy, bioinformatics, ecology, and economics \citep[][Subsection 1.1]{marin2005bayesian}. Given a parametric family of probability density functions $\F = \{f(x,\alpha)\}$ on $\real\udex p$, a mixture model for $X$ means
\begin{align}\label{assume: mNormal x}
X = \tsum\lo {i=1}\udex q I(W = i) Z\udex {(i)},
\end{align}
where $q$ is an unknown integer, each mixture component $Z\udex {(i)}$ is generated from $f(\cdot,\alpha\lo i) \in \F$, $W$ is a latent variable that follows a multinomial distribution with support $\{1,\ldots, q\}$ and probabilities $\pi\lo 1,\ldots, \pi\lo q$, and $W, Z\udex {(1)},\ldots, Z\udex {(q)}$ are mutually independent. For the identifiability of (\ref{assume: mNormal x}), we require each $\pi\lo i$ to be positive and each $\alpha\lo i$ to be distinct, as well as additional regularity conditions on $\F$. The latter can be found in \cite{titterington1985}, details omitted. As a special case, when $\F$ is the family of multivariate normal distributions, (\ref{assume: mNormal x}) becomes a mixture multivariate normal distribution with $\alpha\lo i$ being $(\mu\lo i, \Sigma\lo i)$, the mean and the covariance matrix of the $i$th mixture component. Here, we allow $\Sigma\lo i$'s to differ.

For the most flexibility of the mixture model (\ref{assume: mNormal x}), one can allow the number of mixture components $q$ to be relatively large or even diverge with $n$ in practice, so that it can ``facilitate much more careful description of complex systems" \citep[][Subsection 1.2.3]{marin2005bayesian}, including those distributions who convey non-clustered sample supports. An extreme case is the kernel density estimation that uses an average of $n$ multivariate normal distributions to approximate a general distribution, which can be regarded as a mixture multivariate normal distribution mentioned above but with $q=n$. For simplicity, we set $q$ to be fixed as $n$ grows throughout the theoretical development, by which (\ref{assume: mNormal x}) is a parametric analogue of kernel density estimation, with each mixture component being a ``global" neighborhood of $X$. The setting of non-clustered $X$ approximated by mixture model will be investigated numerically in Section $6$.

Given that $\F$ truly specifies the distributions of the mixture components, the consistency of the mixture model (\ref{assume: mNormal x}) hinges on the true determination of $q$. This can be achieved by multiple methods, such as the Bayesian information criterion \citep[BIC;][]{zhao2008knee}, the integrated completed likelihood \citep{biernacki2000assessing}, the sequential testing procedure \citep{mclachlan2019finite}, and the recently proposed method based on data augmentation \citep{luo2022determine}. In the presence of these estimators, we regard $q$ as known {\it a priori} throughout the theoretical study to ease the presentation. The generality of doing so is justified in the following lemma, which shows that the asymptotic properties of any statistic that involves $q$ will be invariant if $q$ is replaced with its arbitrary consistent estimator $\widehat q$. The same invariance property also holds for using the true value of the dimension $d$ of the central SDR subspace in place of its consistent estimator, which justifies the generality of assuming $d$ known mentioned at the end of the Introduction.

\begin{lemma}\label{lemma: q hat q}
Suppose $\widehat q$ is a consistent estimator of $q$ in the sense that $\prob(\widehat q = q) \rightarrow 1$. For any statistic $S\lo n$ that involves $\widehat q$, let $R\lo n$ be constructed in the same way as $S\lo n$ but with $\widehat q$ replaced by $q$. Then we have $S\lo n - R\lo n = \sop (n\udex {-r})$ for any $r > 0$, that is, $R\lo n$ is asymptotically equivalent to $S\lo n$.
\end{lemma}

\begin{proof}
For any $r > 0$ and $\delta > 0$, we have
\begin{align*}
&\prob(|S\lo n-R\lo n|>n\udex {-r} \delta)\\
& \hspace{.5cm} = \prob(|S\lo n-R\lo n|>n\udex {-r} \delta,\widehat q=q)+\prob(|S\lo n-R\lo n|>n\udex {-r } \delta,\widehat q\neq q)\\
& \hspace{.5cm} = \prob(|S\lo n-R\lo n|>n\udex {-r} \delta,\widehat q\neq q) \leq \prob(\widehat q\neq q)
\rightarrow 0
\end{align*}
The second equation above is due to $S\lo n=R\lo n$ whenever $\widehat q=q$. Thus, we have $S\lo n - R\lo n = \sop (n\udex {-r})$. This completes the proof. 
\end{proof}

Given $q$, the marginal density of $X$ is $\tsum\lo {i=1}\udex q \pi\lo i f(x,\alpha\lo i)$ up to the unknown $(\pi\lo i, \alpha\lo i)$'s, which generates a maximal likelihood approach to estimate these parameters using the EM algorithm; see \cite{cai2019} and \cite{mclachlan2019} for the relative literature.  Hereafter, we treat the estimators $\widehat \pi\lo i$ and $\widehat \alpha\lo i$ as granted, and assume their $n\udex {1/2}$-consistency and asymptotic normality whenever needed. The same applies to $\mu\lo i$ and $\Sigma\lo i$ defined above, where $\widehat \mu\lo i$ and $\widehat \Sigma\lo i$ for general parametric family $\F$ can be derived by appropriate transformations of $\widehat \alpha\lo i$.

To connect the mixture model (\ref{assume: mNormal x}) with a parametric model on the moments of $X|\beta\lo 0\trans X$, we assume that the linearity condition and the constant variance condition holds for each mixture component; that is,
for $i\in\{1,\ldots, q\}$,
\begin{align}
	& E (X | \beta\lo 0\trans X, W = i) = P\trans (\Sigma\lo i, \beta\lo 0)(X - \mu\lo i), \label{assume: linearity Wi} \\
	& \var (X | \beta\lo 0\trans X, W = i) = \Sigma\lo i - \Sigma\lo i P (\Sigma\lo i, \beta\lo 0). \label{assume: const var Wi}
\end{align}
Referring to the discussion about (\ref{assume: linearity}) and (\ref{assume: const var}) in the Introduction, this includes but is not limited to the cases where each mixture component of $X$ is multivariate normal. For clarification, hereafter we call (\ref{assume: linearity}) the overall linearity condition and call (\ref{assume: linearity Wi}) the mixture component-wise linearity condition whenever necessary, and likewise call (\ref{assume: const var}) and (\ref{assume: const var Wi}) the overall and the mixture component-wise constant variance conditions, respectively. Because $W$ is latent, we next provide two useful derivations of (\ref{assume: linearity Wi}). For each $i\in\{1,\ldots, q\}$, let $\pi\lo i(X)$ be $\prob(W = i | X)$, which is a known function up to $(\pi\lo i, \alpha\lo i)$'s, and let $\pi\lo i(\beta\trans X)$ be $\prob(W = i | \beta\trans X)$ for any $\beta\in \real\udex {p\times d}$. Generally, $\pi\lo i(\beta\trans X)$ differs from $\pi\lo i(X)$ unless in the extreme case $W\indep X | \beta\trans X$. The latter is not assumed in this article. The estimators $\widehat \pi\lo i (X)$ and $\widehat \pi\lo i(\beta\trans X)$ are readily derived from the mixture model fit.

\begin{lemma}\label{lemma: linear gen}
	Suppose $X$ follows a mixture model that satisfies (\ref{assume: linearity Wi}). For $i\in\{1,\ldots, q\}$, let $E\lo i (X | \beta\lo 0\trans X)$ be $E\{(X - \mu\lo i) \pi\lo i (X) | \beta\lo 0\trans X\}$. We have
	\begin{align}
		& E\lo i (X | \beta\lo 0\trans X) = \pi\lo i (\beta\lo 0\trans X) P\trans (\Sigma\lo i, \beta\lo 0)(X - \mu\lo i), \label{eq: mean xpi|x mnoromal} \\
		& E(X | \beta\lo 0\trans X) = \tsum\lo {i=1}\udex q \pi\lo i (\beta\lo 0\trans X) \{P\trans (\Sigma\lo i, \beta\lo 0) (X - \mu\lo i) + \mu\lo i\}.
		\label{eq: mean xbx mixNormal}
	\end{align}
\end{lemma}

\begin{proof}
To show (\ref{eq: mean xpi|x mnoromal}), we have
\begin{align*}
E\{(X-\mu\lo i) \pi \lo i (X) | \beta\lo 0\trans X\} & = E\{(X-\mu\lo i) I(W = i) | \beta\lo 0\trans X\} \\
& = \tsum\lo {j=1}\udex q \pi\lo j(\beta\lo 0\trans X) E\{(X-\mu\lo i) I(W = i) | \beta\lo 0\trans X, W = j\} \\
& = \pi\lo i (\beta\lo 0\trans X) E(X-\mu\lo i | \beta\lo 0\trans X, W = i) \\
& = \pi\lo i (\beta\lo 0\trans X) P\trans (\Sigma\lo i, \beta\lo 0)(X - \mu\lo i)
\end{align*}
To show (\ref{eq: mean xbx mixNormal}), we have
\begin{align*}
E(X | \beta\lo 0\trans X) = E\{E(X | \beta\lo 0\trans X, W) | \beta\lo 0\trans X\} = \tsum\lo {i=1}\udex q \pi\lo i (\beta\lo 0\trans X) \{P\trans (\Sigma\lo i, \beta\lo 0) (X - \mu\lo i) + \mu\lo i\}.
\end{align*}
This completes the proof.
\end{proof}

From the proof of Lemma \ref{lemma: linear gen}, $\pi\lo i (X)$ in $E\lo i (X | \beta\lo 0\trans X)$ can be replaced with the indicator $I(W=i)$, so (\ref{eq: mean xpi|x mnoromal}) resembles the mixture component-wise linearity condition (\ref{assume: linearity Wi}) for a single mixture component of $X$. Similarly, (\ref{eq: mean xbx mixNormal}) is an ensemble of (\ref{assume: linearity Wi}) across all the mixture components of $X$.

Depending on how $\pi\lo i(\beta\lo 0\trans X)$ is distributed, that is, how $W$ is associated with $\beta\lo 0\trans X$, the functional form of $E(X | \beta\lo 0\trans X)$ in (\ref{eq: mean xbx mixNormal}) can be approximated by various parametric models. In this sense, (\ref{eq: mean xbx mixNormal}) relaxes the overall linearity condition (\ref{assume: linearity}) to numerous parametric models, among which the most appropriate is automatically chosen by the data. For example, if the mixture components of $\beta\lo 0\trans X$ are well separate, then each $\pi\lo i(\beta\lo 0\trans X)$ is close to a Bernoulli distribution and thus $E(X | \beta\lo 0\trans X)$ is nearly a linear spline of $\beta\lo 0\trans X$. This is illustrated by $E(X\lo 1 | X\lo 2)$ in the upper-left and upper-middle panels of Figure~\ref{figure: approx mean var}, where $X$ is a balanced mixture of $N((0, 2)\trans, (r\lo 1\udex {|i-j|}))$ and $N((0, -2)\trans, (r\lo 2\udex {|i-j|}))$, with $r\lo 1 = .5$ and $r\lo 2 =-.5$ in the former and both equal to $.5$ in the latter. In the other extreme, if $\pi\lo i (\beta\lo 0\trans X)$'s are degenerate, which occurs if the mixture components of $X$ are identical along the directions in $\beta\lo 0$, then $E(X | \beta\lo 0\trans X)$ reduces to satisfy the overall linearity condition (\ref{assume: linearity}). When $\pi\lo i (\beta\lo 0\trans X)$'s are continuously distributed, $E(X | \beta\lo 0\trans X)$ conveys other forms, such as a quadratic function if we change $\mu\lo 1$ and $\mu\lo 2$ in the upper-left panel to $\mu\lo 1 = -\mu\lo 2 = (0, .2)\trans$, as depicted in the upper-right panel of Figure~\ref{figure: approx mean var}.

Similar to (\ref{assume: linearity Wi}), the mixture component-wise constant variance condition (\ref{assume: const var Wi}) can also be modified in two directions towards practical use, with the aid of $\pi\lo i (X)$ and $\pi\lo i (\beta\lo 0\trans X)$, respectively.

\begin{figure}[H]
\begin{center}
\includegraphics[width=.3\textwidth]{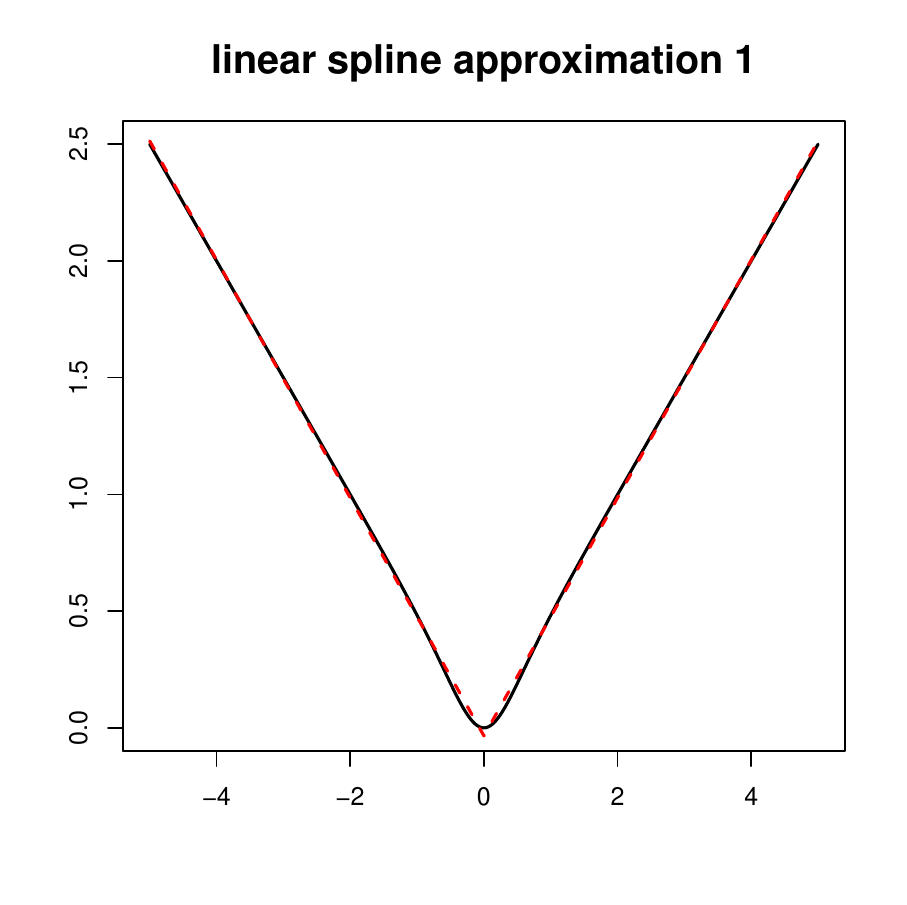}
\includegraphics[width=.3\textwidth]{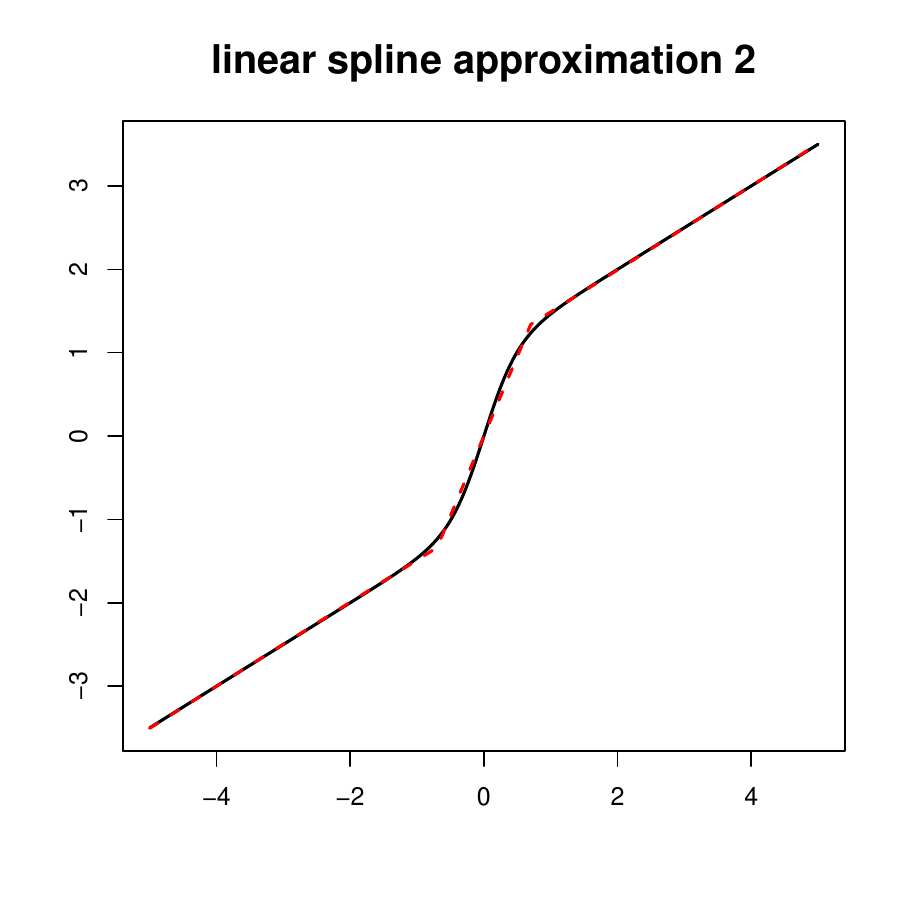}
\includegraphics[width=.3\textwidth]{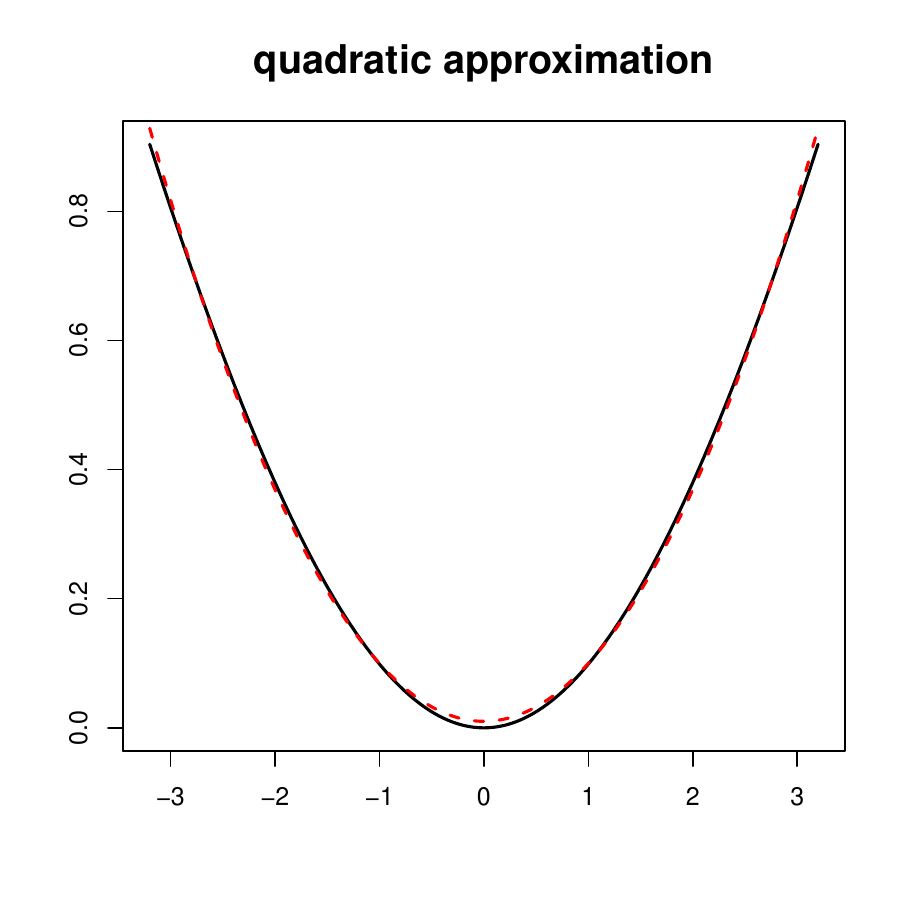}
\vspace{-.5cm}
\includegraphics[width=.3\textwidth]{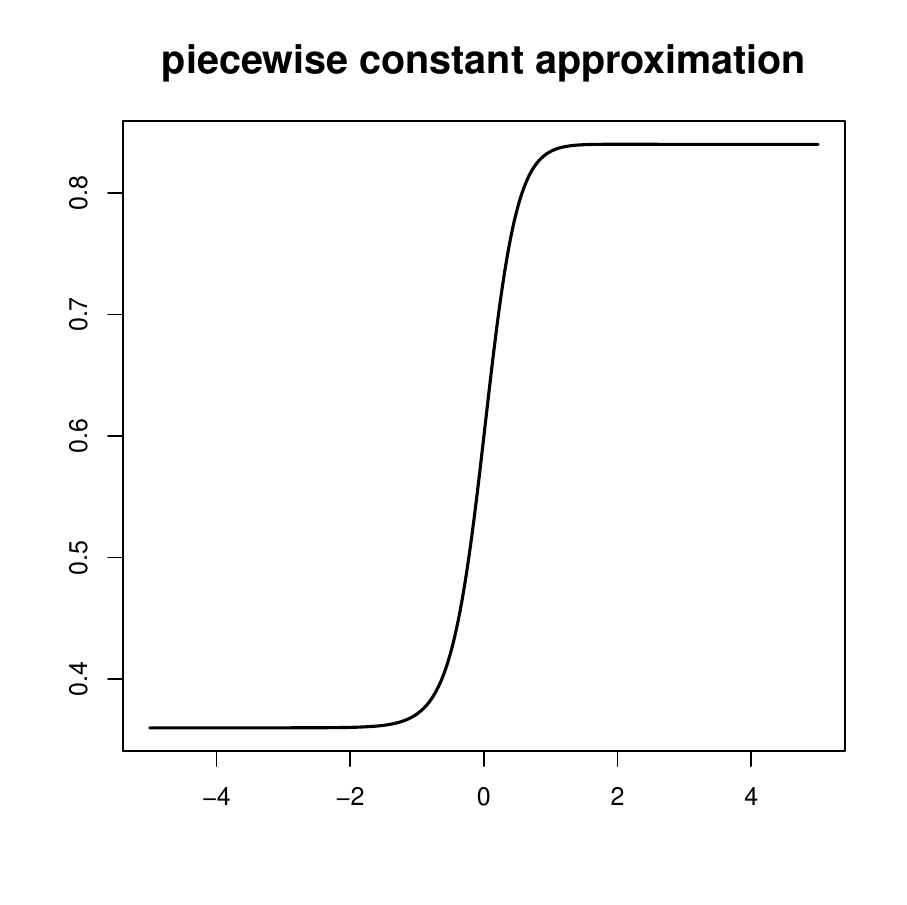}
\includegraphics[width=.3\textwidth]{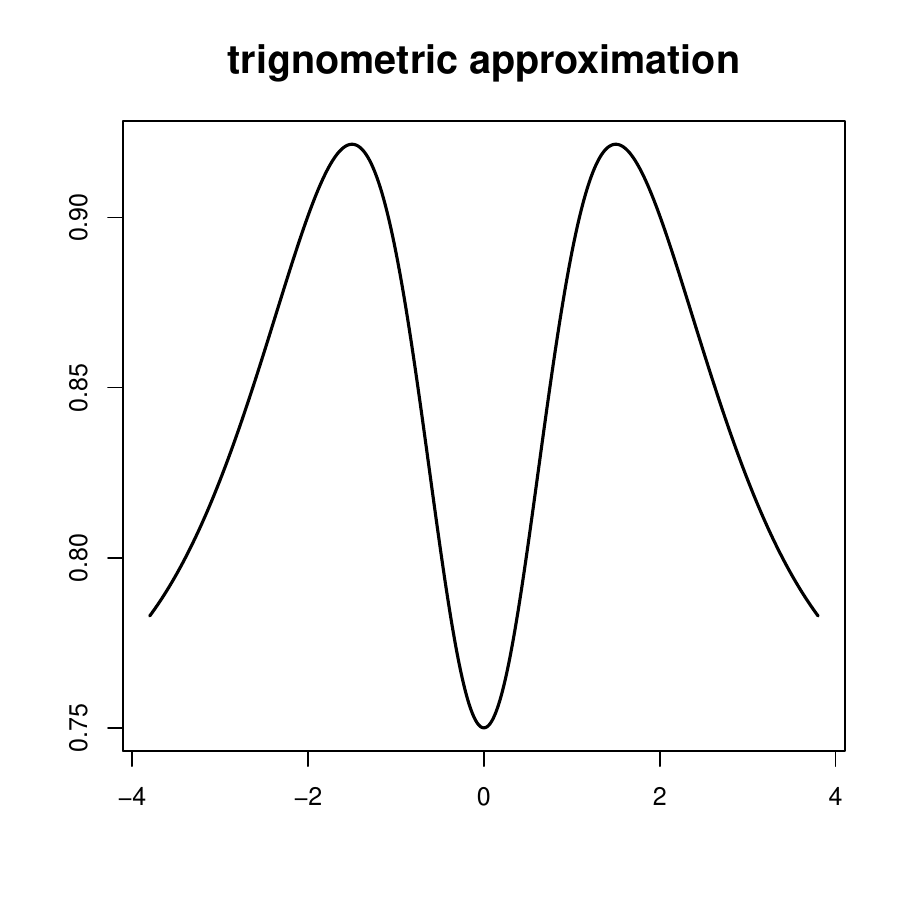}
\includegraphics[width=.3\textwidth]{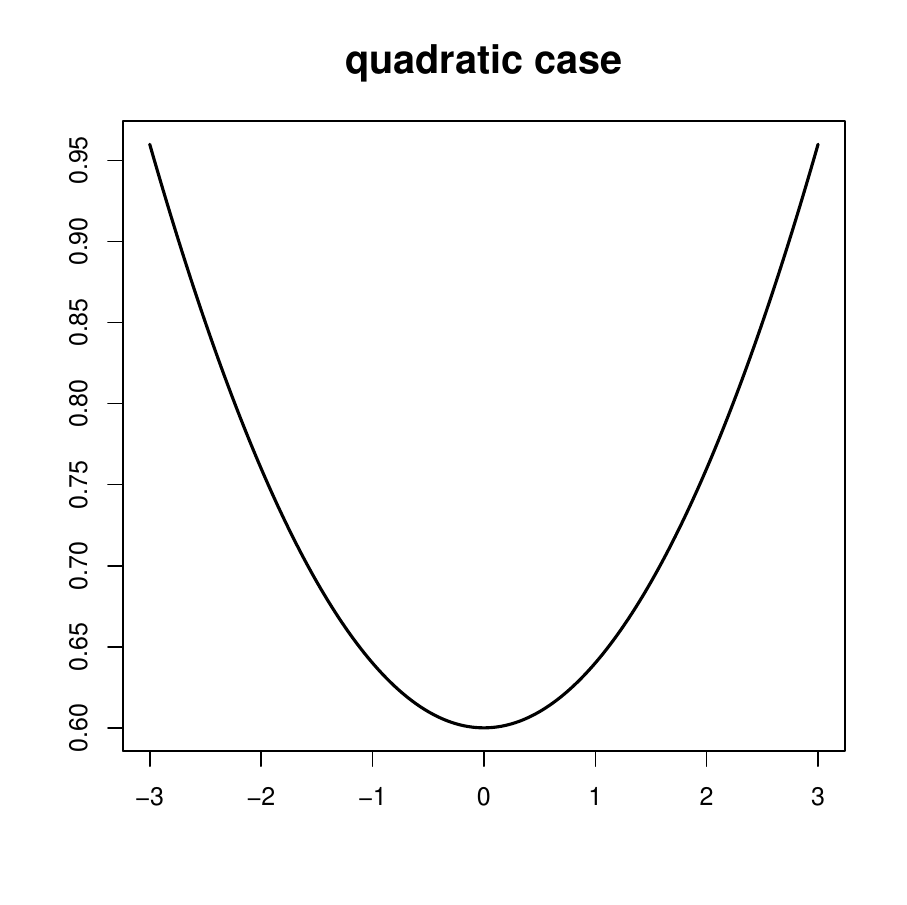}
\caption{\small In the upper panels, the solid curve in each plot is $E(X\lo 1 | X\lo 2)$, and the dashed curve is a symmetric linear spline, a monotone linear spline, and a quadratic function in the left, middle, and right panel, respectively. In the lower panels, the solid curve in each plot is $\var(X\lo 1 | X\lo 2)$ that approximates to a piecewise constant function, approximates to a trigonometric function, and is exactly quadratic, respectively.}
\label{figure: approx mean var}
\end{center}
\end{figure}

\begin{lemma}\label{lemma: con var gen}
Suppose $X$ follows a mixture model that satisfies (\ref{assume: linearity Wi}) and (\ref{assume: const var Wi}). Let $\var \lo i (X | \beta\lo 0\trans X) = E \{(X - \mu\lo i)\udex {\otimes 2} \pi\lo i (X) | \beta\lo 0\trans X\} - \{E\lo i (X | \beta\lo 0\trans X)\}\udex {\otimes 2}$ for $i = 1,\ldots, q$, and $\mu\lo {k}(\beta\lo 0\trans X) = \tsum\lo {i=1}\udex q \pi\lo i (\beta\lo 0\trans X) \{P\trans (\Sigma\lo i, \beta\lo 0) (X - \mu\lo i) + \mu\lo i\}\udex {\otimes k}$ for $k=1, 2$. We have,
\begin{align}
& \var\lo i (X | \beta\lo 0\trans X) = \pi\lo i (\beta\lo 0\trans X) [ \Sigma\lo i Q(\Sigma\lo i, \beta\lo 0) + \{1 - \pi\lo i (\beta\lo 0\trans X)\} \{P\trans (\Sigma\lo i, \beta\lo 0) (X - \mu\lo i)\}\udex {\otimes 2} ], \label{eq: var xpi|x mNormal} \\
& \var (X | \beta\lo 0\trans X) = \tsum\lo {i=1}\udex q \pi\lo i(\beta\lo 0\trans X) \Sigma\lo i Q(\Sigma\lo i, \beta\lo 0) + \mu\lo 2 (\beta\lo 0\trans X) - \mu\lo 1\udex {\otimes 2} (\beta\lo 0\trans X).
\label{eq: var xbx mixNormal}
\end{align}
\end{lemma}

\begin{proof}
To show (\ref{eq: var xpi|x mNormal}), we have
\begin{align*}
E \{(X-\mu\lo i)\udex {\otimes 2} \pi\lo i (X) | \beta\lo 0\trans X\} & = E [E \{(X-\mu\lo i)\udex {\otimes 2} I(W = i) | \beta\lo 0\trans X, W\} | \beta\lo 0\trans X] \\
& = \pi\lo i (\beta\lo 0\trans X) E \{(X-\mu\lo i)\udex {\otimes 2} | \beta\lo 0\trans X, W = i\} \\
& = \pi\lo i (\beta\lo 0\trans X) \{\var (X | \beta\lo 0\trans X, W = i) \\
&\qquad + E\udex {\otimes 2} (X - \mu\lo i | \beta\lo 0\trans X, W = i) \} \\
& = \pi\lo i (\beta\lo 0\trans X) [\Sigma\lo i Q(\Sigma\lo i, \beta\lo 0) + \{P\trans (\Sigma\lo i, \beta\lo 0) (X  - \mu\lo i) \}\udex {\otimes 2}].
\end{align*}
Together with (\ref{eq: mean xpi|x mnoromal}), this implies (\ref{eq: var xpi|x mNormal}).

To show (\ref{eq: var xbx mixNormal}), note that $\var\{E(X | \beta\lo 0\trans X, W) | \beta\lo 0\trans X\} = \mu\lo 2 (\beta\lo 0\trans X) - \mu\lo 1\udex {\otimes 2} (\beta\lo 0\trans X)$, by which we have
\begin{align*}
\var(X | \beta\lo 0\trans X) & = E\{\var (X | \beta\lo 0\trans X, W) | \beta\lo 0\trans X\} + \var\{E(X | \beta\lo 0\trans X, W) | \beta\lo 0\trans X\} \\
& = \tsum\lo {i=1}\udex q \pi\lo i(\beta\lo 0\trans X) \Sigma\lo i Q(\Sigma\lo i, \beta\lo 0) + \mu\lo 2 (\beta\lo 0\trans X) - \mu\lo 1\udex {\otimes 2} (\beta\lo 0\trans X).
\end{align*}
This completes the proof.
\end{proof}

Same as $E (X | \beta\lo 0\trans X)$ in (\ref{eq: mean xbx mixNormal}), the functional form of $\var (X | \beta\lo 0\trans X)$ in (\ref{eq: var xbx mixNormal}) can be approximated by various commonly seen parametric models. This is illustrated in the lower panels of Figure~\ref{figure: approx mean var}, where we again set $X$ to be a balanced mixture of $N(\mu\lo 1, (r\lo 1\udex {|i-j|}))$ and $N(-\mu\lo 1, (r\lo 2\udex {|i-j|}))$, with $(\mu\lo 1\trans, r\lo 1, r\lo 2)$ being $((0, 2), .4, .8)$, $((0, .8), .5,-.5)$, and $((0,0), .4, .8)$, respectively. The resulting $\var (X\lo 1 | X\lo 2)$ approximates to a piecewise constant function and a trigonometric function in the first two cases. In the third case, $\pi\lo i(X\lo 2)$ is degenerate, and, by (\ref{eq: var xbx mixNormal}), $\var (X\lo 1 | X\lo 2)$ is exactly a quadratic function. As mentioned below Lemma \ref{lemma: linear gen}, the overall linearity condition (\ref{assume: linearity}) is satisfied for $E(X\lo 1| X\lo 2)$ in this case. A relative discussion about the inverse regression methods under linear $E(X|\beta\lo 0\trans X)$ and quadratic $\var (X | \beta\lo 0\trans X)$, but with elliptically distributed $X$, can be found in \cite{luo2018}.

Compared with the general parametric families of $E(X | \beta\lo 0\trans X)$ discussed in \cite{li2009} and \cite{dong2010}, the proposed parametric form (\ref{eq: mean xbx mixNormal}) has a clear interpretation; that is, it is most suitable to use if the sample distribution of $X$ conveys a clustered pattern. In addition, because we model $E(X | \beta\lo 0\trans X)$ by averaging a few linear functions of $\beta\lo 0\trans X$, the result is relatively robust to the outliers of $X$. The functional form of $\var (X | \beta\lo 0\trans X)$ is also relaxed in our work, whereas it is fixed as a constant in \cite{dong2010}.

Recall that the MAVE-type methods adopt the linearity condition in the local neighborhoods of $X$, and that each mixture component in the mixture model (\ref{assume: mNormal x}) can be regarded as a global neighborhood of $X$. Since we adopt the linearity condition and the constant variance condition for each mixture component of $X$, a corresponding modification of the inverse regression methods will naturally serve as a bridge, or a compromise, that connects inverse regression with the MAVE-type methods. This modification will also inherit the advantages of both sides, including the $n\udex {1/2}$-consistency and the exhaustive recovery of the central SDR subspace, etc., as seen next.

\section{Adjusting SIR for $X$ under mixture model}\label{sec:sir}

\def\mo{\Omega\lo {1}}
\def\so{{\mathcal S}\lo 1}
\def\sro{{\mathcal S}\lo {1}^*}
\def\hmo{\widehat {\Omega}\lo {1}}
\def\hso{\widehat {\mathcal S}\lo 1}
\def\hsro{\widehat {\mathcal S}\lo {r1}}
\def\hssir{\widehat {\mathcal S}\lo {\mathrm {SIR}}}
\def\hsrsir{\widehat {\mathcal S}\lo {\mathrm {rSIR}}}
\def\Mrsir{{\Omega}\lo {\mathrm {rSIR}}}
\def\hmsir{\widehat {\Omega}\lo {\mathrm {SIR}}}
\def\hmrsir{\widehat {\Omega}\lo {\mathrm {rSIR}}}
\def\tr{{\mathrm {tr}}}

Based on the discussions in Section 2, there are two intuitive strategies for adjusting inverse regression for $X$ that follows a mixture model. First, one can conduct inverse regression on each mixture component of $X$ using (\ref{eq: mean xpi|x mnoromal}) and (\ref{eq: var xpi|x mNormal}), and then merge the SDR results. Second, one can construct and solve an appropriate estimation equation (\ref{eq: ee general}) with the aid of the functional forms (\ref{eq: mean xbx mixNormal}) for $E(X|\beta\lo 0\trans X)$ and (\ref{eq: var xbx mixNormal}) for $\var(X|\beta\lo 0\trans X)$. Because these two strategies use the terms $\pi \lo i (X)$ and $\pi \lo i (\beta\lo 0\trans X)$, respectively, which characterize the ``global" neighborhoods of $X$ and $\beta\lo 0\trans X$, they can be parallelized with MAVE and its refined version \citep{xia2002}. We start with adjusting SIR in this section, and discuss both strategies sequentially.

To formulate the first strategy, we first parametrize the SDR result specified for an individual mixture component of $X$ as the conditional central subspace on $W=i$, which is the unique subspace of $\real\udex p$ that satisfies
\begin{align}\label{assume: cs i}
Y \indep X \mid (\beta\trans X, W = i),
\end{align}
with minimal dimension. Denote this space by $\cs\udex {(i)}$. For a general non-latent variable $W$, the space spanned by the union of $\cs\udex {(i)}$'s, denoted by $\spn {\bigcup\lo {i=1}\udex q \cs\udex {(i)}}$, is called the partial central subspace in \cite{chiaromonte2002sufficient}. Naturally, the first strategy aims to estimate $\spn {\bigcup\lo {i=1}\udex q \cs\udex {(i)}}$, so it recovers $\cs$ if 
\begin{align}\label{eq: cs i ens = cs}
\cs = {\mathcal {S}} \left(\bigcup\lo {i=1}\udex q \cs\udex {(i)} \right).
\end{align}
Because $W$ is the latent variable constructed for the convenience of modeling the distribution of $X$, we can naturally assume that $W$ is uninformative to $Y$ given $X$. Intuitively, this means that the term $W=i$ can be removed from (\ref{assume: cs i}), by which $\cs \udex {(i)}$ becomes identical to $\cs$. The only issue is that the support of $X|W=i$ can be a proper subset of the support of $X$, under which $\cs \udex {(i)}$ may only capture a local pattern of $Y|X$ and becomes a proper subspace of $\cs$. Nonetheless, since the supports of $X|W=i$ for $i=1,\ldots, q$ must together cover the support of $X$, the union of $\cs \udex {(i)}$'s must span the entire $\cs$. This reasoning is rigorized in the following lemma, which justifies (\ref{eq: cs i ens = cs}) and thus justifies the consistency of the first strategy. 

\begin{lemma}\label{lemma: cs i ens = cs}
Suppose $Y\indep W | X$. Then the coincidence between $\cs$ and $\spn{\bigcup\lo {i=1}\udex q \cs\udex {(i)}}$ in (\ref{eq: cs i ens = cs}) always holds. In addition, if the support of $X|W=i$ coincides with the support of $X$ for some $i=1,\ldots, q$, then $\cs$ further coincides with $\cs\udex {(i)}$ specified for the $i$th mixture component of $X$.
\end{lemma}

\proof
Since $Y\indep W | X$, $\cs$ clearly satisfies (\ref{assume: cs i}) for each $i=1,\ldots, q$. Since $\cs\udex {(i)}$ is by definition the intersection of all the subspaces of $\real\udex p$ that satisfy (\ref{assume: cs i}) \citep{chiaromonte2002sufficient}, it must be a subspace of $\cs$. Thus, we have
\begin{align}\label{prf: cs i in cs}
{\mathcal {S}} \left(\bigcup\lo {i=1}\udex q \cs\udex {(i)} \right) \subseteq \cs.
\end{align}
Conversely, let $\gamma$ be an arbitrary basis matrix of $\spn {\bigcup\lo {i=1}\udex q \cs\udex {(i)}}$. By definition, $\gamma$ satisfies (\ref{assume: cs i}) for each $i=1,\ldots, q$, which means that, for any $g(Y) \in L\lo 2 (Y)$,
\begin{align}\label{prf: cs in cs i ens}
E\{g(Y) | X, W\} = E\{g(Y) | \gamma\trans X, W\}.
\end{align}
Since $Y\indep W | X$, we have $E\{g(Y) | X, W\} = E\{g(Y) | X\}$. Thus, (\ref{prf: cs in cs i ens}) implies 
\begin{align}\label{prf: cs in cs i ens 2}
E\{g(Y) | X\} = E\{g(Y) | \gamma\trans X, W\}, 
\end{align}
which means that $E\{g(Y) | \gamma\trans X, W\}$ is invariant as $W$ varies and thus equal to $E\{g(Y) | \gamma\trans X\}$. Hence, (\ref{prf: cs in cs i ens 2}) implies $E\{g(Y) | X\} = E\{g(Y) | \gamma\trans X)$, which, by the arbitrariness of $g(\cdot)$, means $Y\indep X |\gamma\trans X$. By the definition of $\cs$, we have $\cs \subseteq \spn {\gamma}$. Together with (\ref{prf: cs i in cs}), we have (\ref{eq: cs i ens = cs}). Suppose the support of $X|W= i$ is identical to the support of $X$. Then, for any $g(Y) \in L\lo 2 (Y)$, $E\{g(Y)|X,W=i\}$ coincides with $E\{g(Y)|X\}$ in terms of the same functional form and the same domain. The proof of $\cs = \cs\udex {(i)}$ thus follows the same as above.
\eop

The coincidence between the overall support of $X$ and that specified for an individual mixture component occurs, for example, if $X$ has a mixture multivariate normal distribution. However, such coincidence should be often unreliable in practice, as the clustered pattern of $X$ determines that each of its individual mixture components is likely clouded within a specific region of the support of $X$. Thus, we recommend using the ensemble $\spn {\bigcup\lo {i=1}\udex q \cs\udex {(i)}}$ as a conservative choice rather than using an individual $\cs\udex {(i)}$, although the former means more parameters to estimate. 

To recover $\spn {\bigcup\lo{i=1}\udex q \cs\udex {(i)}}$ by conducting SIR within the mixture components of $X$, we first formulate the slicing technic in SIR mentioned in the Introduction by replacing $Y$ in (\ref{eq: M}) with $\vec Y\lo S \equiv (I(Y\lo S=1),\ldots, I(Y\lo S = H))$, which generates 
\begin{align}\label{eq: 1st order M}
	\Sigma\lo {X}\inv E\{X \vec Y\lo S\}
\end{align}
that again spans a subspace of $\cs$ under the overall linearity condition (\ref{assume: linearity}). In light of (\ref{eq: mean xpi|x mnoromal}) that reformulates the mixture component-wise linearity condition (\ref{assume: linearity Wi}), we incorporate the grouping information for the $i$th mixture component, i.e. $\pi\lo i (X)$, into (\ref{eq: 1st order M}) and merge the results across all the mixture components. These deliver 
\begin{align}\label{eq: modimm first order}
	\mo \equiv E [\{\Sigma\lo i\inv (X - \mu\lo i) \pi\lo i (X) \vec Y\lo S\}\lo {i\in\{1,\ldots, q\}}],
\end{align}
whose linear span $\spn \mo$ is employed to recover $\cs$. Because $\spn \mo$ is constructed under the mixture model (\ref{assume: mNormal x}), we call any method that recovers $\spn \mo$ $\msir$, M in the subscript for the mixture model. As seen in the next theorem, the consistency of $\msir$ in the population level is immediately implied by Lemma \ref{lemma: linear gen}, under the mixture component-wise linearity condition (\ref{assume: linearity Wi}). 

\begin{theorem}\label{thm: so}
Suppose $X$ follows a mixture model that satisfies (\ref{assume: linearity Wi}). Then $\spn {\mo}$ is always a subspace of $\cs$. In addition, if, for any $v\! \in \!\cs$, there exists $i\in \{1,\ldots, q\}$ and $h\in \{1,\ldots, H\}$ such that $v \in \spn {\Sigma\lo i\inv E(X - \mu\lo i|W =i, Y\lo S= h )}$, then $\spn {\mo}$ is identical to $\cs$.
\end{theorem}

\begin{proof}
For each $i=1,\ldots, q$, we have
\begin{align*}
E\{(X - \mu\lo i ) \pi\lo i (X) \vec Y\lo S\} & = E[E\{(X - \mu\lo i ) \pi\lo i (X) | \beta\lo 0\trans X, \vec Y\lo S\} \vec Y\lo S] \\
& = E[E\{(X - \mu\lo i ) \pi \lo i (X) | \beta\lo 0\trans X\} \vec Y\lo S],
\end{align*}
which, by (\ref{eq: mean xpi|x mnoromal}), means $\Sigma\lo i\inv E\{(X - \mu\lo i ) \pi\lo i (X) \vec Y\lo S\} \in \cs$. This implies $\spn {\mo} \subseteq \cs$. The second statement of the theorem is straightforward, so its proof is omitted.
\end{proof}

Referring to the generality of the mixture model discussed in Section 2, we regard $\msir$ as building a path that connects SIR and ADR. In one extreme, it coincides with SIR if $X$ has only one mixture component or more generally if both $\beta\lo 0\trans \Sigma\lo i \beta\lo 0$ and $\beta\lo 0\trans \mu\lo i$ are invariant across different mixture components of $X$; the coincidence in the latter case is readily implied by (\ref{eq: mean xpi|x mnoromal}), under which $\mo$ is a list of $q$ identical copies of (\ref{eq: 1st order M}). In the other extreme, if we let $q$ diverge with $n$, then the mixture components will resemble the local neighborhoods of $X$ and thus $\msir$ will essentially resemble ADR. Generally, depending on the underlying distribution of $X$, $\msir$ will provide a data-driven balance between the estimation efficiency and the reliability of SDR results. 

Similarly to the exhaustiveness of ADR, an important advantage of $\msir$ over SIR is that its exhaustiveness is more general, as justified in Theorem \ref{thm: so}. In particular, it can detect symmetric associations between $Y$ and $X$, which cannot be captured by SIR. To see this, suppose $X$ follows a symmetric mixture model with respect to the origin, and
\begin{align}\label{model: quadratic}
	Y = X\lo 1\udex 2 + \epsilon
\end{align}
where $\epsilon$ is a random error. Then $\msir$ can clearly detect $X\lo 1$ as long as $X\lo 1 | Y$ is asymmetric with respect to zero for at least one mixture component, which occurs if one of $\mu\lo i$'s has nonzero first entry. Figure~\ref{figure: model} illustrates this with $X\lo 1$ being a balanced mixture of $N(-2, 1)$ and $N(2, 1)$. In general, the asymmetry of $\beta\lo 0\trans X | Y$ is satisfied for at least one mixture component of $X$ if $\alpha\lo i$'s differ along the directions in $\cs$, in which case the multiplicative weight $\pi\lo i (X)$ in $\mo$ disturbs any symmetric pattern between $\beta\lo 0\trans X$ and $Y$ into asymmetry. Since $\alpha\lo i$'s are regulated to differ from each other, with a sufficiently large $q$ and assuming a continuous prior distribution for $\spn {\beta\lo 0}$ \citep{hall1993}, the exhaustiveness of $\msir$ is satisfied with probability one in a Bayesian sense.

\begin{figure}[!htbp]
\begin{center}
\includegraphics[scale=.35]{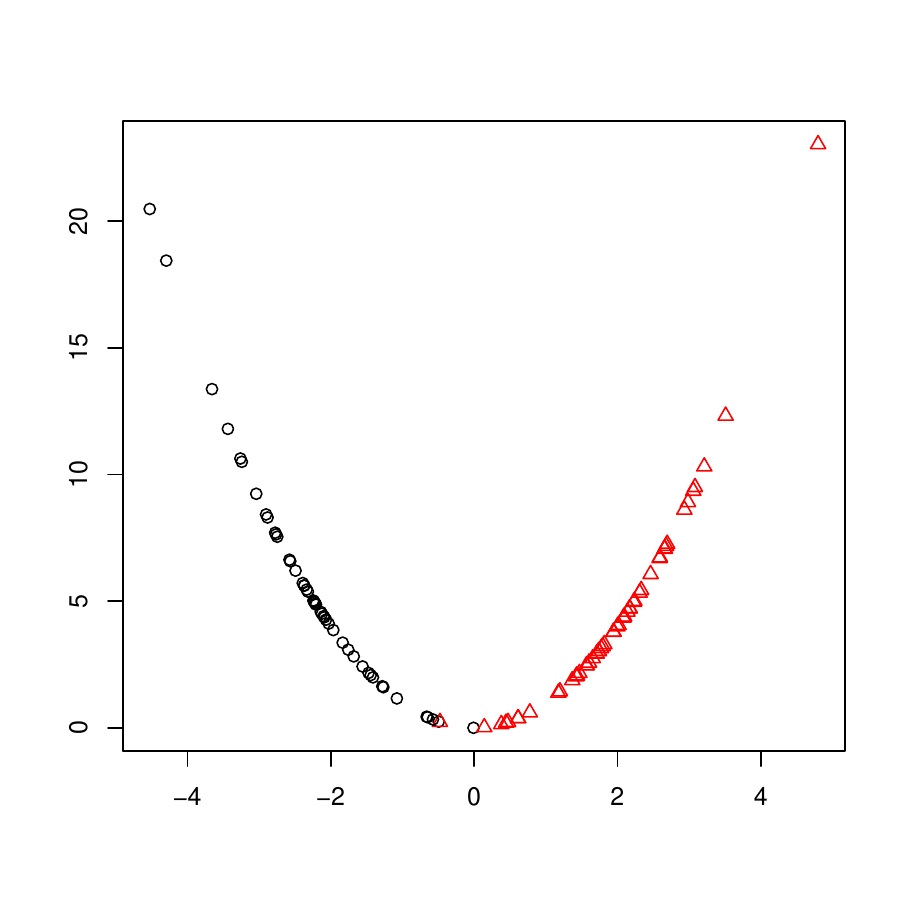}
\caption{\small The scatter plot of $(X\lo 1, X\lo 1\udex 2)$ with $X\lo 1 \sim .5N(-2, 1) + .5 N(2, 1)$ and $n=100$: ``$\circ$" for the first mixture component, and ``$\triangle$" for the second mixture component.}
\label{figure: model}
\end{center}
\end{figure}


Related to this, another advantage of $\msir$ over SIR is that it allows more flexibility in choosing the number of slices $H$. When implementing SIR in practice, researchers need to address the tradeoff in selecting $H$: while a small $H$ means better estimation accuracy due to less number of parameters, it elevates the risk of non-exhaustive recovery of $\cs$, as SIR can only recover up to $H-1$ directions. By contrast, $\msir$ can recover up to $q(H-1)$ directions, relieving the freedom of choosing a coarse slicing to enhance the estimation accuracy. This advantage is crucial when $Y$ is discrete or categorical with a small number of categories, in which case $H$ must also be small and the issue of non-exhaustiveness is less worrisome for $\msir$. The same advantage is also enjoyed by ADR, i.e. the localized SIR, which sets $H=2$ automatically. Again, this complies with the connection between $\msir$ and ADR discussed below Theorem \ref{thm: so}.  


To implement $\msir$ in the sample level, we can replace $\mu\lo i$, $\Sigma\lo i$, and $\pi\lo i (X)$ in $\mo$ with $\widehat \mu\lo i$, $\widehat \Sigma\lo i$, and $\widehat \pi\lo i (X)$ mentioned in Section~\ref{sec:mixnorm}, and replace the population mean in $\mo$ with the sample mean. The resulting $\hmo$ is $n\udex {1/2}$-consistent and asymptotically normal. By \cite{luo2016ladle}, this permits using the ladle estimator on $\hmo$ to determine the dimension $d$ of $\cs$, given the identity between $\spn \mo$ and $\cs$ discussed above. When the sample size is limited compared with the dimension $p$ or the number of mixture components $q$, we also recommend using the predictor augmentation estimator \citep[PAE;][]{luo2021pae} to determine $d$, which does not require the asymptotic normality of $\hmo$. As discussed above Lemma \ref{lemma: q hat q} in Section 2, these estimators permit us to safely assume $d$ to be known throughout the theoretical study. Given $\hmo$, we use its first $d$ left singular vectors to span an $n\udex {1/2}$-consistent estimator of $\cs$. The proofs of these results are simple, so they are omitted. The details about the implementation of $\msir$ is deferred to Appendix B.

When the dimension of $X$ is large, the estimation bias in the clustering stage, especially in estimating $\Sigma\lo i\inv$, will adversely affect $\msir$. One way to address this issue is to impose additional constraints on the distribution of $X$; see Section~\ref{sec:sirHD} later. Alternatively, observe that $\beta\lo 0\trans X$ follows a mixture model in $\real\udex d$. The low dimensionality of $\pi\lo i (\beta\lo 0\trans X)$ makes its estimation more accurate than that of $\pi\lo i (X)$, i.e. regardless of the dimension of $X$, given a consistent initial estimator of $\beta\lo 0$. As $\pi\lo i (\beta\lo 0\trans X)$ can be rewritten as $E\{\pi\lo i (X)|\beta\lo 0\trans X\}$, the enhanced estimation power can also be understood as benefited from targeting at the conditional mean of $\pi\lo i (X)$, which is intuitively smoother than $\pi\lo i (X)$ itself. In addition, (\ref{eq: mean xbx mixNormal}) implies that $E(X | \beta\lo 0\trans X)$ can be estimated by linearly regressing $X$ on $\{\pi\lo i (\beta\lo 0\trans X), \pi\lo i (\beta\lo 0\trans X) \beta\lo 0\trans X\}\lo {i\in\{1,\ldots, q\}}$, which does not involve $\Sigma\lo i$'s and thus is robust to the dimension of $X$.

Based on these observations, we now turn to the second strategy of estimating $\cs$ mentioned in the beginning of this section, which serves as a refinement of $\msir$. Namely, we use (\ref{eq: mean xbx mixNormal}) to construct an appropriate estimating equation (\ref{eq: ee general}) for $\cs$. The consistency of this approach is endorsed by \cite{ma2012}, as $E(\alpha(X) | \beta\trans X)$ in (\ref{eq: ee general}) needs to be truly specified only when $\spn \beta$ coincides with $\spn {\beta\lo 0}$.

Assume that $\pi\lo i (\beta\lo 0\trans X)$ is known a priori, where in practice we can replace $\spn {\beta\lo 0}$ with a consistent estimator such as $\msir$. Let
\begin{align}\label{eq: RX}
	R\lo X(\spn \beta) \equiv \tsum\lo {i=1}\udex q \pi\lo i (\beta\lo 0\trans X) \{P\trans(\Sigma\lo i, \beta) (X - \mu\lo i ) + \mu\lo i \}.
\end{align}
By (\ref{eq: mean xbx mixNormal}), $R\lo X(\cs)$ coincides with $E(X| \beta\lo 0\trans X)$. We explain later why using $\beta\lo 0$ instead of the argument $\beta$ in $\pi\lo i (\beta\lo 0\trans X)$. To measure the relative ``importance" of an outcome to the SDR estimation, we introduce
{\small {
\begin{align}\label{eq: D by Stein's lemma}
	D (\beta\lo 0\trans X) \equiv \tsum\lo {i=1}\udex q \tsum\lo {j=1}\udex q (\beta\lo 0\trans \Sigma\lo j \beta\lo 0)\inv E\{ \beta\lo 0 \trans (X - \mu\lo j) \vec Y\lo S \pi\lo i (\beta\lo 0 \trans X) \pi\lo j (\beta\lo 0\trans X) \} \pi\lo i(\beta\lo 0\trans X).
\end{align} }}

\noindent
This term is best interpretable when $\beta\lo 0\trans X $ follows a mixture normal distribution: by Stein's Lemma, it can rewritten as a weighted average of the derivative of the regression function $E(\vec Y\lo S | \beta\lo 0\trans X)$ with respect to $\beta\lo 0\trans X$, i.e.
\begin{align}\label{eq: D original}
	D(\beta\lo 0\trans X) = \tsum\lo {i=1}\udex q E[\partial \{E(\vec Y\lo S | \beta\lo 0\trans X) \pi\lo i(\beta\lo 0\trans X)\}  / \partial \beta\lo 0 \trans X ] \pi\lo i(\beta\lo 0\trans X),
\end{align}
which reveals the strength of the point-wise effect of $\beta\lo 0\trans X$ on regressing $\vec Y\lo S$.

We now plug $R\lo X(\spn \beta)$ and $D(\beta\lo 0\trans X)$ into (\ref{eq: ee general}), with $\mu \lo g (\beta\trans X)$ set at zero as permitted by the double-robust property. To avoid obscures caused by dimensionality \citep{ma2013}, we rewrite (\ref{eq: ee general}) as minimizing the magnitude of its left-hand side over all the $d$-dimensional subspaces of $\real\udex p$; that is, we minimize
\begin{align}\label{eq: obj first order}
	\go (\spn \beta) \equiv \left\| E [\{X - R\lo X (\spn \beta)\} \vecc\trans \{D(\beta\lo 0\trans X) \diag (\vec Y\lo S)\} ] \right\|\lo F\udex 2,
\end{align}
where $\|\cdot\|\lo F$ denotes the Frobenius norm of a matrix and $\diag (v)$ transforms any vector $v$ to the diagonal matrix with $v$ being the diagonal. Same as for $\pi\lo i(\beta\lo 0\trans X)$, $\cs$ in $D(\beta\lo 0\trans X)$ will be replaced with $\msir$, the explanation for using $D(\beta\lo 0\trans X)$ instead of $D(\beta\trans X)$ deferred to later. Clearly, the minimum value of $\go (\cdot)$ is zero, and can be reached by $\cs$. Following \cite{ma2012}, under fairly general conditions, the intersection of all the minimizers of $\go (\cdot)$ also minimizes $\go (\cdot)$, which means that it is the unique minimizer of $\go (\cdot)$ of the smallest dimension, and it is always a subspace of $\cs$ under the mixture component-wise linearity condition (\ref{assume: linearity Wi}). We call any method that recovers this space $\rmsir$, R in the subscript for refined.

Similar to $\msir$, the exhaustiveness of $\rmsir$ depends on the variability of the association between $Y$ and $\beta\lo 0\trans X$ across different mixture components of $\beta\lo 0\trans X$. In the special case that the $\Sigma\lo i$'s are identical to each other, the following theorem justifies a sufficient condition for the exhaustiveness of $\rmsir$.

\begin{theorem}\label{thm: sro}
Suppose $X$ follows a mixture model that has identical $\Sigma\lo i$'s and satisfies (\ref{assume: linearity Wi}). If $[E \{\tsum\lo {i=1}\udex q \pi\lo i (\beta\lo 0\trans X) \pi\lo j (\beta\lo 0\trans X) \beta\lo 0\trans (X - \mu\lo i) \vec Y\lo S\}]\lo {j\in\{1,\ldots, q\}}$ has rank $d$, then $\rmsir$ exhaustively recovers $\cs$.
\end{theorem}

\begin{proof}
For any $\beta$ such that $\spn \beta$ is a proper subspace of $\cs$, suppose $\spn \beta$ is $k$-dimensional, $\beta$ is of full column rank, $\beta\lo0$ is a basis of $\cs$, and, without loss of generality, $(\beta, \gamma) = \beta\lo 0$ for some semi-orthogonal $\gamma \in \real\udex {p\times (d-k)}$. We have $P(\Sigma\lo i, \beta)\gamma = \beta (\beta\trans \Sigma\lo i \beta)\inv (\beta\trans \Sigma\lo i \gamma)$ for $i\in\{1,\ldots, q\}$, which, by (\ref{eq: RX}), means
\begin{align*}
	\gamma\trans\{X - R\lo X(\spn\beta)\} & = \tsum\lo i\udex q \pi\lo i (\beta\lo 0\trans X) \{\gamma\trans - (\gamma\trans \Sigma\lo i \beta) (\beta\trans \Sigma\lo i \beta)\inv \beta\trans\} (X - \mu\lo i) \\
	& = \tsum\lo i \udex q \pi\lo i (\beta\lo 0\trans X) (0, A\lo i) (\beta\lo 0\trans \Sigma\lo i \beta\lo 0)\inv \beta\lo 0\trans (X - \mu\lo i)
\end{align*}
where $A\lo i = \gamma\trans \Sigma\lo i \gamma - (\gamma\trans \Sigma\lo i \beta) (\beta\trans \Sigma\lo i \beta) (\beta\trans \Sigma\lo i \gamma)$. Since $\Sigma\lo i$ is invariant of $i$, as is $A\lo i$. Denote the former by $\Sigma$ and the latter by $A$, and write $\go (\spn \beta)$ as $\|\Psi(\spn \beta)\|\lo F\udex 2$. We then have
\begin{align*}
	\gamma\trans \Psi (\spn\beta) & = E[\{\tsum\lo {i=1} \udex q \pi\lo i (\beta\lo 0\trans X) (0, A ) (\beta\lo 0\trans \Sigma\lo i \beta\lo 0)\inv \beta\lo 0\trans (X - \mu\lo i) \vec Y\lo S\} D\trans (\beta\lo 0\trans X)  ] \\
	& = (0,A) \tsum\lo {j=1}\udex q E\udex {\otimes 2}[\tsum\lo {i=1}\udex q \pi\lo i (\beta\lo 0\trans X) \pi\lo j (\beta\lo 0\trans X) (\beta\lo 0\trans \Sigma\lo i \beta\lo 0)\inv \beta\lo 0\trans (X - \mu\lo i) \vec Y\lo S] \\
	& \equiv (0, A) \tsum\lo {j=1}\udex q M\lo j\udex {\otimes 2}.
\end{align*}
Because $\Sigma$ is invertible, $A$ must be nonzero. Hence, if $\tsum\lo {j=1}\udex q M\lo j\udex {\otimes 2} = (M\lo j)\lo {j\in\{1,\ldots, q\}}\udex {\otimes 2}$ is invertible, which is the condition adopted in this theorem, then $\gamma\trans \Psi (\spn\beta)$ is always nonzero, and, as $\go (\spn \beta) \geq \|\gamma\gamma\trans \Psi (\spn \beta)\|\lo F\udex 2$, $\go (\spn \beta)$ must also be nonzero for any $\beta$ that spans a proper subspace of $\cs$. Thus, $\rmsir$ exhaustively recovers $\cs$. This completes the proof.
\end{proof}

The exhaustiveness condition in Theorem \ref{thm: sro} essentially requires that, for any direction $v$ in $\cs$, the association between $v\trans X$ and $Y$ is asymmetric in at least one mixture component of $\beta\lo 0\trans X$. Again, this will be likely to hold in a Bayesian sense if the number of mixture components $q$ is sufficiently large in the mixture model of $X$, and it will be trivially true if we let $q$ diverge with $n$, e.g. in the scenario of approximating non-clustered distribution of $X$ by the mixture model. Same as $\msir$, depending on the complexity of $W$, $\rmsir$ builds a path from SIR that adopts the aggressive parametric assumption on $X$ to the fully nonparametric SDR methods, i.e. of MAVE-type, that require the weakest assumption on $X$. In one extreme that $W$ is degenerate, $D(\beta\lo 0\trans X)$ will also be degenerate and $R\lo X (\spn \beta)$ will reduce to a linear $E(X | \beta\trans X)$, which together imply the identity between $\rmsir$ and SIR. In the other extreme that $W$ is identical to $\beta\lo 0\trans X$, the integrand in (\ref{eq: obj first order}) would resemble the efficient score for homoscedastic data \citep{luo2014}, had $\beta\lo 0$ in the integrand been replaced with $\beta$ and $E\{\vec Y\lo S | \beta\trans X\}$ been consistently estimated instead of being misspecified at zero.

The concern behind the use of $\beta\lo 0$ instead of the argument $\beta$ in $\pi\lo i (\beta\lo 0\trans X)$ and $D (\beta\lo 0\trans X)$ is as follows. If we replace $\beta\lo 0$ with the argument $\beta$ in both $\pi\lo i (\beta\lo 0\trans X)$ in $R\lo X (\spn \beta)$ and $D(\beta\lo 0\trans X)$, then the unique minimizer of (\ref{eq: obj first order}) of smallest dimension, if exists, will still fall in $\cs$. However, the corresponding exhaustiveness will require stronger assumptions. For example, this space will be trivially zero-dimensional in Model (\ref{model: quadratic}), whereas $\rmsir$ still recovers the one-dimensional $\cs$. For this reason, we decide not to do so when constructing $\go (\cdot)$.

To implement $\rmsir$, we assume its exhaustiveness as well as the exhaustiveness of $\msir$, and we start with estimating $\pi\lo i (\beta\lo 0\trans X)$ by rewriting this term as $E(\pi\lo i(X) | \beta\lo 0\trans X)$ and running a kernel regression of $\widehat \pi\lo i(X)$ on $\widetilde \beta\trans X$, where $\spn {\widetilde \beta}$ denotes the result of $\msir$. By plugging $\widehat \mu\lo i$'s, $\widehat \Sigma\lo i$'s, and the estimates of $\pi\lo i (\beta\lo 0\trans X)$'s into (\ref{eq: obj first order}), we have $\hgo (\cdot)$ that consistently estimates $\go (\cdot)$. $\rmsir$ is then the unique minimizer of $\hgo (\cdot)$, which is $n\udex {1/2}$-consistent and asymptotically normal.
Referring to the discussions above, it can outperform $\msir$ in the sample level, which is analogous to the relative effectiveness of the refined MAVE with respect to MAVE \citep{xia2002}. For continuity of the context, we leave the detailed form of $\hgo (\cdot)$ and the algorithm for its minimization to Appendix B.


As mentioned in the Introduction, when $X$ follows a mixture skew-elliptical distribution, \cite{guan2017} proposed the generalized Stein's Lemma based method (StI) to consistently recover $\cs$. Due to the violation of the mixture component-wise linearity condition (\ref{assume: linearity Wi}), both $\msir$ and $\rmsir$ are theoretically inconsistent under \citeauthor{guan2017}'s settings. Nonetheless, while StI requires identical $\Sigma\lo i$'s up to multiplicative scalars and uses the average of SDR results across different mixture components, we allow $\Sigma\lo i$'s to vary freely and we collect the information for SDR from all the mixture components to form the final result. Therefore, our approaches are sometimes exclusively useful and are more likely to recover $\cs$ exhaustively. As seen next, our approaches can also be extended to the high-dimensional settings, as well as to the second-order inverse regression methods.

\section{An extension towards the high-dimensional sparse settings}\label{sec:sirHD}

\def\hgssir{\widehat G\lo {\mathrm {SIR}}\udex S}
\def\hgrssir{\widehat G\lo {\mathrm {rSIR}}\udex S}
\def\hbsirsp{\widehat {\beta}\lo {\mathrm {SIR}}\udex {{\mathrm S}}}
\def\hssirsp{\widehat {\mathcal S}\lo {\mathrm {S-SIR}}}
\def\hsrsirsp{\widehat {\mathcal S}\lo {\mathrm {rSIR}}\udex {\mathrm S}}
\def\hbsiror{\widehat {\beta}\lo {\mathrm {SIR}}\udex {{\mathrm O}}}


In the literature, SIR has been widely studied under the high-dimensional settings, with the aid of the sparsity assumption on $\cs$ that only a few rows of $\beta\lo 0$ are nonzero. Representative works include \cite{chen2010}, \cite{li2007sparse}, \cite{lin2019}, and \cite{tan2018}, etc. A common strategy employed in these works is to reformulate SIR as a least squares method and then incorporate penalty functions. Accordingly, the overall linearity condition (\ref{assume: linearity}) is commonly adopted, which, as mentioned in the Introduction, can be restrictive in practice due to the sparsity assumption.

We now modify $\msir$ and $\rmsir$ towards sparsity under the high-dimensional settings. For the applicability of existing clustering analysis, here we follow \cite{cai2019} to assume that $X$ has a mixture multivariate normal distribution with the number of mixture components $q$ known {\it a priori} and all $\Sigma\lo i$'s identical to some invertible $\Sigma$, and that
\begin{align}\label{assume: sparse mixnorm}
	\Sigma\inv (\mu\lo i - \mu\lo j) \mbox{ is sparse for all } i \neq j,
\end{align}
under which the consistent estimators $\widehat \mu\lo i$'s, $\widehat \Sigma$, and $\widehat \pi\lo i (X)$'s have been derived in \cite{cai2019}. Given the mixture model fit, we modify $\msir$ and $\rmsir$ following the spirit of \cite{tan2018} for its computational efficiency.

\def\fM{{\mathcal {M}}}

Following \cite{tan2018}, we change the parametrization in SDR to $\Pi = \beta\udex {\otimes 2}$ and further extend the parameter space to the set of $p$-dimensional positive semi-definite matrices, denoted by $\fM$. Let
\begin{align}\label{eq: SIR1 HD}
	\M\lo {E(X|Y)} \equiv  \left[[E\{(X - \mu\lo i) \pi\lo i (X) \vec Y\lo S\}]\lo {i\in\{1,\ldots, q\}}\right] \udex {\otimes 2},
\end{align}
and $\widehat \M\lo {E(X|Y)}$ be its estimator using $\widehat \mu\lo i$'s and $\widehat \pi\lo i (X)$'s. We minimize
\begin{align}\label{eq: sparse SIR HD}
	- \tr(\widehat \M\lo {E(X|Y)} \Pi) + \rho \|\Pi\|\lo 1\mbox{ subject to } \tr(\widehat \Sigma \udex {1/2}  \Pi \widehat \Sigma\udex {1/2}) \leq d, \|\widehat \Sigma \udex {1/2}  \Pi \widehat \Sigma\udex {1/2}\|\lo {\mathrm {sp}} \leq 1,
\end{align}
over $\Pi \in \fM$, where $\rho$ is a tuning parameter, $\tr(\cdot)$ and $\|\cdot\|\lo {\mathrm {sp}}$ denote the trace and the spectral norm of a matrix, respectively, and $\|A\|\lo 1 = \tsum\lo {i,j} |a\lo {ij}|$ for any matrix $A  = (a\lo {ij}) \in \fM$.

Because (\ref{eq: sparse SIR HD}) is a convex minimization problem, it has the unique minimizer, which we denote by $\widehat \Pi$. By simple algebra, $\widehat \Pi$ must be on the boundary of the constraints in (\ref{eq: sparse SIR HD}), and must have the form $\widehat \beta\udex {\otimes 2}$ for some $\widehat \beta \in \real\udex {p\times d}$ with $\widehat \Sigma\udex {1/2} \widehat \beta$ being semi-orthogonal. We estimate $\cs$ by $\spn {\widehat \beta}$, and call this estimator the sparse $\msir$ or simply $\smsir$ due to the natural bond between $\M\lo {E(X|Y)}$ and $\mo$ defined in (\ref{eq: modimm first order}). Its consistency is justified in the following theorem, where we use 
\begin{align}\label{eq: distance S hat S}
\delta(\widehat \beta, \beta\lo 0) \equiv \| P(\I\lo p, \widehat \beta) - P(\I\lo p, \beta\lo 0) \|\lo F
\end{align}
to measure the distance between $\spn {\widehat \beta}$ and $\cs$.

\begin{theorem}\label{thm: sparse sir}
Let $s\lo 1=\max\lo {i,j\in\{1,\ldots, q\}}\|\Sigma\inv (\mu\lo i - \mu\lo j)\|\lo \infty$ under the sparsity assumption (\ref{assume: sparse mixnorm}), let $s\lo 2$ be the number of nonzero rows of $\beta\lo 0$ under the sparse SDR assumption, and let $\lambda\lo d$ be the $d$th largest eigenvalue of $\Sigma\udex {-1/2} \M\lo {E(X|Y)} \Sigma\udex {-1/2}$. Under the regularity conditions (C1)-(C3) (see Appendix A), if $d\udex 2 s\lo 1 s\lo 2 \log p = o(n)$, $s\lo 1 s\lo 2\udex 2 \log p = o(n)$, $d s\lo 1 \log p =o ( n \rho \lambda\lo d)$, and $\rho s\lo 2 = o( \lambda\lo d)$, then we have
\begin{align}\label{eq: sparse SIR HD consistency}
	D(\widehat \beta, \beta\lo 0) = \bop (d s\lo 1 \log p  / ( n \rho \lambda\lo d ) + \rho s\lo 2 / \lambda\lo d + d (s\lo 1 s\lo 2 \log p / n)\udex {1/2} ).
\end{align}
\end{theorem}

\begin{proof}
Theorem 3 can be proved by incorporating the theoretical results of \cite{cai2019} into the proof of Theorem 1 in \cite{tan2018}. Hence, we only provide a sketch of the proof here for the case $q=2$. Let $\gamma = \Sigma\inv (\mu\lo 1 - \mu\lo 2)$. Under Condition (C1) in Appendix A of this article and Lemma 3.1, Lemma 3.2, and Theorem 3.1 in \cite{cai2019}, we have
\begin{align*}
\max(| \widehat \pi\lo 1 - \pi\lo 1 |, \|\widehat \mu\lo 1 - \mu\lo 1\|\lo 2, \|\widehat \mu\lo 2 - \mu\lo 2\|\lo 2, \|\widehat \gamma  - \gamma\|\lo 2) = \bop ((s\lo 1 \log p / n)\udex {1/2}).
\end{align*}
Together with the sub-Gaussian condition of $X$ implied by its mixture multivariate normal distribution and Conditions (C2)-(C3) in Appendix A of this article, it is easy to modify the proofs of Lemma 1 and Proposition 1 in \cite{tan2018} to derive
\begin{align*}
\|\widehat \M\lo {E(X|Y)}\! - \!\M\lo {E(X|Y)}\|\lo \max \!=\! \bop ((s\lo 1\! \log p / n)\udex {1/2}) \mbox{ and } \|\widehat \Sigma \!-\! \Sigma\|\lo \max\! = \!\bop ((s\lo 1\! \log p / n)\udex {1/2}),
\end{align*}
where $\|A\|\lo \max = \max\{|a\lo {ij}|:i,j\in\{1,\ldots, p\}\}$ for any matrix $A = (a\lo {ij}) \in \real\udex {p\times p}$. Similarly, in our scenario, Lemma S1 in the supplementary material for \cite{tan2018} can be modified to
\begin{align}\label{prf: hd}
\|\tilde \Lambda - \Lambda\|\lo {\mathrm {sp}} \leq C (s\lo 1 s\lo 2 \log p / n)\udex {1/2} \mbox{ and } \|\tilde \Pi - \Pi\|\lo F \leq Cd (s\lo 1 s\lo 2 \log p /n)\udex {1/2}
\end{align}
where $\Lambda$ is defined in Equation (4) of \cite{tan2018}, $\Pi = \beta\lo 0\udex {\otimes 2}$, and the definitions of $\tilde \Lambda$ and $\tilde \Pi$ can be found above Lemma S1 in the supplementary material for \cite{tan2018}; Lemma S5 in the supplementary material for \cite{tan2018} can be modified to that the term $\{s \log (ep) / n\}\udex 2$ be changed to $\{s\lo 1 s\lo 2 \log (ep) / n\}\udex {1/2}$; Lemma S6 in the supplementary material for \cite{tan2018} can be modified to that the convergence rate in its conclusion is $C (s\lo 1 \log p /n)\udex {1/2}$; the rest of the lemmas, including Lemma S2, Lemma S3, and Lemma S4, will still hold. Consequently, in the proof of \cite{tan2018}'s Theorem 1 (see their supplementary material), the condition $s\lo 1 s\lo 2\udex 2 \log p = o(n)$ is needed to ensure (S25), and the term $\gamma$ (which is a different term from $\gamma=\Sigma\inv (\mu\lo 1 -\mu\lo 2)$ defined above) is of order $\bop ((d s\lo 1 s\lo 2 \log p / n)\udex {1/2})$. A modification of equation (S29) in their proof then yields
\begin{align}\label{prf: hd 2}
\|\Delta\|\lo F  & = \bop (\gamma / \lambda\lo d + \gamma\udex 2 / (\rho s\lo 2 \lambda\lo d ) + \rho s\lo 2 / \lambda\lo d )
= \bop (\gamma\udex 2 / (\rho s\lo 2 \lambda\lo d ) + \rho s\lo 2 / \lambda\lo d ) \nonumber \\
& = \bop (d s\lo 1 \log p  / ( n \rho \lambda\lo d ) + \rho s\lo 2 / \lambda\lo d ),
\end{align}
where the definition of $\Delta$ can be found below Lemma S1 in the supplementary material for \cite{tan2018}. By the subsequent logic flow of their proof, (\ref{prf: hd}) and (\ref{prf: hd 2}) together imply (\ref{eq: sparse SIR HD consistency}). This completes the proof.
\end{proof}

In Theorem \ref{thm: sparse sir}, we allow divergence of both the number of variables of $X$ that are uniquely informative to $Y$ and the dimension of $\cs$. (\ref{eq: sparse SIR HD consistency}) suggests that the optimal value of $\rho$ is of order $(d s\lo 1 \log p)\udex {1/2} / (s\lo 2 n)\udex {1/2}$. In practice, we follow \cite{tan2018} to suggest tuning $\rho$ by cross validation, details omitted.

To modify $\rmsir$ towards sparsity, we follow Theorem \ref{thm: sro} to assume the equality of $\Sigma\lo i$'s, under which $\rmsir$ delivers the $d$ leading eigenvectors of $\Sigma\inv \Mrsir \Sigma\inv$ in the population level with
\begin{align}\label{eq: SIR2 HD}
	\Mrsir \equiv E\udex {\otimes 2}\left[\{\tsum\lo {i=1}\udex q \pi\lo i(\beta\lo 0\trans X) (X - \mu \lo i) \} \vecc\trans \{D(\beta\lo 0\trans X) \diag (\vec Y\lo S) \} \right].
\end{align}
Let $\hmrsir$ be a consistent estimator of $\Mrsir$ using $\widehat \mu\lo i$'s to replace $\mu\lo i$'s and using $\smsir$ to replace $\cs$. Similar to (\ref{eq: sparse SIR HD}), we minimize
\begin{align}\label{eq: sparse rSIR HD}
	- \tr(\hmrsir \Pi) + \rho \|\Pi\|\lo 1 \mbox{ subject to } \tr (\widehat \Sigma \udex {1/2}  \Pi \widehat \Sigma\udex {1/2}) \leq d, \|\widehat \Sigma \udex {1/2}  \Pi \widehat \Sigma\udex {1/2}\|\lo {\mathrm {sp}} \leq 1,
\end{align}
over $\Pi \in \fM$, which again has the unique minimizer of the form $\widehat \beta\lo R\udex {\otimes 2}$ for some $\widehat \beta\lo R \in \real\udex {p\times d}$ that satisfies $\widehat \beta\lo R\trans \widehat \Sigma \widehat \beta\lo R = \I\lo d$. We call $\spn {\widehat \beta\lo R}$ the sparse $\rmsir$ or simply $\srmsir$. Following a similar reasoning to Theorem \ref{thm: sparse sir}, $\srmsir$ converges to $\cs$ at the same rate as $\smsir$. Because it does not involve $\pi\lo i (X)$, it can be more robust than $\smsir$ against the estimation bias in fitting the mixture model of $X$, caused by large $p$ or the non-normality of the mixture components, etc.

As mentioned above (\ref{assume: sparse mixnorm}), the implementation of both $\smsir$ and $\srmsir$ requires $q$ to be truly specified {\it a priori} in the clustering stage. Because there is lack of consistent estimators of $q$ under the high-dimensional settings in the existing literature, $q$ can be potentially underestimated in practice, causing inconsistency of $\smsir$ and $\srmsir$. Nonetheless, as these methods reduce to the sparse SIR proposed in \cite{tan2018} when $q$ is forced to be one, which is the worst case scenario for estimating $q$, they are still expected to outperform the existing sparse SIR if $q$ is underestimated to be some $\widehat q > 1$. To determine $d$ under the high-dimensional settings, which coincides with the rank of both $\Omega\lo {E(X|Y)}$ and $\Mrsir$, we recommend applying the aforementioned PAE \citep{luo2021pae} to either $\widehat \Omega\lo {E(X|Y)}$ or $\hmrsir$, details omitted.

\section{Adjusting SAVE for $X$ under mixture model}\label{sec:save}

\def\Msave{\Omega\lo {2}}
\def\Mdr{{\Omega}\lo {\mathrm {DR}}}
\def\gsave{{\mathrm G}\lo 2}
\def\bsave{\beta\lo {\mathrm {SAVE}}}
\def\st{{\mathcal S}\lo {2}}
\def\bsave{\beta\lo {\mathrm {rSAVE}}}
\def\hmsave{\widehat {\Omega}\lo {\mathrm {SAVE}}}
\def\hgsave{\widehat {\mathrm G}\lo {\mathrm {SAVE}}}
\def\m{\Omega}

We now parallelize the two proposed strategies above to adjust SAVE for $X$ under the mixture model (\ref{assume: mNormal x}). The results can be developed similarly for the other second-order inverse regression methods, e.g. pHd and directional regression; see Appendix A for detail.

The first strategy amounts to conducting SAVE within each mixture component of $X$ and then taking the ensemble of the results. For $h\in\{1,\ldots, H\}$, let $Y\lo {S,h}$ be the $h$th entry of $\vec Y\lo S$, i.e. $Y\lo {S, h}$ is one if and only if $Y\lo S$ equals $h$. For $i\in\{1,\ldots, q\}$, we introduce
\begin{align*}
	E \{\pi\lo i (X) Y\lo {S, h}\} \I\lo p  - \Sigma\lo i\inv \left[E\{(X - \mu\lo i )\udex {\otimes 2} \pi\lo i (X)  Y\lo {S, h}\} - H E\udex {\otimes 2}\{(X - \mu\lo i) \pi\lo i (X) Y\lo {S, h}\} \right],
\end{align*}
denoted by $\m\lo {h, i}$, whose ensemble over $h$ mimics SAVE for the $i$th mixture component of $X$. We recover $\cs$ by the column space of $\Msave \equiv (\m\lo {h, i})\lo {h\in\{1,\ldots, H\}, i\in\{1,\ldots,q\}}$, and call the corresponding SDR method $\msave$. When $X$ has only one mixture component, $\msave$ reduces to SAVE. Its general consistency is justified in the following theorem.

\begin{theorem}\label{thm: save}
Suppose $X$ follows a mixture model that satisfies (\ref{assume: linearity Wi}) and (\ref{assume: const var Wi}). Then $\Msave$ always spans a subspace of $\cs$. In addition, $\Msave$ spans $\cs$ if, for any $v \in \cs$, $E\{(v\trans X, v\trans X\udex {\otimes 2} v)\pi\lo i(X) Y\lo {S, h}\}$ varies with $h$ for at least one of $i\in\{1,\ldots, q\}$.
\end{theorem}

\begin{proof}
For each $h\in\{1,\ldots, H\}$, we can rewrite $\m\lo {h, i}$ as
\begin{align*}
	\m\lo {h, i} & = [E \{\pi\lo i (X) | Y\lo {S, h} = 1\}/H ] \I\lo p  -\\
	& \hspace{.4cm} \Sigma\lo i\inv [E\{(X - \mu\lo i )\udex {\otimes 2} \pi\lo i (X) | Y\lo {S, h} = 1\}  - E\udex {\otimes 2}\{(X - \mu\lo i) \pi\lo i (X) | Y\lo {S, h} = 1\} ]/H \\
	& = [ E \{\pi\lo i (X) | Y\lo {S, h} = 1\} \I\lo p - \Sigma\lo i\inv \var \{(X-\mu\lo i) I(W=i) | Y\lo {S, h}=1\}]/H \\
	& = E \{ \pi\lo i (X) \I\lo p - \Sigma\lo i\inv \var\lo i (X  |\beta\lo 0\trans X) | Y\lo {S, h} = 1\} / H \\
 & \hspace{.4cm} - \Sigma\lo i \inv \var \{ E\lo i (X| \beta\lo 0\trans X) | Y\lo {S, h}=1\} / H.
\end{align*}
Thus, by (\ref{eq: mean xpi|x mnoromal}) and (\ref{eq: var xpi|x mNormal}), we have $\m \lo {h, i} = \beta\lo 0 A\lo {h, i}$ for some $A\lo {h, i} \in \real\udex {d\times p}$, which implies $\spn {\Msave} \subseteq \cs$. The condition for the coincidence between $\spn {\Msave}$ and $\cs$ can be derived straightforwardly following the proofs of Theorem 3 and 4 in \cite{li2007}.
\end{proof}

By Theorem \ref{thm: save}, the exhaustiveness of $\msave$ requires each informative direction of $X$ to be captured by SAVE in at least one mixture component of $X$, so it is quite general \citep{li2007}. The implementation of $\msave$ involves replacing the parameters of the mixture model in $\Msave$ with the corresponding estimators in Section 2 and replacing the population moments with the sample moments. The resulting $\widehat \M\lo 2$ is $n\udex {1/2}$-consistent and asymptotically normal, by which the ladle estimator \citep{luo2016ladle} can be applied again to determine $d$. The leading $d$ left singular vectors of $\widehat \M\lo 2$ then span a $n\udex {1/2}$-consistent estimator of $\cs$. These results are straightforward, so we omit the details.

Under framework (\ref{eq: ee general}), SAVE can be reformulated by taking
\begin{align*} 
	g(Y,\beta\trans X) = \phi (\beta\trans X) \left[ E (X\udex {\otimes 2}\otimes \vec Y\lo S | \beta\trans X) - \{\var(X | \beta\trans X) + E\udex {\otimes 2}(X | \beta\trans X) \} \otimes \vec Y\lo {S} \right]
\end{align*} 
with $\phi(\beta\trans X)$ set at a constant for simplicity and $\alpha(X) = E\{g(Y, \beta\trans X) | X\}$, and $\mu\lo g (\beta\trans X)$ set at zero because $E\{g(Y, \beta\lo 0\trans X) | \beta\lo 0\trans X\}$ is zero, and $\mu\lo \alpha (\beta\trans X)$ also at zero due to the double-robust property. To adjust for the mixture model of $X$, we apply the functional forms (\ref{eq: mean xbx mixNormal}) for $E(X|\beta \trans X)$ and (\ref{eq: var xbx mixNormal}) for $\var(X|\beta \trans X)$, and we set $\phi(\beta\trans X)$ at $\pi\lo i\udex 2 (\beta\trans X)$ for each $i\in\{1,\ldots, q\}$ and then merge the estimating equations. For simplicity, one may also set $\phi(\beta\trans X)$ at a constant as in SAVE. The resulting objective function is 
\begin{align*}
\gsave (\spn \beta) \equiv \tsum\lo {i=1}\udex q \tsum\lo {h=1}\udex H & \left\| E \{\pi\lo i\udex 2 (\beta\trans X) X\udex {\otimes 2} Y\lo {S, h}\} - E\{\pi\lo i\udex 2 (\beta\trans X) \mu\lo 2 (\beta\trans X) Y\lo {S, h}\} \right. \\
&  \hspace{.2cm} \left. - \tsum\lo {j=1}\udex q E\{\pi\lo i\udex 2 (\beta\trans X) \pi\lo j(\beta\trans X) Y\lo {S, h}\} \Sigma\lo j Q (\Sigma\lo j, \beta) \right\|\lo F\udex 2,
\end{align*}
where $\mu\lo 2 (\beta\trans X)$ is defined in Lemma \ref{lemma: con var gen}. Same as $\go (\cdot)$, there also exists the unique minimizer of $\gsave(\cdot)$ with the smallest dimension, which is always a subspace of $\cs$ under the mixture component-wise linearity and constant variance conditions (\ref{assume: linearity Wi}) and (\ref{assume: const var Wi}). The exhaustiveness of this space roughly follows that of SAVE, although $\pi\lo i(\beta\trans X)$ is used instead of $\pi\lo i(\beta\lo 0\trans X)$ in $\gsave (\cdot)$.

Let $\widehat \G\lo 2(\cdot)$ be the estimator of $\gsave(\cdot)$ using model fitting results in Section 2. We call the unique $d$-dimensional minimizer of $\widehat \G\lo 2(\cdot)$ the refined $\msave$ or simply $\rmsave$, whose $n\udex {1/2}$-consistency can be easily proved. The implementation of $\rmsave$ is deferred to Appendix B. Again, the advantage of $\rmsave$ over $\msave$ comes from that the mixture components are embedded in the low dimensional $\beta\trans X$ rather than the ambient $X$, much comparable with the advantage of the refined MAVE over MAVE.

\section{Simulation studies}

We now use simulation models to evaluate the effectiveness of the proposed methods. We will start with generating $X$ by some simple mixture multivariate normal models, under which we examine the performance of the proposed $\msir$ and $\rmsir$. To better comprehend the applicability of $\msir$ and $\rmsir$, we will then assess their robustness under more complex mixture model of $X$ and under the violation of the mixture model assumption itself. The adjusted sparse SIR, i.e. $\smsir$ and $\srmsir$, will be evaluated afterwards under the high-dimensional settings where $p \geq n$, and the adjustments of SAVE will be assessed finally. The section is divided into four subsections in this order. A complementary simulation study is presented in Appendix C, where we evaluate how the estimation error in fitting the mixture model and how a hypothetical misspecification of $q$ impact the proposed SDR results. 


\subsection{Adjusted SIR for mixture normal $X$}

We first evaluate the performance of the proposed $\msir$ and $\rmsir$ in comparison of SIR and the MAVE-type methods, i.e. MAVE, ADR, and eCVE, in the case that $X$ follows a mixture multivariate normal distribution. The \texttt{R} packages \texttt{meanMAVE} and \texttt{CVarE} are used to implement MAVE and eCVE, respectively, where the tuning parameters are automatically selected. To implement ADR, we  follow the suggestion in \cite{wang2020} to set $H=2$ and characterize the local neighborhoods of $X$ by the $k$-nearest neighbors with $k = 4p$. To fit the mixture model in implementing $\msir$ and $\rmsir$, we use the maximal likelihood estimation, with $q$ determined by BIC, both available from the \texttt{R} package \texttt{mixtools}. For SIR, $\msir$, and $\rmsir$, we set $H=5$ uniformly unless $Y$ is discrete. 

We apply these SDR methods to the following four models, where $\varepsilon$ is an independent error with standard normal distribution. In the first three models, $X$ is a balanced mixture of two multivariate normal distributions, i.e. $\pi\lo 1 N(\mu\lo 1, \Sigma\lo 1) + \pi\lo 2 N(\mu\lo 2, \Sigma\lo 2)$ with $\pi\lo 1 = \pi\lo 2 = .5$, and we set $\mu\lo 2 = -\mu\lo 1$ and $\Sigma\lo i = [r\lo i]\lo p$ for some scalar $r\lo i$ for $i=1,2$. Here, $[r\lo i]\lo p$ denotes the $p$-dimensional square matrix whose $(i,j)$th entry is $r\lo i\udex {|i-j|}$, and we allow $r\lo 1 \neq r\lo 2$. The case of unbalanced $X$ is studied in Model $4$ and furthermore in Subsection 6.2 later. 
\begin{align*}
&\mbox{Model 1:} \quad  Y = \sign{X\lo 1+1}\,\log(|X\lo 2-2+\varepsilon|),  \\
&\qquad\qquad\quad \mu\lo 1=(1.5, 1.5, 1.5,0,\cdots,0)\trans, (r\lo 1, r\lo 2) = (.5, .5), \\
&\mbox{Model 2:} \quad  Y = \sign{X\lo 1}\,(5-X\lo1\udex 2+X\lo 2\udex 2)+ .5 \varepsilon,  \\
&\qquad\qquad\quad \mu\lo 1 =  (.8,.8,.8,0,\ldots, 0)\trans, (r\lo1,r\lo2)=(.5,-.5), \\
&\mbox{Model 3:} \quad  Y = (X\lo 1 +X\lo2)\udex 2+ \exp(X\lo3+X\lo4)\varepsilon,  \\
&\qquad\qquad\quad \mu\lo 1 = (2,2, \cdots,2)\trans/p\udex {1/2}, (r\lo1,r\lo2)=(.5,-.5),\\
&\mbox{Model 4:} \quad  Y \sim \mbox{Bernoulli}[\{1+\exp(2X\lo1+ 2X\lo2+4)\}\inv], \\
&\qquad\qquad\quad  X \sim .25 N(\mu\lo 1, [.3]\lo p) + .5 N(0, I\lo p) + .25 N(\mu\lo 3, [-.3]\lo p),  \\
&\qquad\qquad\quad \mu\lo 1=(2,2,0,\cdots,0)\trans, \mu\lo 3=(-2,2,0,\cdots,0)\trans.
\end{align*}
In Models $1-3$, $Y$ is continuous and the central subspace is two-dimensional. The monotonicity or asymmetry between $\beta\lo 0\trans X$ and $Y$ varies in these models: while both $X\lo 1$ and $X\lo 2$ have monotone effects on $Y$ in Model 1, the monotonicity only applies to $X\lo 1$ in Model 2, and neither effects are asymmetric in Model 3. Consequently, these models are in favor of SIR from the most to the least, if we temporarily ignore the impact from the distribution of $X$. By contrast, the asymmetry between $\beta\lo 0\trans X$ and $Y$ becomes stronger when specified for the individual mixture components in all the three models. The deviation between $\mu\lo i$'s also varies in these models, representing different degrees of separation between the mixture components of $X$. In Model 4, $X$ consists of three unbalanced mixture components, and $Y$ is binary with one-dimensional central subspace. The monotonicity between $\beta\lo 0\trans X$ and $Y$ in Model $4$ suggests that this model is in favor of SIR.

For each model, we set $(n,p)$ at $(500,10)$, and generate $200$ independent copies. To measure the accuracy of an SDR estimator $\spn {\widetilde \beta}$ for each model, we use the sample mean and sample standard deviation of $\delta(\widetilde \beta, \beta\lo 0)$ defined in (\ref{eq: distance S hat S}). The performance of the aforementioned SDR methods under this measure are recorded in Table \ref{tab: ss SIR}. For reference, if $\widetilde \beta$ is randomly generated from the uniform distribution on the unit hyper-sphere $\{\beta \in \real\udex {p\times d}: \beta\trans \beta =\I\lo d\}$, then $E\{\delta(\widetilde \beta, \beta\lo 0)\}$ is $\{2 d (p-1) /p\}\udex {1/2}$, which is $1.34$ for $d=1$ and $1.90$ for $d=2$ given $p=10$. 

From Table \ref{tab: ss SIR}, SIR fails to capture the central subspace in Models $1-3$. The MAVE-type methods have an uneven performance across the models. In particular, ADR fails to capture the local patterns of Model 3 where these patterns vary dramatically with the local neighborhoods of $X$, MAVE fails to capture the non-continuous effect of $X\lo 1$ in Model 1 and the effect of $X$ that falls out of the central mean subspace in Model 3, and eCVE shows a similar inconsistency to MAVE in Model 1. These comply with the intrinsic limitations of the corresponding SDR methods. By contrast, both $\msir$ and $\rmsir$ are consistent in all the models, and they mostly outperform SIR and the MAVE-type methods even when the latter are also consistent. Compared with $\msir$, $\rmsir$ is generally a slight improvement.

\begin{center}
\begin{table}[t]
\centering
\caption{Performance of SDR methods for mixture normal $X$}
\label{tab: ss SIR}
\begin{threeparttable}
\begin{tabular}[l]{l|cccc}
\hline \hline
\empty  & Model 1  & Model 2  & Model 3  & Model 4       \\ 
\hline
SIR      &.908(.136) &.910(.128) 
         &1.08(.179) &.634(.127) \\
ADR      &.611(.176) &.473(.138) 
        &1.31(.113)  &.645(.189) \\
MAVE     &.745(.229) &.223(.056) 
        &.895(.278)  &.326(.089) \\
eCVE    &1.42(.126)  &.452(.128) 
        &.609(.222)  &1.29(.226) \\
$\msir$  &.590(.151) &.436(.094) 
        &.553(.111)  &.374(.105) \\
$\rmsir$ &.555(.146) &.402(.085) 
        &.502(.109)  &.338(.098) \\
\hline\hline
\end{tabular}
\begin{tablenotes}
\footnotesize
\item In each cell of Columns 2-5, $a(b)$ is the sample mean (sample standard deviation) of $\delta(\widehat \beta, \beta\lo 0)$ for the corresponding SDR method, based on $200$ replications.
\end{tablenotes}
\end{threeparttable}
\end{table}
\end{center}

\subsection{Robustness of adjusted SIR}

We now evaluate the performance of $\msir$ and $\rmsir$ more carefully when the distribution of $X$ deviates from the simple balanced mixture multivariate normal distribution. These deviations include the case of unbalanced mixture components, the case of skewed distributions for the individual mixture components, and the case of non-clustered $X$ that violates the mixture model assumption itself. 

To generate an unbalanced mixture model for $X$, we set $(\pi\lo 1, \pi\lo 2)$ in Models $1-3$ in Subsection 6.1 above to each of $(.1, .9)$ and $(.3, .7)$, representing severe unbalance and mild unbalance between mixture components, respectively. The performance of the SDR methods are summarized in Table \ref{tab: different weight}, in the same format as Table \ref{tab: ss SIR}. Because similar phenomena to Table \ref{tab: ss SIR} can be observed, Table \ref{tab: different weight} again suggests the consistency of $\msir$ and $\rmsir$ when $X$ has a mixture multivariate normal distribution. However, when $\pi$ becomes smaller, the advantage of these methods over the existing SDR methods becomes less substantial, which can be anticipated in theory as a smaller $\pi$ can be regarded as a weaker clustered pattern of $X$. Note that, for each model, the overall strength of SDR pattern differs when we change the distribution of $X$. Thus, it is not meaningful to compare the performance of an individual SDR method in an individual model for different values of $(\pi\lo 1, \pi\lo 2)$, i.e. across Table \ref{tab: ss SIR} and Table \ref{tab: different weight}.

\begin{center}
\begin{table}[t]
\centering
\caption{Performance of SDR methods for unbalanced $X$}
\label{tab: different weight}
\begin{threeparttable}
\setlength{\tabcolsep}{3mm}{
\begin{tabular}[l]{ll|ccc}
\hline \hline
$\pi$ &Methods & Model 1  & Model 2  &Model 3 \\ 
\hline
\multirow{6}{*}{$.1$}
&SIR &.729(.100) &.692(.140) &.937(.149)\\
&ADR &.572(.152) &.485(.126) &1.15(.228)\\
&MAVE&.395(.109) &.211(.054) &.868(.283)\\
&eCVE&.503(.105) &.309(.055) &.702(.335)\\
&$\msir$  &.461(.098) &.541(.131)  
          &.458(.094) \\
&$\rmsir$ &.449(.106)   &.531(.132)  
          &.436(.100) \\
\hline
\multirow{6}{*}{$.3$}
&SIR &.713(.117) &.954(.255) &1.36(.127)\\
&ADR &.586(.142) &.451(.109) &1.23(.214)\\
&MAVE &.409(.102)&.218(.058) &.898(.263)\\
&eCVE &.841(.117)&.436(.045) &.594(.210)\\
&$\msir$ &.503(.099) &.528(.133)  
        &.481(.101)   \\
&$\rmsir$ &.490(.097) &.498(.126) 
        &.481(.098) \\
\hline\hline
\end{tabular}}
\begin{tablenotes}
\footnotesize
\item The meanings of numbers in each cell follow those in Table~\ref{tab: ss SIR}.
\end{tablenotes}
\end{threeparttable}
\end{table}
\end{center}

Next, we add skewness to the mixture components of $X$. As mentioned in Section 4, this case was studied in \cite{guan2017}, where StI was specifically proposed to recover $\cs$. Because both $\msir$ and $\rmsir$ are theoretically inconsistent in this case due to the violation of the mixture component-wise linearity condition (\ref{assume: linearity Wi}), we evaluate their sample-level effectiveness using SIR as a reference and using StI as the benchmark. For simplicity, we use three simulation models, i.e. Models (4.1), (4.3) and (4.5), in \cite{guan2017}:  
\begin{align*}
&\mbox{Model (4.1):} \quad  Y = \beta\lo 1\trans X +  \varepsilon, \\
&\mbox{Model (4.3):} \quad  Y = \sin(\beta\lo 1\trans X/3 + \varepsilon), \\
&\mbox{Model (4.5):} \quad  Y = (\beta\lo 1\trans X+ .3\varepsilon) \,/\,(2+|\beta\lo 2\trans X-4+\varepsilon|),
\end{align*}
where $\beta\lo1=(1,1,1,1,0, \ldots,0)\trans$, $\beta\lo2=(1,1,0,0,1,1,0\ldots,0)\trans$, and $\varepsilon$ is the same random error as in Subsection 6.1 above. For all the three models, $X$ is generated by $.5 SN\lo p(\mu\lo1,I\lo p, C\lo 1) + .3 SN\lo p(\mu\lo2,I\lo p, C\lo2) + .2 SN\lo p(\mu\lo3,I\lo p, C\lo3)$, with 
{\small {
\begin{align*}
& \mu\lo 1 = (3,3,3, 0,\ldots,0)\trans, \mu\lo 2 = (3,0,3,3,0,\ldots, 0)\trans, \mu\lo3 = (3,0,0,3,3,0,\ldots,0)\trans, \\
& C\lo 1 =\beta\lo 1, C\lo2 =3\beta\lo 1, C\lo 3 =(1,1, 0,0, 0,0, 1,1,0,0)\trans
\end{align*} }}

\noindent in the first two models, and with
{\small {
\begin{align*}
& \mu\lo 1 = (3,3,3,3, 0,\ldots,0)\trans, \mu\lo 2 = (3,0,0,3,3,0,\ldots, 0)\trans, \mu\lo 3 = (3,3,0,0,3,3,0, \ldots,0)\trans, \\
& C\lo 1 =2\beta\lo 1, C\lo2 =3\beta\lo 1, C\lo 3 =(-1,-1, 0,0, 0,0, -1,-1,0,0)\trans 
\end{align*} }}

\noindent
in the third model. Here, $SN(\mu,\Sigma,C)$ denotes a multivariate skew normal distribution, with $C$ being the skewness parameter whose magnitude indicates the severity of skewness; see \cite{guan2017} for more details. Roughly, the skewness is more severe for the mixture components of $X$ in Model (4.5) than in the other two models.   

Same as \cite{guan2017}, we set $(n,p)=(400,10)$. To permit a direct transfer of their simulation results, we also use $\delta\udex 2 (\widehat \beta, \beta\lo 0)$ instead of $\delta (\widehat \beta, \beta\lo 0)$ to measure the estimation error of a SDR method, which is a linear transformation of their working measure. The performance of all the SDR methods under $\delta\udex 2 (\widehat \beta, \beta\lo 0)$ is summarized in Table \ref{tab: guan}. From this table, the proposed $\msir$ and $\rmsir$ have a similar performance to StI in comparison of SIR for the first two models, and they are sub-optimal to StI but still substantially ourperform SIR for the third model, where again the skewness is more severe in the mixture components of $X$. Overall, these suggest that both $\msir$ and $\rmsir$ are robust against mild skewness in the mixture components of $X$ if $X$ follows a mixture model, and $\msir$ is slightly better than $\rmsir$ in this context.

\begin{center}
\begin{table}[t]
\centering
\caption{Performance of SDR methods for mixture skew normal $X$}
\label{tab: guan}
\begin{threeparttable}
\setlength{\tabcolsep}{5mm}{
\begin{tabular}[l]{l|cccc}
\hline \hline
Method & Model (4.1)  & Model (4.3)  & Model (4.5) \\ 
\hline
SIR      &.156(.041) &.219(.112) &1.20(.072) \\
$\msir$  &.026(.012) &.050(.046) &.816(.201) \\
$\rmsir$ &.062(.044) &.078(.050) &.826(.232) \\
StI      &.026(.010) &.038(.022) &.168(.084) \\
\hline\hline
\end{tabular}}
\begin{tablenotes}
\footnotesize
\item In each cell of Columns 2-4, $a(b)$ is the sample mean (sample standard deviation) of $\delta\udex 2(\widehat \beta, \beta\lo 0)$ for the corresponding SDR method, based on $200$ replications.
\end{tablenotes}
\end{threeparttable}
\end{table}
\end{center}

To evaluate the applicability of $\msir$ and $\rmsir$ for more complex data, we next generate $X$ such that it conveys a unimodal, continuous, but curved pattern. Namely, we first generate $(X\lo 1, X\lo 3,\ldots, X\lo p)$ under $N(0, [.3]\lo {p-1})$, and, given $X\lo 1$, we generate $X\lo 2$ in each of the following ways:
\begin{align}\label{eq: generate X2}
	X\lo 2 = |X\lo 1| + \varepsilon\lo 1, \quad X\lo 2 = \cos(2 X\lo 1) + \varepsilon\lo 1
\end{align}
where $\varepsilon\lo 1$ is an independent error distributed as $N(0, .1\udex 2)$. The sample support of $(X\lo 1, X\lo 2)$ conveys a $V$ shape and a $W$ shape, respectively. For each case, we generate $Y$ from Model 1, Model 2 and Model 4 in Subsection 6.1 above. The corresponding performance of SDR methods are summarized in Table \ref{tab: non-mixnormal}, again measured by $\delta (\widetilde \beta, \beta\lo 0)$ in (\ref{eq: distance S hat S}). The results for Model 3 are omitted because $X\lo 1 + X\lo 2$ is close to zero with a large probability under (\ref{eq: generate X2}), which causes an intrinsic difficulty to recover its effect for this model. The similar phenomena in Table \ref{tab: non-mixnormal} to Table \ref{tab: ss SIR} suggest the robustness of both $\msir$ and $\rmsir$ against the violation of the mixture model assumption on $X$, which means that these methods can have a wider application in practice. 

\begin{center}
\begin{table}[H]
\centering
\caption{Performance of SDR methods for V- and W-shaped $X$}
\label{tab: non-mixnormal}
\begin{threeparttable}
\begin{tabular}[l]{ll|ccc}
\hline \hline
\empty &Methods  & Model 1 & Model 2  & Model 4  \\
\hline
\multirow{6}{*}{V-shaped}
&SIR         &.278(.059)  &1.14(.220) &1.26(.099)\\
& ADR        &.654(.180)  &.641(.231) &1.01(.296) \\
& MAVE       &.087(.019)  &.064(.013) &1.08(.358)\\
&eCVE        &1.41(.007)  &1.42(.010) &1.33(.114) \\
& $\msir$    &.292(.093)  &.132(.112) &.357(.163)\\
& $\rmsir$   &.166(.109)  &.115(.130) &.328(.185)\\
\hline
\multirow{6}{*}{W-shaped}
&SIR         &.317(.090)  &.289(.067) &.196(.040) \\
& ADR        &.353(.096)  &.310(.067) &.442(.110) \\
& MAVE       &.076(.021)  &.061(.013) &.372(.092) \\
&eCVE        &1.36(.243)  &.328(.085) &.648(.395)  \\
& $\msir$    &.242(.102)  &.236(.154) &.502(.178)\\
& $\rmsir$   &.151(.091)  &.161(.079) &.495(.180)\\
\hline
\hline
\end{tabular}
\begin{tablenotes}
\footnotesize
\item The meanings of numbers in each cell follow those in Table~\ref{tab: ss SIR}.
\end{tablenotes}
\end{threeparttable}
\end{table}
\end{center}

\subsection{Adjusted high-dimensional sparse SIR}

We now compare the proposed $\smsir$ and $\srmsir$ with the sparse SIR proposed in \cite{tan2018} under the high-dimensional settings, where $n$ is set at $200$ and $p$ is set at each of $200$ and $300$. The ratio $p/n$ under the latter setting is the same as the simulation studies in \cite{tan2018}. The MAVE-type methods are omitted under these settings due to their theoretical inconsistency caused by localization. 

To generate a high-dimensional mixture multivariate normal $X$, we set $q=2$ with $\pi\lo 1 = \pi\lo 2 = .5$, and set both $\Sigma\lo {1}$ and $\Sigma\lo{2}$ to be $\Sigma = 1.5[-.3]\lo p$. To set $\mu\lo 1$ and $\mu\lo 2$, we follow \cite{cai2019} to introduce a discriminant vector $\gamma = (2u,\ldots,2u,0,\ldots,0)\trans$ with the first ten entries being nonzero, and let $\mu\lo{1}=(u,\ldots,u)\trans$ and $\mu\lo 2 = \mu\lo 1-\Sigma \gamma$. The magnitude of $u$ controls the degree of separation between the two mixture components, which we set to be $.5$, $.8$, and $1$ sequentially. Note that with $u=.5$, the mixture pattern is rather weak along any direction of $X$. Given $X$, we generate $Y$ from each of
\begin{alignat*}{2}
	&\mbox{Model A:} &\quad  &Y = \{(X\lo 1 + X\lo 2 + X\lo 3)/\sqrt{3}+1.5\}\udex2 + .5\varepsilon \\
	&\mbox{Model B:} &\quad  &Y = \exp\{(X\lo 1 + X\lo 2 + X\lo 3)/\sqrt{3}\} + \varepsilon
\end{alignat*}
where $\varepsilon$ is the same random error as in Subsection 6.1. Both models share the same one-dimensional $\cs$. In Model A, the effect of $X$ on $Y$ is only moderately asymmetric over all the observations, but it is more monotone if specified for each mixture component of $X$. Both effects are monotone in Model B.

Because $Y$ is continuous in both Model A and Model B, we still set the number of slices $H$ at five in implementing the sparse SIR and the proposed $\smsir$ and $\srmsir$. Again, as mentioned in Section 4, $q$ is assumed to be known when fitting the mixture multivariate normal distribution of $X$ in implementing the proposed methods. The accuracy of these SDR methods in estimating $\cs$ is summarized in Table~\ref{tab: ss SIR HD}, and their variable selection consistency, as measured by the true positive rate and the false positive rate, is summarized in Table~\ref{tab: ss SIR HD VS}, both based on $200$ independent runs.

\begin{center}
\begin{table}[t]
\caption{Performance of sparse SDR methods in estimating $\cs$ under HD settings}
\label{tab: ss SIR HD}
\centering
\begin{threeparttable}
\resizebox{\linewidth}{!}{
\begin{tabular}[l]{@{}clccccccc}
\hline \hline
\multirow{2}{*}{Model}  &\multirow{2}{*}{Method} & \multicolumn{3}{c}{$p=200, n=200$}&\multirow{2}{*}{} &\multicolumn{3}{c}{$p=300, n=200$}\\
\cline{3-5} \cline{7-9}& &$u=.5$  &$u=.8$ &$u=1$& &$u=.5$  &$u=.8$ &$u=1$\\
\hline
\multirow{3}{*}{A}
&SIR        &.536(.169)  &.641(.314)  &1.16(.039)  &
            &.712(.085)  &.805(.069)  &.992(.051)  \\
&$\smsir$   &.342(.098)  &.431(.087)  &.487(.068)  &
            &.411(.197)  &.686(.115)  &.709(.102)  \\
&$\srmsir$  &.236(.083)  &.401(.095)  &.420(.078)  &
            &.407(.122)  &.569(.091)  &.592(.103)  \\
\hline
\multirow{3}{*}{B}
&SIR        &.665(.142)  &.795(.115)  &.973(.427)  &
            &.692(.086)  &.846(.161)  &.972(.061)  \\
&$\smsir$   &.636(.357)  &.570(.125)  &.608(.220)  &
            &.620(.233)  &.818(.111)  &.832(.139)  \\
&$\srmsir$  &.554(.407)  &.569(.106)  &.604(.185)  &
            &.586(.138)  &.713(.104)  &.704(.124)  \\
\hline
\hline
\end{tabular}}
\begin{tablenotes}
\footnotesize
\item The meanings of numbers in each cell follow those in Table~\ref{tab: ss SIR}.
\end{tablenotes}
\end{threeparttable}
\end{table}
\end{center}

\begin{center}
\begin{table}[t]
\caption{Performance of sparse SDR methods in variable selection under HD settings}
\label{tab: ss SIR HD VS}
\centering
\begin{threeparttable}
\resizebox{\linewidth}{!}{
\begin{tabular}[l]{@{}clccccccc}
	\hline \hline
\multirow{2}{*}{Model}  &\multirow{2}{*}{Method} & \multicolumn{3}{c}{$p=200, n=200$}&\multirow{2}{*}{} &\multicolumn{3}{c}{$p=300, n=200$}\\
\cline{3-5} \cline{7-9}&  &$u=.5$  &$u=.8$ &$u=1$& &$u=.5$  &$u=.8$ &$u=1$\\
\hline
\multirow{3}{*}{A}
&SIR   &1.00(.004) &1.00(.017)  &1.00(.050) &
       &1.00(.032) &1.00(.023)  &1.00(.024)\\
&$\smsir$
       &1.00(.002) &1.00(.008)  &1.00(.012) &
       &1.00(.006) &1.00(.021)  &1.00(.023)\\
&$\srmsir$
       &1.00(.001) &1.00(.008) &1.00(.008) &
       &1.00(.004) &1.00(.011) &1.00(.013)\\
\hline
\multirow{3}{*}{B}
&SIR
       &1.00(.021) &.976(.405)  &.987(.031) &
       &1.00(.018) &1.00(.039)  &1.00(.043)\\
&$\smsir$
       &.933(.006) &.956(.198)  &1.00(.020) &
       &1.00(.011) &1.00(.021)  &1.00(.024)\\
&$\srmsir$
       &.940(.009) &.964(.186)  &1.00(.021) &
       &1.00(.007) &1.00(.021)  &1.00(.021) \\
\hline
\hline
\end{tabular}}
\begin{tablenotes}
\footnotesize
\item In each cell of Columns 3-8, $a(b)$ is the true(false) positive rate of the corresponding 
\item sparse SDR method, based on $200$ replications.
\end{tablenotes}
\end{threeparttable}
\end{table}
\end{center}

From these two tables, both $\smsir$ and $\srmsir$ are clearly more effective than the sparse SIR, except for Model B with $u=.5$ where they are comparable. In both models, the sparse SIR quickly fails to recover $\cs$ as $u$ increases (although still consistent in variable selection), whereas $\smsir$ and $\srmsir$ are much more robust against this change. The reason for the slightly compromised performance of $\smsir$ and $\srmsir$ when $u$ increases, is that a larger $u$ makes the second mixture component of $\beta\lo 0\trans X$ move towards the origin in Model A and leave the origin from the left in Model B, both leading to a vanishing signal of $Y|\beta\lo 0\trans X$ in this mixture component. When $u$ is $.5$, $\smsir$ and $\srmsir$ still slightly outperform the sparse SIR, so they are useful even when the distribution of $X$ is only weakly clustered.

We next change $X$ to follow a V-shaped distribution, by first generating $(X\lo 1, X\lo 3,\ldots, X\lo p)$ from $N(0, [-.3]\lo {p-1})$ and then generating $X\lo 2$ as in the first case of (\ref{eq: generate X2}) given $X\lo 1$. Same as above, we fix $q=2$ when approximating this distribution by a mixture multivariate normal distribution using \citeauthor{cai2019}'s method, for which a more complex non-clustered distribution such as W-shaped cannot be well approximated and thus is omitted. Based on $200$ independent runs, the performance of the sparse SDR methods for Model A and Model B is recorded in Table \ref{tab: ss SIR HD VW} and Table \ref{tab: ss SIR HD VS VW}, in the same format as Table \ref{tab: ss SIR HD} and Table \ref{tab: ss SIR HD VS}, respectively. Compared with Table \ref{tab: ss SIR HD} and Table \ref{tab: ss SIR HD VS}, a less substantial but similar phenomenon can be observed, suggesting the usefulness of both $\smsir$ and $\srmsir$ when $X$ has a curved and non-clustered distribution.

\begin{center}
\begin{table}[t]
\caption{Performance of sparse SDR methods in estimating $\cs$ for HD V-shaped $X$}
\label{tab: ss SIR HD VW}
\centering
\begin{threeparttable}
\begin{tabular}[l]{@{}clccc}
\hline \hline
Model  &Method & $p=200, n=200$ & &$p=300, n=200$\\
\hline
\multirow{3}{*}{model A}
&SIR      &.625(.079) &&.672(.068)\\
&$\smsir$ &.463(.078) &&.502(.068)\\
&$\srmsir$&.432(.089) &&.499(.087)\\
\hline
\multirow{3}{*}{model B}
&SIR      &.614(.110) &&.694(.088)\\
&$\smsir$ &.503(.108) &&.545(.086)\\
&$\srmsir$&.484(.106) &&.510(.102)\\
\hline
\hline
\end{tabular}
\begin{tablenotes}
\footnotesize
\item The meanings of numbers in each cell follow those in Table~\ref{tab: ss SIR}.
\end{tablenotes}
\end{threeparttable}
\end{table}
\end{center}

\begin{center}
\begin{table}[t]
\caption{Performance of sparse SDR methods in variable selection for HD V-shaped $X$}
\label{tab: ss SIR HD VS VW}
\centering
\begin{threeparttable}
\begin{tabular}[l]{@{}clccc}
\hline \hline
Model  &Method & $p=200, n=200$& &$p=300, n=200$\\
\hline
\multirow{3}{*}{model A}
&SIR      &1.00(.001)&&1.00(.001)\\
&$\smsir$ &1.00(.001)&&1.00(.001)\\
&$\srmsir$&1.00(.001)&&1.00(.001)\\
\hline
\multirow{3}{*}{model B}
&SIR      &.987(.001) &&1.00(.001)\\
&$\smsir$ &1.00(.004) &&1.00(.004)\\
&$\srmsir$&1.00(.004) &&1.00(.004)\\
\hline
\hline
\end{tabular}
\begin{tablenotes}
\footnotesize
\item The meanings of numbers in each cell follow those in Table~\ref{tab: ss SIR HD VS}.
\end{tablenotes}
\end{threeparttable}
\end{table}
\end{center}

\subsection{Adjusted SAVE under various settings}

To illustrate the performance of the proposed $\msave$ and $\rmsave$ in comparison of SAVE and the MAVE-type methods, we generate $Y|X$ by the following two models where $\varepsilon$ is the same random error as in Subsection 6.1.
\begin{align*}
&\mbox{Model \one:} \quad  Y = \sin(X\lo 1 + X\lo 2+2) + .5 \varepsilon, \\
&\mbox{Model \two:} \quad  Y = \sin(X\lo 1)+\sin(X\lo2) + .5 \varepsilon.
\end{align*}
Similarly to the studies for the adjusted SIR,
we first generate $X$ under the mixture model $.5 N(\mu\lo 1, [.3]\lo p) + .5 N(-\mu\lo 1, [r\lo 2]\lo p)$, where $\mu\lo 1 = (-1.5, -1.5,0,\cdots,0)\trans / p$ and we set $r\lo 2 =-.3$ for Model \one \ and $r\lo 2=.3$ for Model \two; we then generate $X$ under the V-shaped distribution and under the W-shaped distribution in the same way as in Subsection 6.2. Note that both $\msave$ and $\rmsave$ are theoretically consistent only in the first case.

Let $(n,p) = (500, 10)$. For each of SAVE, $\msave$, and $\rmsave$, we set $H=5$. The MAVE-type methods are implemented in the same way as in Subsection $6.1$. The results are summarized in Table~\ref{tab: ss SAVE}. Clearly, both $\msave$ and $\rmsave$ are effective even when $X$ does not convey a clustered pattern, and they outperform both SAVE and the MAVE-type methods in most cases.  

\begin{center}
\begin{table}[t]
\caption{Performance of SDR methods in comparison of SAVE}
\label{tab: ss SAVE}
\centering
\begin{threeparttable}
\begin{tabular}[l]{@{}lccccccc}
\hline \hline
\multirow{2}{*}{Method} & \multicolumn{3}{c}{Model $\one$}&\multirow{2}{*}{} &\multicolumn{3}{c}{Model $\two$}\\
\cline{2-4} \cline{6-8}& MMN  &V-shaped &W-shaped& & MMN  &V-shaped &W-shaped\\
\hline
SAVE      &1.40(.015) &1.06(.342) 
          &.267(.050) &
          &.515(.114) &.249(.071) 
          &.311(.092)\\
ADR       &.312(.071) &.398(.072) 
          &.275(.079) &
          &.620(.207) &.916(.314) 
          &.403(.096)\\
MAVE      &.871(.625) &.129(.031) 
          &.092(.028) &
          &.371(.286) &.944(.344) 
          &.408(.236)\\
eCVE      &.191(.208) &.353(.088) 
          &.729(.425) &
          &1.14(.414) &1.42(.009) 
          &1.38(.072) \\
$\msave$  &.347(.099) &.158(.055) 
          &.532(.187) &
          &.581(.166) &.103(.029) 
          &.241(.186)  \\
$\rmsave$ &.334(.094) &.143(.028) 
          &.514(.171) &
          &.575(.160) &.106(.032) 
          &.268(.129)\\
\hline
\hline
\end{tabular}
\begin{tablenotes}
\footnotesize
\item ``MMN" represents the case that $X$ is generated under the mixture multivariate normal distribution. The meanings of numbers in each cell follow those in Table~\ref{tab: ss SIR}.
\end{tablenotes}
\end{threeparttable}
\end{table}
\end{center}

\subsection{A real data example}

The MASSCOL2 data set, available at MINITAB (in catalogue STUDENT12), was collected in $56$ four-year colleges of Massachusetts in 1995, with the research interest on assessing how the academic strength of the incoming students of each college and the features of the college itself jointly affect the percentage of the freshman class that graduate. After removing missing values, $44$ observations are left in the data. Based on preliminary exploratory data analysis, we choose seven variables as the working predictor, including the percentage of the freshman class that were in the top $25\%$ of their high school graduating class (denoted by Top 25), the median SAT Math score (MSAT), the median SAT Verbal score (VSAT), the percentage of applicants accepted by each college (Accept), the percentage of enrolled students among all who were accepted (Enroll), the student/faculty ratio (SFRatio), and the out-of-state tuition (Tuition). The same predictor was also selected in \cite{li2009}, who applied SDR to an earlier version of the data set. The response variable, which is the percentage of the freshman class that graduate, is labeled as WhoGrad in the data set. We next analyze this data set using the proposed $\msir$ and $\rmsir$.

As suggested by BIC, the distribution of $X$ in this data set can be best approximated by the mixture multivariate normal distribution if we set the number of mixture components $q$ to be two or three. Here, we use three as a conservative choice. The dimension $d$ of $\cs$ for this data set is determined to be two. To choose $H$ for $\msir$, we measure the performance of $\msir$ by the predictability of the corresponding reduced predictor, for which we calculate the out-of-sample $R\udex 2$ based on the leave-one-out cross validation. In each training set, the regression function is fitted by kernel mean estimation using the \textsf{R} package {\tt{npreg}}, with the bandwidth selected automatically from the \textsf{R} package {\tt{npregbw}}. The resulting optimal choice of $H$ is two. Likewise, the optimal choice of $H$ for $\rmsir$ is three. We then apply $\msir$ and $\rmsir$ accordingly.

To illustrate the goodness of the mixture model fit for $X$ in this data set, we draw the scatter plots of the two-dimensional reduced predictors from $\msir$ and $\rmsir$, in the upper-left and upper-right panels of Figure \ref{figure: data}, respectively. The point type for each observation $x$ is determined by an independent multinomial random variable, with probability $(\widehat \pi\lo 1 (x),\widehat \pi\lo 2 (x), \widehat \pi\lo 3 (x))$ derived from the mixture model fit. The nature of clustering in these plots suggests the appropriateness of the mixture model and consequently the effectiveness of $\msir$ and $\rmsir$ for this data set. The latter is further supported by the aforementioned out-of-sample $R\udex 2$ for the reduced predictor from these methods, which is $.87$ for $\msir$ and is $.84$ for $\rmsir$. For reference, we also apply SIR, ADR, and eCVE to the data, and the out-of-sample $R\udex 2$ for their reduced predictors are $.79$, $.61$, and $.62$, respectively. Roughly, this means that the variance of the error in the reduced data derived by $\msir$ is $38\%$ less than that by SIR, $67\%$ less than that by ADR, and $66\%$ less than that by eCVE. The out-of-sample $R\udex 2$ for MAVE is $.86$, so MAVE performs equally well for this data set as the proposed SDR methods. 

\begin{center}
\begin{figure}[t]
\setlength{\abovecaptionskip}{0pt}
\setlength{\belowcaptionskip}{-10pt}
\includegraphics[width=.33\textwidth]{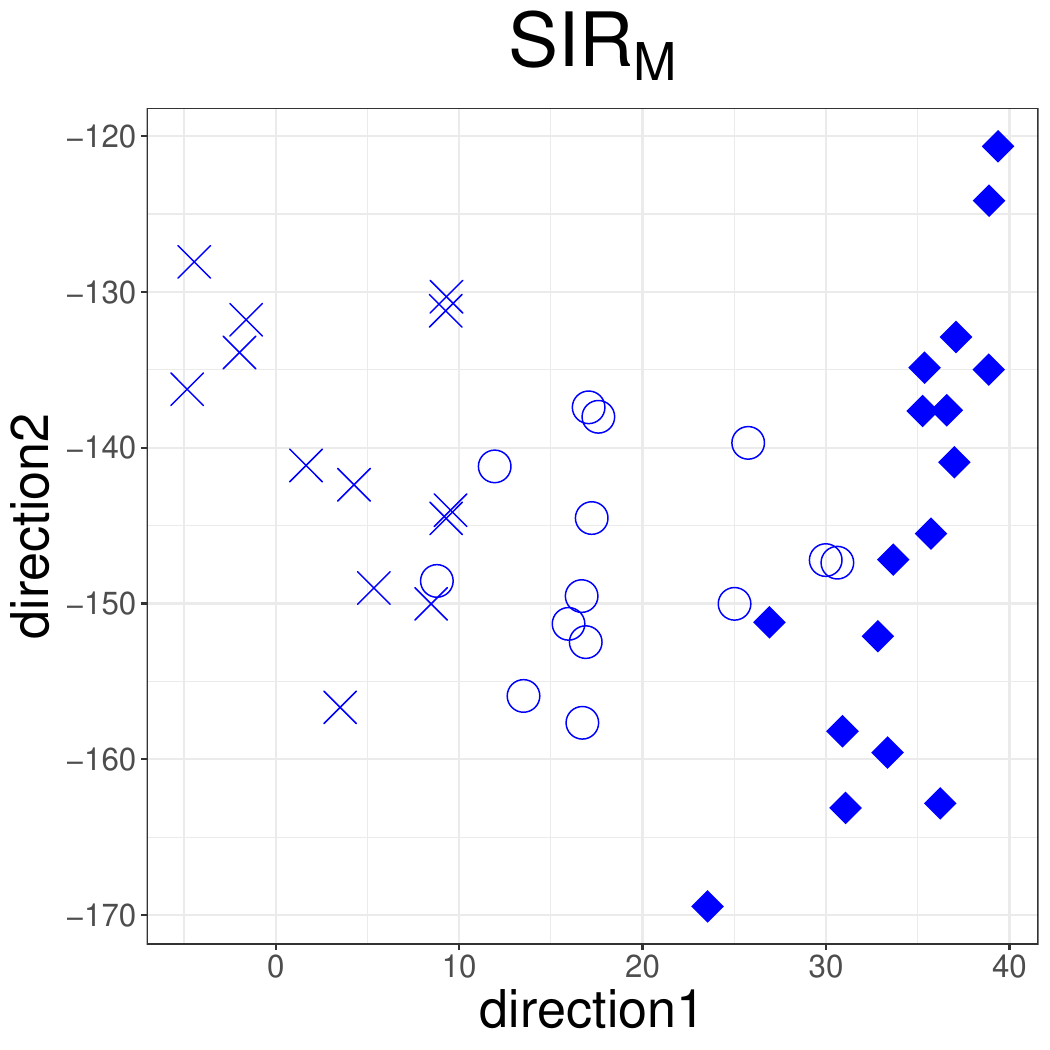}
\includegraphics[width=.33\textwidth]{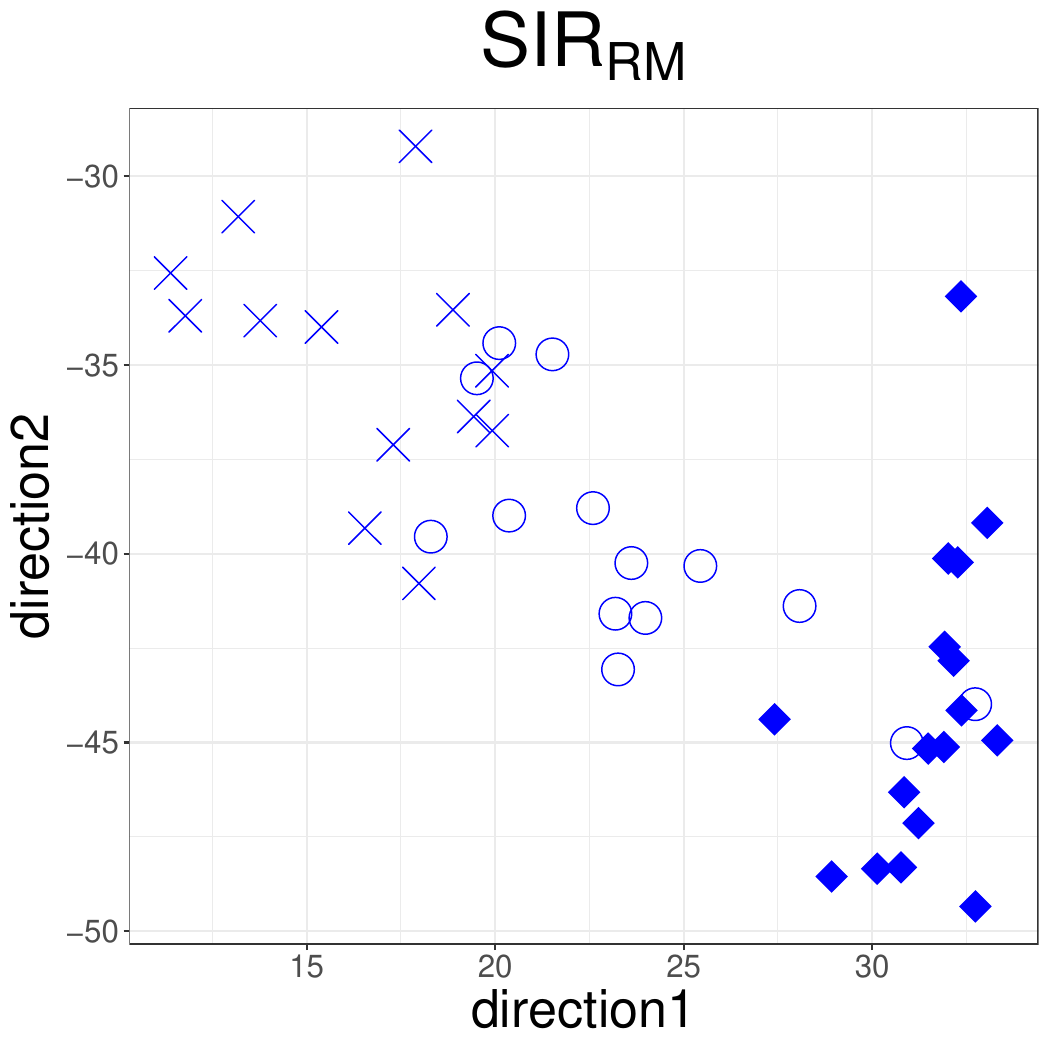}\vspace{1.5em}
\includegraphics[width=.40\textwidth]{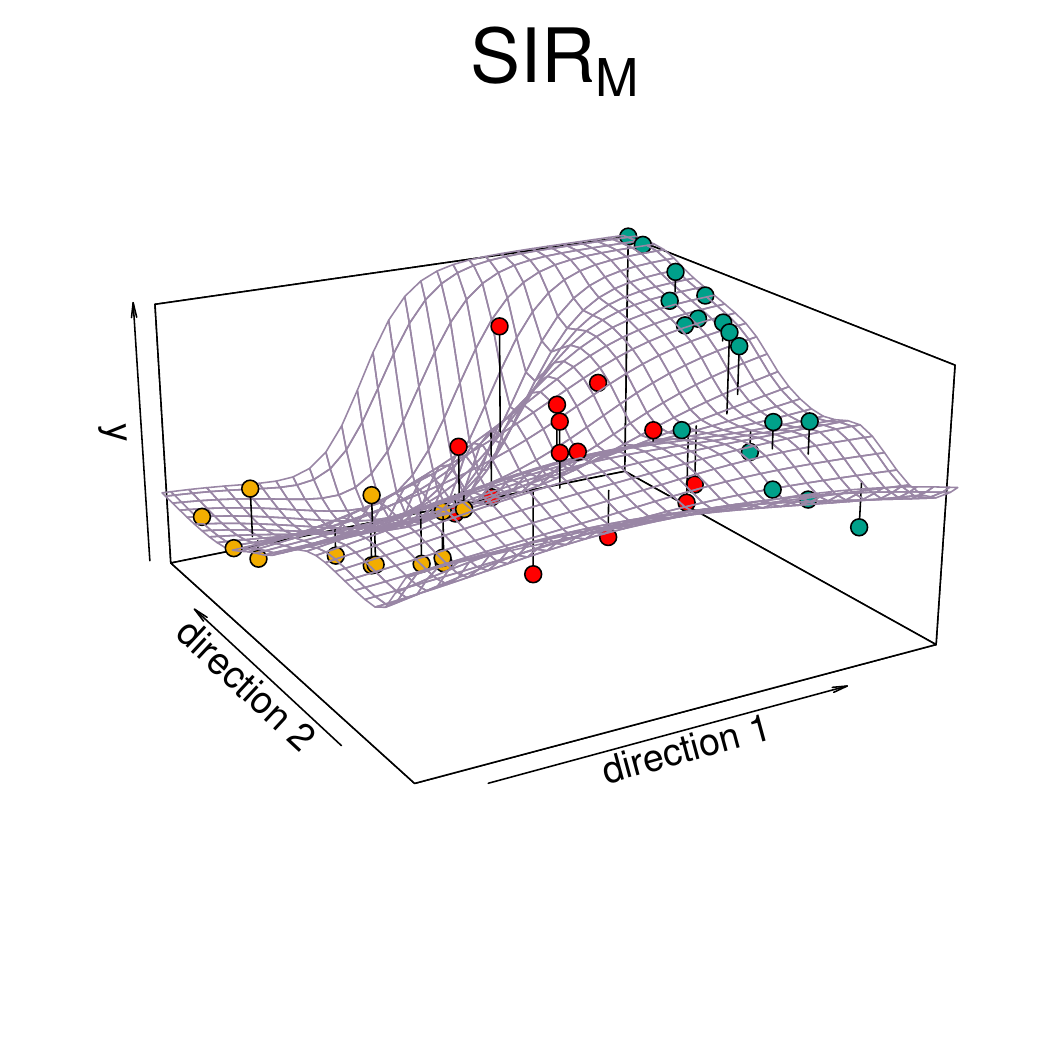}
\includegraphics[width=.40\textwidth]{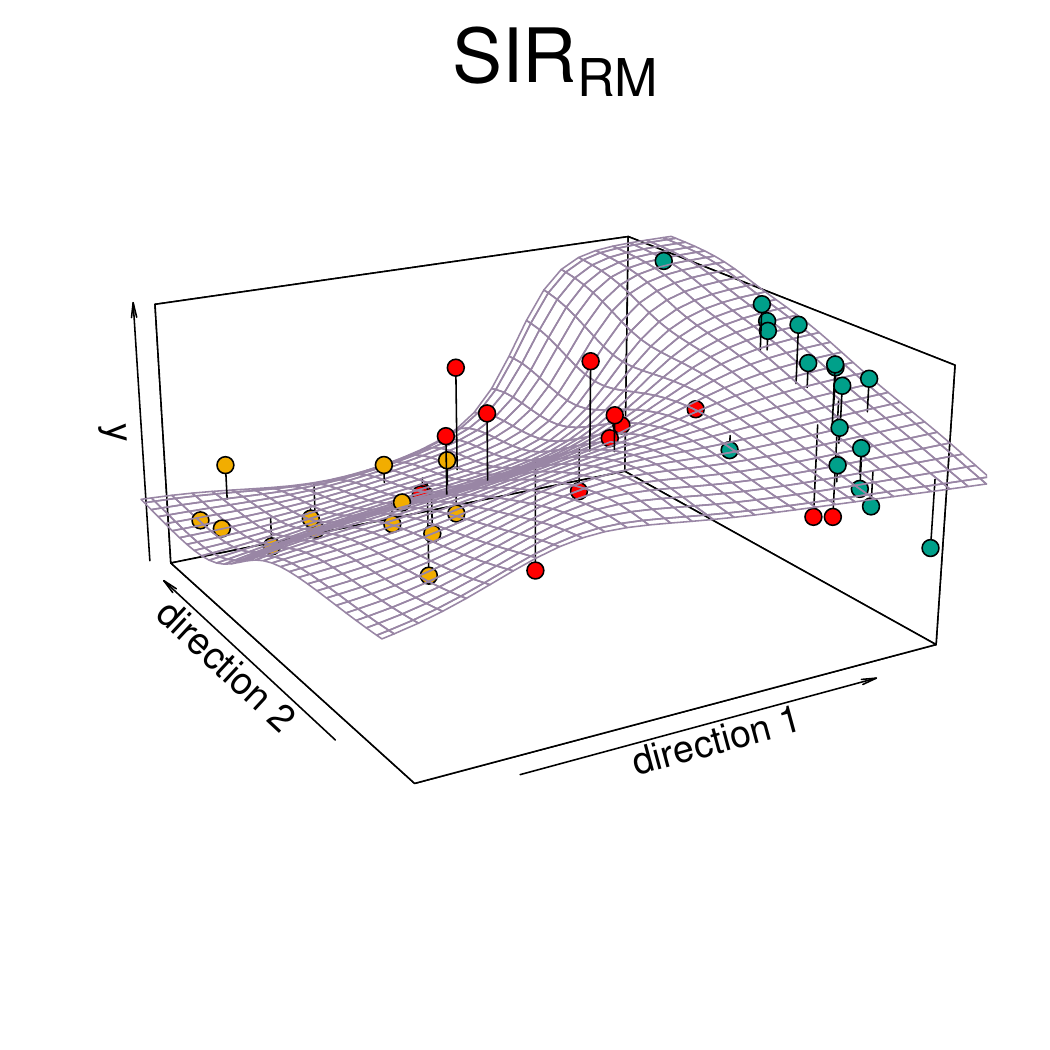}
\caption{In the upper panel are the scatter plots of the reduced predictors, the left for $\msir$ and the right for $\rmsir$; in the lower panel are the 3D scatter plots of the reduced data with $Y$ in the $z$-axis, as well as the fitted surfaces by the kernel mean estimation, again the left for $\msir$ and the right for $\rmsir$.}
\label{figure: data}
\end{figure}
\end{center}

To visualize the effect of the reduced predictor from $\msir$ on the response, we draw the $3$-dimensional scatter plot of the reduced data in the lower-left panel of Figure \ref{figure: data}. Clearly, both components of the reduced predictor uniquely affect the response, and their effect interact with each other: the effect of the first component changes from moderately positive to dramatically positive as the second component grows, and the effect of the second component changes from monotone decreasing to monotone increasing as the first component grows. A similar pattern can be observed from the $3$-dimensional scatter plot of the reduced data from $\rmsir$, as depicted in the lower-right panel of Figure \ref{figure: data}. To better interpret these, we report the coefficients of the working reduced predictor for each SDR method in Table \ref{tab: data coefficients}, with each variable in the original predictor standardized to have the unit variance. 

From Table \ref{tab: data coefficients}, both  $\msir$ and $\rmsir$ suggest that the median SAT Verbal score, the percentage of applicants accepted by each college, and the out-of-state tuition are the dominating variables in predicting the percentage of the freshman class that graduate, among which the out-of-state tuition seems to constantly contribute a positive effect. These may help the researchers comprehend the mechanism underlying the data set in the future investigation.

\begin{center}
\begin{table}[t]
\caption{Coefficients of the reduced predictor from SDR for MASSCOL2 data set}
\label{tab: data coefficients}
\centering
\begin{threeparttable}
\begin{tabular}[l]{@{}lrrcrr}
\hline \hline
\multirow{2}{*}{$X$}   & \multicolumn{2}{c}{$\msir$}&\multirow{2}{*}{} &\multicolumn{2}{c}{$\rmsir$}\\
\cline{2-3} \cline{5-6}&direction 1  & direction 2 & &direction 1 &direction 2   \\
\hline
Top 25 & .201 &-.014 & &.193   &-.186\\
MSAT   &-.197 &-.041 & &-.076  &.330\\
VSAT   &-.125 &-.516 & &-.213  &-.726\\
Accept &-.193 &-.846 & &-.189  &-.431\\
Enroll &.069  &.016  & &.133   &.299\\
SFRatio&.223  &-.035 & &.324   &-.027\\
Tuition&.901  &-.118 & &.866   &.226\\
\hline
\hline
\end{tabular}
\end{threeparttable}
\end{table}
\end{center}

\section{Discussion}

In the proposed work, we conduct SDR under the mixture model of $X$, which relaxes the linearity and constant variance conditions commonly adopted in the inverse regression methods, and meanwhile provides a parametric analogue of the MAVE-type methods such as ADR. Complying with the wide applicability of mixture models in approximating unknown distributions, our numerical studies suggest the effectiveness of the proposed SDR methods for generally distributed predictors, even when the clustered pattern is not obvious.

One issue that is underrepresented in this article is the difficulty of fitting the mixture model. This includes a possibly high computational cost for implementing the maximal likelihood estimator, which is sensitive to the initial value, as well as the potential failure of recovering the mixture components with small proportions when the underlying mixture model is severely unbalanced. As mentioned in Section 4, under the high-dimensional settings, only the mixture multivariate normal distribution has been studied in the existing literature, with the number of mixture components $q$ assumed known, and with the homogeneity of $\Sigma\lo i$'s and the sparsity of $\Sigma\lo i\inv (\mu\lo i - \mu\lo j)$'s imposed for estimation consistency. These limitations may restrict the applicability of the proposed SDR methods, especially for large-dimensional data.   

Finally, we speculate that the proposed strategy can be applied to adjust other SDR methods besides inverse regression. For example, we can assume a mixture model on $X|Y$ instead of on $X$, which will adjust the SDR methods that assume $X|Y$ to have an elliptical contoured distribution or fall in the exponential family \citep{cook2007fisher,cook2008pfc}. These extensions may further facilitate the applicability of SDR for data sets with large-dimensional and non-elliptically distributed predictor. 

\begin{acks}[Acknowledgments]
\end{acks}

\appendix
\section{Regularities and other adjusted SDR methods}

In this Appendix, we include the adjustments of pHd and directional regression, as well as the regularity conditions (C1)-(C3) needed for Theorem 3 in the main text. The equations in this Appendix are labeled as (A.1), (A.2), etc.

\setcounter{equation}{0}
\renewcommand\theequation{A.\arabic{equation}}

\subsection{Adjusted pHd and directional regression}
\def\mphd{{\Omega}\lo {\mathrm {pHd }}}
\def\hmphd{\widehat {\Omega}\lo {\mathrm {pHd }}}
\def\hsphd{{\mathrm {pHd}}\lo {\mathrm {M}}}
\def\gphd{G\lo {\mathrm {pHd}}}
\def\hgphd{\widehat G\lo {\mathrm {pHd}}}
\def\hsrphd{{\mathrm {pHd}}\lo {\mathrm {RM}}}

\def\mdr{{\Omega}\lo {\mathrm {DR}}}
\def\gdr{G\lo {\mathrm {DR}}}
\def\hgdr{\widehat G\lo {\mathrm {DR}}}
\def\hsdr{{\mathrm {DR}}\lo {\mathrm {M}}}
\def\hsrdr{{\mathrm {DR}}\lo {\mathrm {RM}}}

The candidate matrix for pHd is $\Sigma\lo X\inv E[(X-\mu\lo X)\udex {\otimes 2} \{Y -E(Y)\}] \Sigma\lo X\inv$. Thus, the first adjustment of pHd is to recover the central mean subspace $\cms$ by the column space of
\begin{align}\label{eq: mphd}
	\mphd = \left[ \Sigma\lo i\inv \left[ E\{\pi\lo i (X) Y\} \Sigma\lo i  - E\{\pi\lo i (X)(X - \mu\lo i)\udex {\otimes 2} Y\} \right] \Sigma\lo i\inv \right]\lo {i\in\{1,\ldots, q\}}.
\end{align}
A $n\udex {1/2}$-consistent estimator $\hmphd$ can be constructed by plugging the mixture model fit and using the sample moments, whose leading left singular vectors span a $n\udex {1/2}$-consistent estimator of $\cms$, denoted by $\hsphd$. In addition, define
\begin{align*}
	\gphd (\spn \beta) = \sum\lo {j=1}\udex q\left\| E \left[\{ \pi\lo i (\beta\trans X) \Sigma\lo i Q(\Sigma\lo i, \beta) - X\udex {\otimes 2} + \mu\lo {2}(\beta\trans X)\} Y \pi\lo j(\beta\trans X)\right]\right\|\lo F\udex 2
\end{align*} 
where $\mu\lo {2}(\beta\trans X) = \tsum\lo {i=1}\udex q \pi\lo i (\beta\trans X) \{P\trans (\Sigma\lo i, \beta) (X - \mu\lo i) + \mu\lo i\}\udex {\otimes 2}$. The second adjustment of pHd is to minimize a sample-level $\hgphd(\cdot)$, and the resulting minimizer spans a $n\udex {1/2}$-consistent estimator of $\cms$, denoted by $\hsrphd$.

Let $(\widetilde X, \widetilde Y)$ be an independent copy of $(X, Y)$. The candidate matrix for directional regression is $\Sigma\lo X\inv E\{2\Sigma\lo X - \var(X - \widetilde X | Y, \widetilde Y)\}\udex 2 \Sigma\lo X\inv$. The first adjustment of directional regression is to recover $\cs$ by the column space of
\begin{align*} 
	\mdr = & E \left[ \Sigma\lo i\inv \left[ 2 E\{\pi\lo i (X) |Y\lo S = s\} E\{\pi\lo i (\widetilde X) | \widetilde Y\lo S = t\} \Sigma\lo i \right.\right.\\
	& \hspace{.5cm} \left.\left.- E\{ \pi\lo i (X) \pi\lo i (\widetilde X)(X - \widetilde X)\udex {\otimes 2} | Y\lo S =s , \widetilde Y\lo S = t\} \right] \Sigma\lo i \inv  \right]\lo {s, t \in \{ 1,\ldots, H\}, i\in\{1,\ldots, q\}}.
\end{align*}
A $n\udex {1/2}$-consistent estimator of $\mdr$ can be constructed by plugging the mixture model fit and using the sample moments, whose leading left singular vectors span a $n\udex {1/2}$-consistent estimator of $\cs$, denoted by $\hsdr$. In addition, we propose the following objective function
\begin{align*}
	\gdr (\spn \beta) = & \left\| E\left[ (\pi\lo i (\beta\trans X) \pi\lo j (\beta\trans \widetilde X))\lo {i,j\in\{1,\ldots, q\}} \right.\right. \\
	& \hspace{.7cm} \left[ \tsum\lo {i=1}\udex q E\{\pi\lo i (\beta\trans X) | Y\lo S = s \} \Sigma\lo i Q(\Sigma\lo i, \beta)\right. \\
	& \hspace{.9cm}+ \tsum\lo {j=1}\udex q E\{\pi\lo {j} (\beta\trans \widetilde X) | \widetilde Y\lo S = t \} \Sigma\lo j Q(\Sigma\lo j, \beta)   \\
	& \hspace{.9cm}- E\{(X - \widetilde X)\udex {\otimes 2} | Y\lo S =s, \widetilde Y\lo S = t\}   \\
	& \hspace{.9cm} + \tsum\lo {i=1}\udex q E[\pi\lo i (\beta\trans X) \{ P\trans (\Sigma\lo i, \beta) (X - \mu\lo i) + \mu\lo i\}\udex {\otimes 2} | Y\lo S = s]  \\
	& \hspace{.9cm} + \tsum\lo {j=1}\udex q E[\pi\lo { j} (\beta\trans \widetilde X) \{ P\trans (\Sigma\lo j, \beta) (\widetilde X - \mu\lo { j}) + \mu\lo j\}\udex {\otimes 2} |\widetilde Y\lo S = t] \\
	& \hspace{.9cm} - [\tsum\lo {i=1}\udex q E[ \pi\lo i (\beta\trans X) \{P\trans (\Sigma\lo i, \beta) (X - \mu\lo i) + \mu\lo i\} | Y\lo S = s] \\
	& \hspace{1.6cm} [\tsum\lo {j=1}\udex q E[ \pi\lo j (\beta\trans \widetilde X) \{P\trans (\Sigma\lo j, \beta) (\widetilde X - \mu\lo j) + \mu\lo j\} | \widetilde Y\lo S = t]]\trans  \\
	& \hspace{.9cm} -  [\tsum\lo {j=1}\udex q E[ \pi\lo j (\beta\trans \widetilde X) \{P\trans (\Sigma\lo j, \beta) (\widetilde X - \mu\lo j) + \mu\lo j\} | \widetilde Y\lo S = t]]  \\
	& \hspace{1.4cm} \left.\left.\left. [\tsum\lo {i=1}\udex q \!E[ \pi\lo i (\beta\trans X) \{P\trans (\Sigma\lo i, \beta) (X \!- \!\mu\lo i) \!+\! \mu\lo i\} | Y\lo S\! =\! s]]\trans \right]\udex {\otimes 2}  \right] \right\|\lo F\udex 2.
\end{align*}
The second adjustment of directional regression is to minimize a sample-level $\hgdr(\cdot)$, and the resulting minimizer spans a $n\udex {1/2}$-consistent estimator of $\cs$, denoted by $\hsrdr$. The minimization of both $\hgphd(\cdot)$ and $\hgdr(\cdot)$ requires iterative algorithms, which resemble
those for implementing $\rmsir$ and $\rmsave$ discussed later in Appendix B and thus are omitted.

\subsection{Regularity conditions for Theorem 3}
\noindent
(C1). The signal-to-noise ratio $\Delta$ of the mixture multivariate normal distribution, defined as $\{(\mu\lo 1 - \mu\lo 2)\trans \Sigma\inv (\mu\lo 1 - \mu\lo 2)\}\udex {1/2}$, satisfies $\Delta > C(c\lo 0, c\lo 1, M, C\lo b)$, where $C(c\lo 0, c\lo 1, M, C\lo b)$ is given in (C.24) in supplementary material for \cite{cai2019}.

\noindent
(C2). $Y$ has a bounded support. $E(X|Y)$ has a total variation of order $0.25$, i.e.
\begin{align*}
	\sup\lo {B > 0} \lim\lo {n \rightarrow \infty} n\udex {-1/4} \sup\lo {\Xi(B)} \tsum\lo {i=2}\udex n \| E(X | Y = a\lo i) - E(X | Y = a\lo {i-1})\|\lo \max = 0,
\end{align*}
where $\|v\|\lo \max = \max \{v\lo i: i\in\{1,\ldots,p\}\}$ for any $v = (v\lo 1,\ldots, v\lo p)\trans \in \real\udex p$, and, for any $B > 0$,
$\Xi (B)$ is the collection of all the $n$-point partitions $-B \leq a\lo {1} \leq \ldots \leq a\lo n \leq B$.

\noindent
(C3). Define the $s\lo 2$-sparse minimal and maximal eigenvalues of $\Sigma$ as
\begin{align*}
	\lambda\lo \min (\Sigma, s\lo 2) = \min \lo {\|v\|_0 \leq s_2} (v\trans \Sigma v / v\trans v), \quad
	\lambda\lo \max (\Sigma, s\lo 2) = \max \lo {\|v\|_0 \leq s_2} (v\trans \Sigma v / v\trans v)
\end{align*}
where $\|v\|\lo 0$ is the number of nonzero elements in $v$ for any $v\in \real\udex p$. Assume that there exist constants $c\lo 1, c\lo 2 > 0$ such that $c\lo 1 \leq \lambda\lo \min (\Sigma, s\lo 2) \leq \lambda\lo \max (\Sigma, s\lo 2) \leq c\lo 2$.

\section{Algorithms for the proposed methods}

\setcounter{equation}{0}
\renewcommand\theequation{B.\arabic{equation}}

\def\hbmsir{\widehat\beta\lo {\mathrm {mSIR}}}
\def\hbrmsir{\widehat\beta\lo {\mathrm {rmSIR}}}
\def\hbmsave{\widehat\beta\lo {\mathrm {mSAVE}}}
\def\hbrmsave{\widehat\beta\lo {\mathrm {rmSAVE}}}
\def\hMsave{\widehat\Omega\lo {2}}

In this Appendix, we elaborate the details in implementing $\msir$, $\rmsir$, $\msave$, and $\rmsave$.

\subsection{Algorithms for $\msir$ and $\rmsir$}

Given $\hmo$, $\msir$ is readily derived by singular value decomposition. We list the detailed steps in Algorithm \ref{alg: es1_sir}.

\begin{algorithm}[t]
	\caption{$\msir$}
	\label{alg: es1_sir}
	\begin{algorithmic}[1] 
		\STATE Use the EM algorithm to calculate the maximal likelihood estimators $\widehat\mu\lo i$, $\widehat\Sigma\lo i$, and $\widehat \pi\lo i$ for $i\in\{1,\dots, q\}$, and derive the posterior probability $\widehat\pi\lo i(X)$ accordingly;
		\label{code:step1}
		\STATE Plug the estimators above into $\mo$ to derive
		\begin{align*}
			\hmo = E\lo n[\{\widehat \Sigma\lo i\inv(X-\widehat\mu\lo i)\widehat\pi\lo i(X)\vec Y\lo S\}\lo {i\in\{1,\ldots, q\}}]
		\end{align*}
		and calculate the linear span of the leading $d$ left singular vectors of $\hmo$ to estimate $\cs$.
	\end{algorithmic}
\end{algorithm}

To implement $\rmsir$, we fix $\beta\lo 0$ at $\widehat \beta\lo 1$ in $D(\beta\lo 0\trans X)$ and each $\pi\lo i (\beta\lo 0\trans X)$ involved in $\go (\cdot)$, where $\widehat \beta\lo 1$ is an arbitrary orthonormal basis matrix of $\spn \hmo$ derived from $\msir$. The estimators $\widehat \pi\lo i (\widehat \beta\lo 1\trans X)$ are then derived from the mixture model fit mentioned in Section 2, which deliver
\begin{align*}
	\widehat D (\widehat \beta\lo 1\trans X) = \tsum\lo{i=1}\udex q\tsum\lo{j=1}\udex q(\widehat\beta\lo1\trans \widehat\Sigma\lo j \widehat\beta\lo1)\inv
	E\lo n\{ \widehat\beta\lo1\trans(X-\widehat\mu\lo j)\vec Y\lo S \widehat\pi\lo i(\widehat\beta\lo1\trans X) \widehat\pi\lo j(\widehat\beta\lo1\trans X) \}\widehat\pi\lo i(\widehat\beta\lo1\trans X).
\end{align*}
Together with $\widehat R\lo X (\spn \beta) = \tsum\lo {i=1}\udex q \widehat \pi\lo i (\widehat \beta\lo 1\trans X) \{P\trans(\widehat \Sigma\lo i, \beta) (X - \widehat \mu\lo i ) + \widehat \mu\lo i \}$, we have
\begin{align}\label{eq: hgo}
\hgo (\spn \beta) \equiv \left\| E\lo n [\{X - \widehat R\lo X (\spn \beta)\} D\lo s ] \right\|\lo F\udex 2.
\end{align}
where $D\lo s$ denotes $\vecc\trans \{\widehat D (\widehat \beta\lo 1\trans X) \diag(\vec Y\lo S)\}$. To minimize $\hgo(\cdot)$, we propose an iterative algorithm as follows. Start with $\msir$ as the initial estimator of $\cs$. In the $k$th iteration, let $\widehat \beta\lo k$ be an orthonormal basis matrix of the most updated estimate of $\cs$. We replace each $P\trans (\widehat \Sigma\lo i, \beta)$ in $\hgo(\spn \beta)$ with $\widehat \Sigma\lo i \beta (\widehat \beta\lo k\trans \widehat \Sigma\lo i \widehat \beta\lo k)\inv \widehat \beta\lo k\trans$, which changes $\hgo (\spn \beta)$ into a quadratic function of $\beta$ that can be readily minimized, with the minimizer being the update $\spn {\widehat \beta\lo {k+1}}$. We repeat the iterations until a prefixed convergence threshold for $\delta (\widehat \beta\lo k, \widehat \beta\lo {k+1})$ defined in (\ref{eq: distance S hat S}) of the main text is met. These are summarized as Algorithm \ref{alg: es2_sir}. 

\begin{algorithm}[t]
	\caption{$\rmsir$}
	\label{alg: es2_sir}
	\begin{algorithmic}[1] 
		\STATE Use $\msir$ to estimate $\cs$ in each $\pi\lo i (\beta\lo 0\trans X)$, and use the mixture model fit to construct $\hgo (\spn \beta)$ as in (\ref{eq: hgo}); let $\spn {\widehat \beta\lo 1}$ derived from $\msir$ be the initial value of $\spn \beta$.
		\label{code:step1}
		\STATE In the $k$th iteration, replace each $P\trans (\widehat \Sigma\lo i, \beta)$ in $\hgo(\spn \beta)$ with $\widehat \Sigma\lo i \beta (\widehat \beta\lo k\trans \widehat \Sigma\lo i \widehat \beta\lo k)\inv \widehat \beta\lo k\trans$, and solve the corresponding quadratic minimization problem to derive $\spn {\widehat \beta\lo {k+1}}$.
		\STATE Repeat the iterations until $\delta (\widehat \beta\lo k, \widehat \beta\lo {k+1}) < c\lo 0$ for some prefixed small scalar $c\lo 0$, e.g. $n\inv$. The final estimator of $\cs$ is the most updated $\spn {\widehat \beta\lo {k+1}}$.
		\label{code:step3}
	\end{algorithmic}
\end{algorithm}

\subsection{Algorithms for $\msave$ and $\rmsave$ }

Given $\Msave$, $\msave$ is readily derived by singular value decomposition. We list the detailed steps in Algorithm \ref{alg: es1_save}.

\begin{algorithm}[t]
	\caption{ $\msave$}
	\label{alg: es1_save}
	\begin{algorithmic}[1] 
		\STATE Apply EM algorithm to estimate $\widehat\mu\lo i$, $\widehat\Sigma\lo i$, and $\widehat\pi\lo i(X)$, $i\in\{1,\ldots, q\}$;
		\label{code:step1}
		\STATE Calculate the span of leading $d$ left singular vectors of $\hMsave$ as $\msave$, denoted by $\spn {\hbmsave}$.
		\label{code:step2}
	\end{algorithmic}
\end{algorithm}

To implement $\rmsave$, we fix $\beta\lo 0$ at $\widehat \beta\lo 1$ in each $\pi\lo i (\beta\lo 0\trans X)$ involved in $\gsave (\cdot)$, where $\widehat \beta\lo 1$ is an arbitrary orthonormal basis matrix of $\spn \hMsave$ derived from $\msave$. The estimators $\widehat \pi\lo i (\widehat \beta\lo 1\trans X)$ are then derived from the mixture model fit mentioned in Section 2. We have, 
\begin{align}\label{eq:hgsave}
\widehat \G\lo 2(\spn \beta)\equiv 
\left\|\tsum\lo{i=1}\udex q\tsum\lo{h=1}\udex H (A\lo{i,h}+B\lo{i,h})\udex{\otimes 2} \right\|\lo F\udex 2. 
\end{align}
where
\begin{align}\notag
A\lo{i,h} &= \E\lo n[X\udex{\otimes 2}\widehat\pi\lo i\udex 2 (\widehat\beta\lo1\trans X)Y\lo {S, h}] - \tsum\lo{j=1}\udex q E\lo n[\widehat\pi \lo i\udex2 (\widehat\beta\lo1\trans X) \widehat\pi\lo j (\widehat\beta\lo1\trans X)Y\lo {S, h}]\widehat\Sigma \lo j \\
\notag
& \quad - E\lo n[\widehat\pi\lo i \udex 2(\widehat\beta\lo1\trans X)
\mu \lo 2({\beta\trans} X) Y\lo {S, h}]\\
\notag
B\lo{i,h} &= \tsum\lo{j=1}\udex q E\lo n[\widehat\pi \lo i\udex2 (\widehat\beta\lo1\trans X) \widehat\pi\lo j (\widehat\beta\lo1\trans X) Y\lo {S, h}]\widehat\Sigma \lo j P(\widehat\Sigma\lo j,{\beta})
\end{align}
To minimize $\widehat \G\lo2(\cdot)$, we propose an iterative algorithm as follows. Start with $\msave$ as the initial estimator of $\cs$. In the $k$th iteration, let $\widehat \beta\lo k$ be an orthonormal basis matrix of the most updated estimate of $\cs$. We replace each $\mu\lo 2(\beta\trans X)$ in $A\lo{i,h}$ with $\mu\lo 2(\widehat\beta\lo k\trans X)$, and replace each $P(\widehat \Sigma\lo i, \beta)$ in $B\lo{i,h}$ with $\beta(\widehat \beta\lo k\trans \widehat \Sigma\lo i \widehat \beta\lo k)\inv \widehat\beta\lo k\trans\widehat \Sigma\lo i$, which changes $\widehat \G\lo2 (\spn \beta)$ into a quadratic function of $\beta$ that can be readily minimized, with the minimizer being the update $\spn {\widehat \beta\lo {k+1}}$. We repeat the iterations until a prefixed convergence threshold for $\delta (\widehat \beta\lo k, \widehat \beta\lo {k+1})$ defined in (\ref{eq: distance S hat S}) of the main text is met. These are summarized as Algorithm \ref{alg:es2_save}.

\begin{algorithm}[t]
\caption{ $\rmsave$ }
\label{alg:es2_save}
\begin{algorithmic}[1] 
\STATE Use $\msave$ to estimate $\cs$ in each $\pi\lo i (\beta\lo 0\trans X)$, and use the mixture model fit to construct $\widehat \G\lo 2(\spn \beta)$ as in (\ref{eq:hgsave}); let $\spn {\widehat \beta\lo 1}$ derived from $\msave$ be the initial value of $\spn \beta$.
\label{code:step1}
\STATE In the $k$th iteration, replace each $\mu\lo 2(\beta\trans X)$ in $A\lo{i,h}$ with $\mu\lo 2(\widehat\beta\lo k\trans X)$, and replace each $P(\widehat \Sigma\lo i, \beta)$ in $B\lo{i,h}$ with $\beta(\widehat \beta\lo k\trans \widehat \Sigma\lo i \widehat \beta\lo k)\inv \widehat\beta\lo k\trans\widehat \Sigma\lo i$, and solve the corresponding quadratic minimization problem to derive $\spn {\widehat \beta\lo {k+1}}$.
\STATE Repeat the iterations until $\delta (\widehat \beta\lo k, \widehat \beta\lo {k+1}) < c\lo 0$ for some prefixed small scalar $c\lo 0$, e.g. $n\inv$. The final estimator of $\cs$ is the most updated $\spn {\widehat \beta\lo {k+1}}$.
\label{code:step3}
\end{algorithmic}
\end{algorithm}

\section{Complementary simulation studies}

In this Appendix, we present the complementary simulation studies to evaluate how the proposed $\msir$ and $\rmsir$ are impacted by the estimation error in fitting the mixture model of $X$, as well as by hypothetical misspecification of the number of mixture components $q$, under the conventional low-dimensional settings and assuming that the mixture model of $X$ holds. The misspecification of $q$ does not occur in the simulation studies of the main text due to the effectiveness of BIC, but it is possible in real data analyses and thus worth the investigation. 

Generally, for an arbitrarily fixed working number of mixture components $k$, the corresponding mixture model misspecifies the distribution of $X$ if $k < q$, and it still correctly specifies the distribution if $k > q$. Thus, the proposed SDR methods theoretically lose consistency if $k < q$ and are still consistent if $k > q$. As an illustration, we set $k$ to be each of $1$, $2$, $3$, and $4$ sequentially for all the models in Subsection 6.1 of the main text, and apply $\msir$ and $\rmsir$ in each case. Note that this procedure differs from the implementation of $\msir$ and $\rmsir$ in Subsection 6.1, as $q$ is forced to be a fixed value rather than being estimated by BIC. After repeating the procedure for $200$ times, the results are summarized in Table \ref{tab: q misspecified}. When $k=1$, both $\msir$ and $\rmsir$ reduce to SIR, so the result of SIR in Table \ref{tab: ss SIR} is used for this case.

\begin{center}
\begin{table}[t]
\centering
\caption{Performance of $\msir$ and $\rmsir$ under misspecification of $q$}
\label{tab: q misspecified}
\begin{threeparttable}
\begin{tabular}[l]{ll|ccccc}
\hline \hline
$\widehat q$ &Methods  & Model 1  & Model 2  & Model 3 & Model 4\\ 
\hline
$1$ 
&SIR      &.908(.136)   &.910(.128) 
          &1.08(.179)   &.634(.127)\\
\hline
\multirow{2}{*}{$2$}
&$\msir$  &.590(.151) &.436(.094) 
          &.553(.111) &.375(.058)\\
&$\rmsir$ &.555(.146) &.402(.088) 
          &.502(.109) &.320(.076) \\
\hline
\multirow{2}{*}{$3$}
&$\msir$  &.698(.191) &.506(.154) 
          &.608(.192) &.374(.105)\\
&$\rmsir$ &.642(.156) &.416(.128) 
          &.545(.174) &.338(.098)\\
\hline
\multirow{2}{*}{$4$}
&$\msir$ &.801(.233)  &.549(.188) 
         &.636(.217)  &.388(.104)\\
&$\rmsir$ &.689(.194) &.435(.163) 
          &.590(.189) &.350(.095)\\
\hline
\multirow{3}{*}{\empty}
&ADR      &.611(.176) &.473(.138) 
          &1.31(.113) &.645(.189)\\
&MAVE     &.745(.229) &.223(.056) 
          &.895(.278) &.326(.089) \\
&eCVE     &1.42(.126) &.452(.128) 
          &.609(.222) &1.29(.226) \\
\hline
\end{tabular}
\begin{tablenotes}
\footnotesize
\item The meanings of numbers in each cell follow those in Table~\ref{tab: ss SIR}.
\end{tablenotes}
\end{threeparttable}
\end{table}
\end{center}

Recall that $q$ is $2$ for Models $1-3$ and is $3$ for Model 4. Thus, Table \ref{tab: q misspecified} illustrates both the impact of overestimation of $q$ and the impact of underestimation of $q$ on the proposed methods for each model. Compared with both the benchmark and the worst case scenario, i.e. the proposed methods with truly specified $q$ and SIR, respectively, an overestimation of $q$ brings negligible additional bias to the proposed methods and makes them still comparable and sometimes outperform the MAVE-type methods. By contrast, an underestimate $\widehat q=1$ brings significant bias to the proposed methods (i.e. SIR) for all the four models. These comply with the theoretical anticipation above. 

An interesting exception in Table \ref{tab: q misspecified} is the case of $\widehat q$ forced to be $2$ for Model 4, which suggests the effectiveness of the proposed methods even though $q$ is slightly underestimated. The reason is that the mixture of three multivariate normal distributions in this model can still be approximated by a mixture of two multivariate normal distributions to a good extent. To see this, we draw the scatter plot of $(X\lo 1, X\lo 2)$, the sub-vector of $X$ that has the clustered pattern, in Figure \ref{figure: model 4 underest q}. The point type of each observation $x$ is determined by an independent Bernoulli random variable, whose probability of success is $\widehat \pi\lo 1 (x)$ derived from the mixture model fit with $\widehat q = 2$. The two fitted clusters are then formed by the observations that have the same point types. From Figure \ref{figure: model 4 underest q}, each cluster conveys an approximate elliptical distribution, and, in particular, their loess fits of $E(X | \beta\lo 0\trans X)$ are approximately linear.

\begin{figure}[H]
\begin{center}
\includegraphics[width=.3\textwidth]{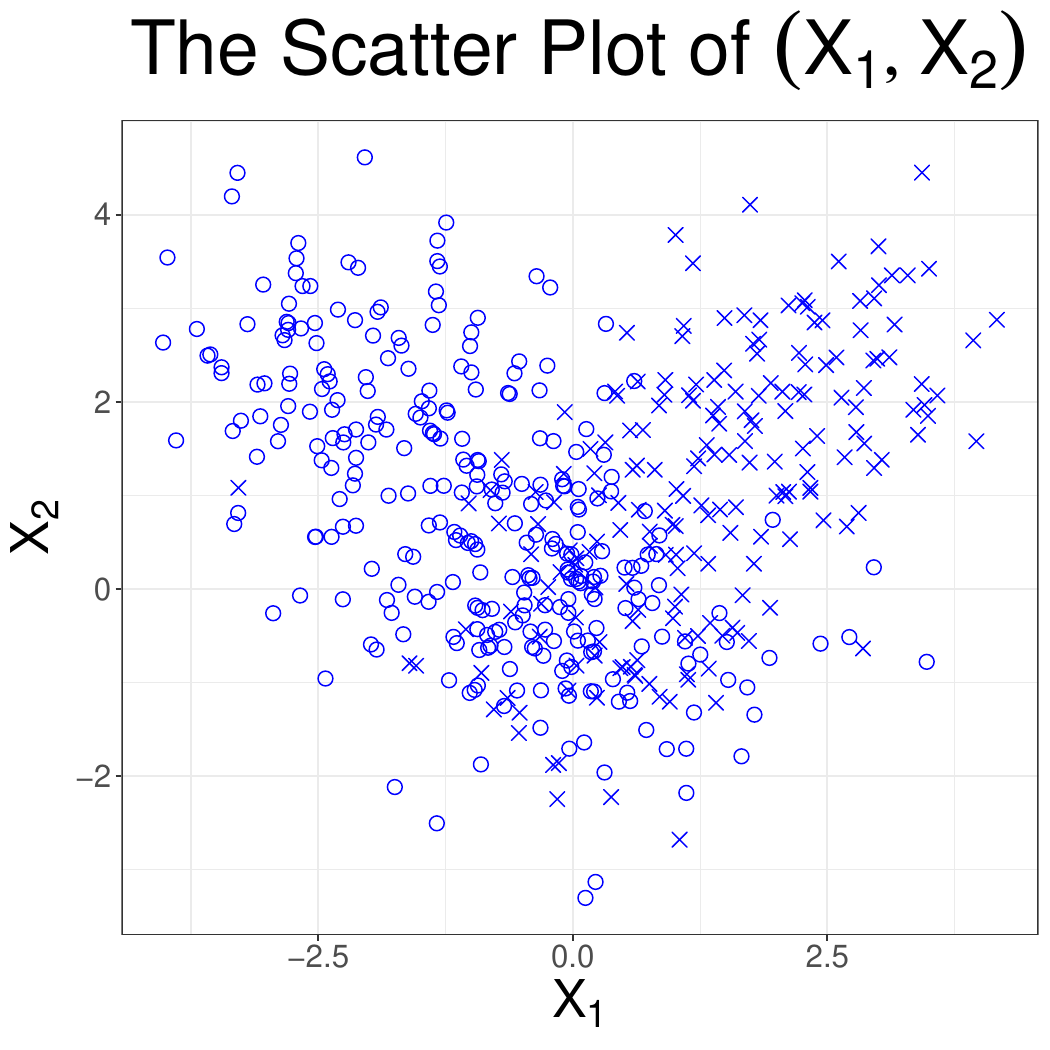}
\includegraphics[width=.3\textwidth]{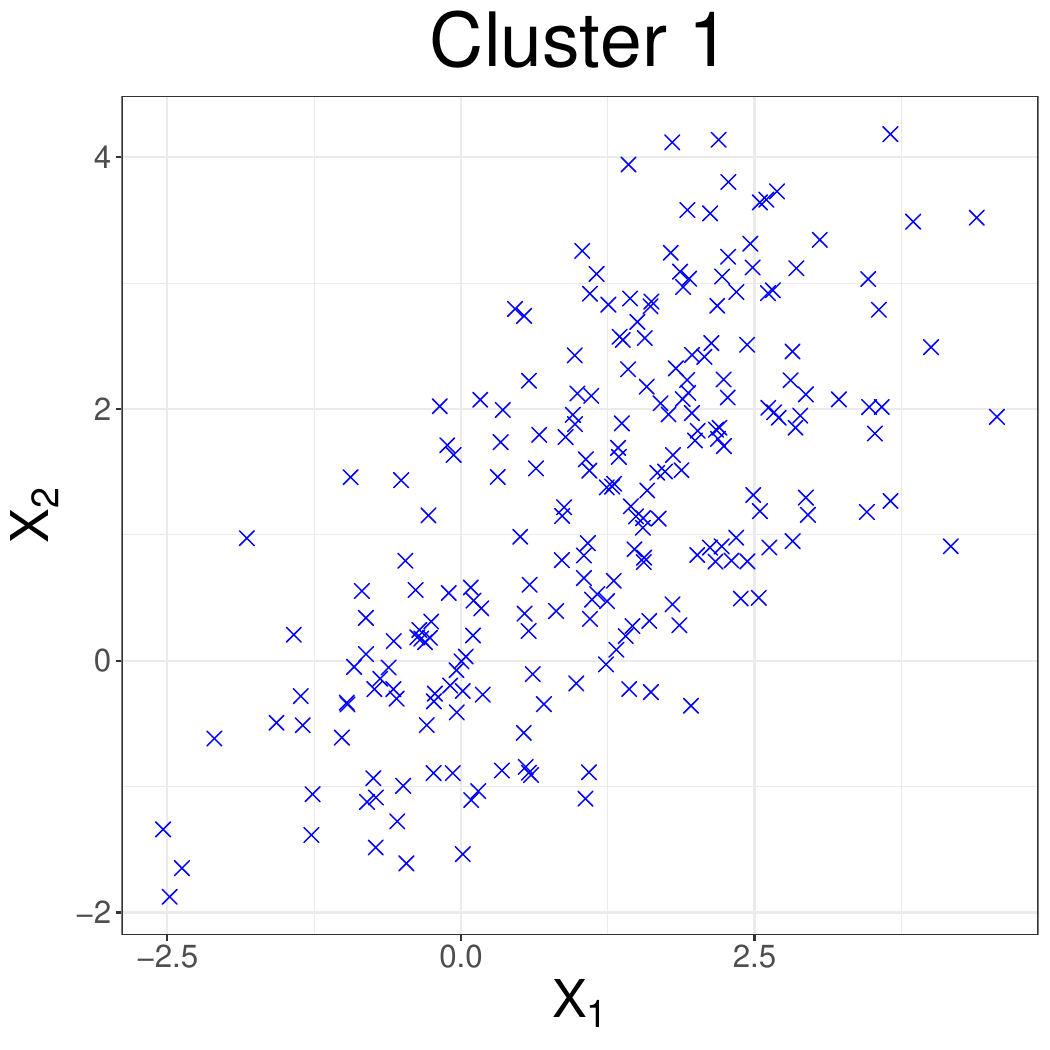}
\includegraphics[width=.3\textwidth]{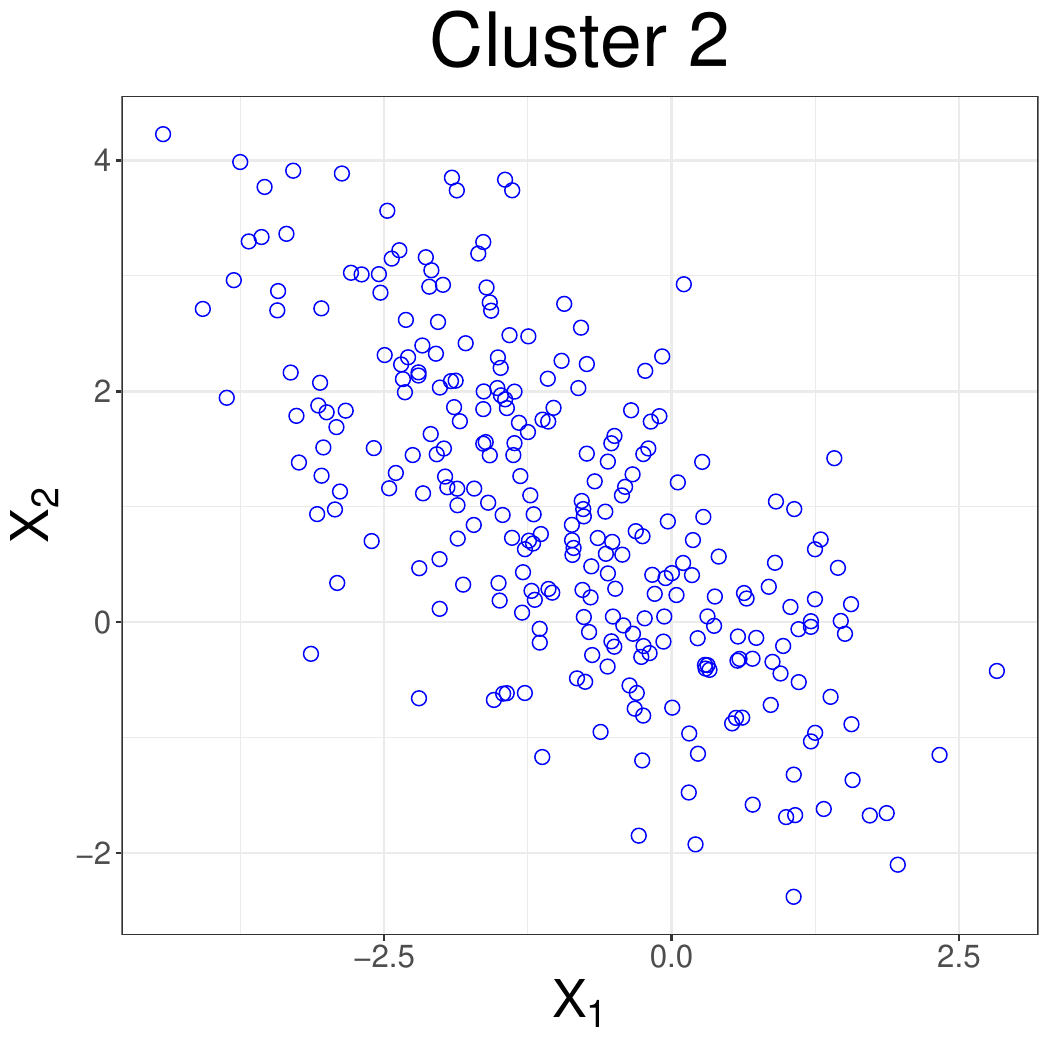}
\vspace{-.5cm}
\includegraphics[width=.3\textwidth]{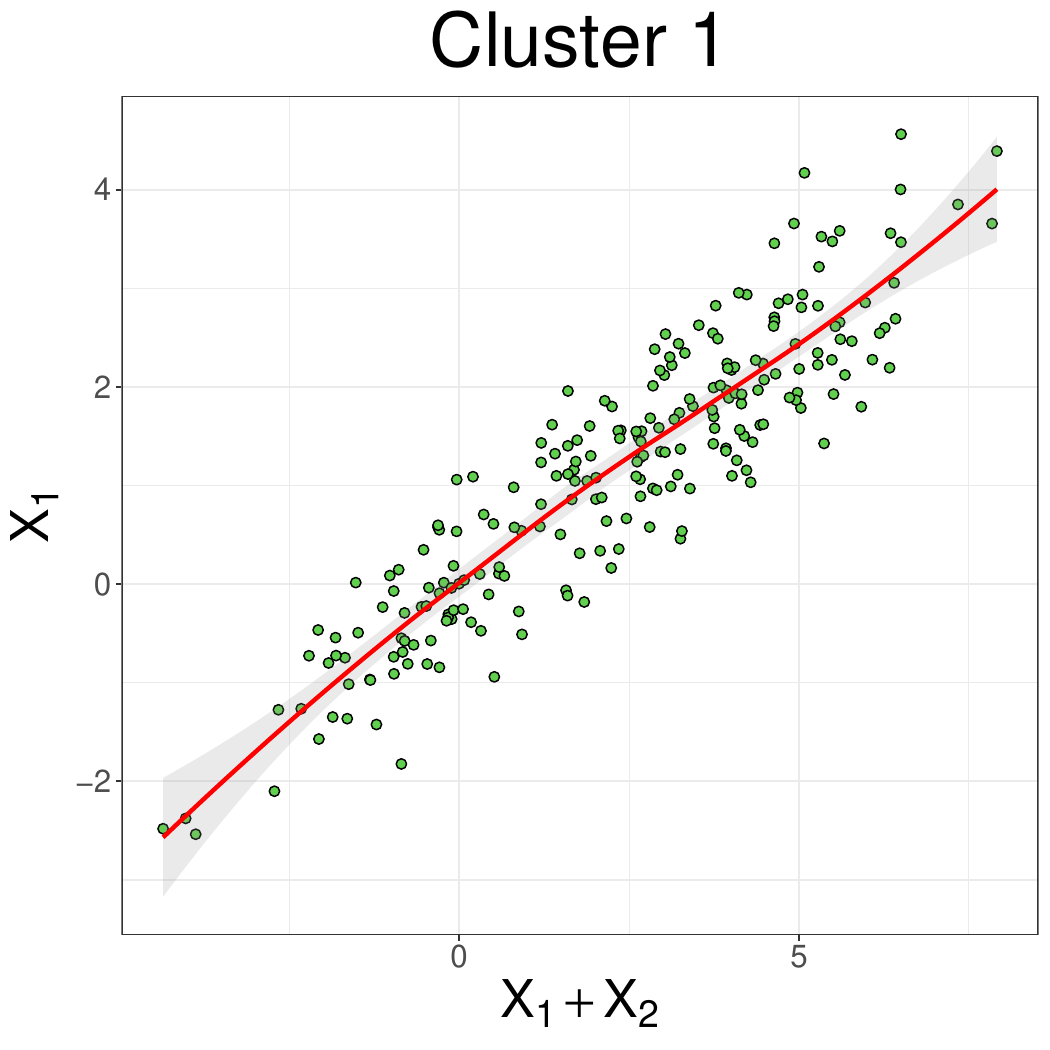}
\includegraphics[width=.3\textwidth]{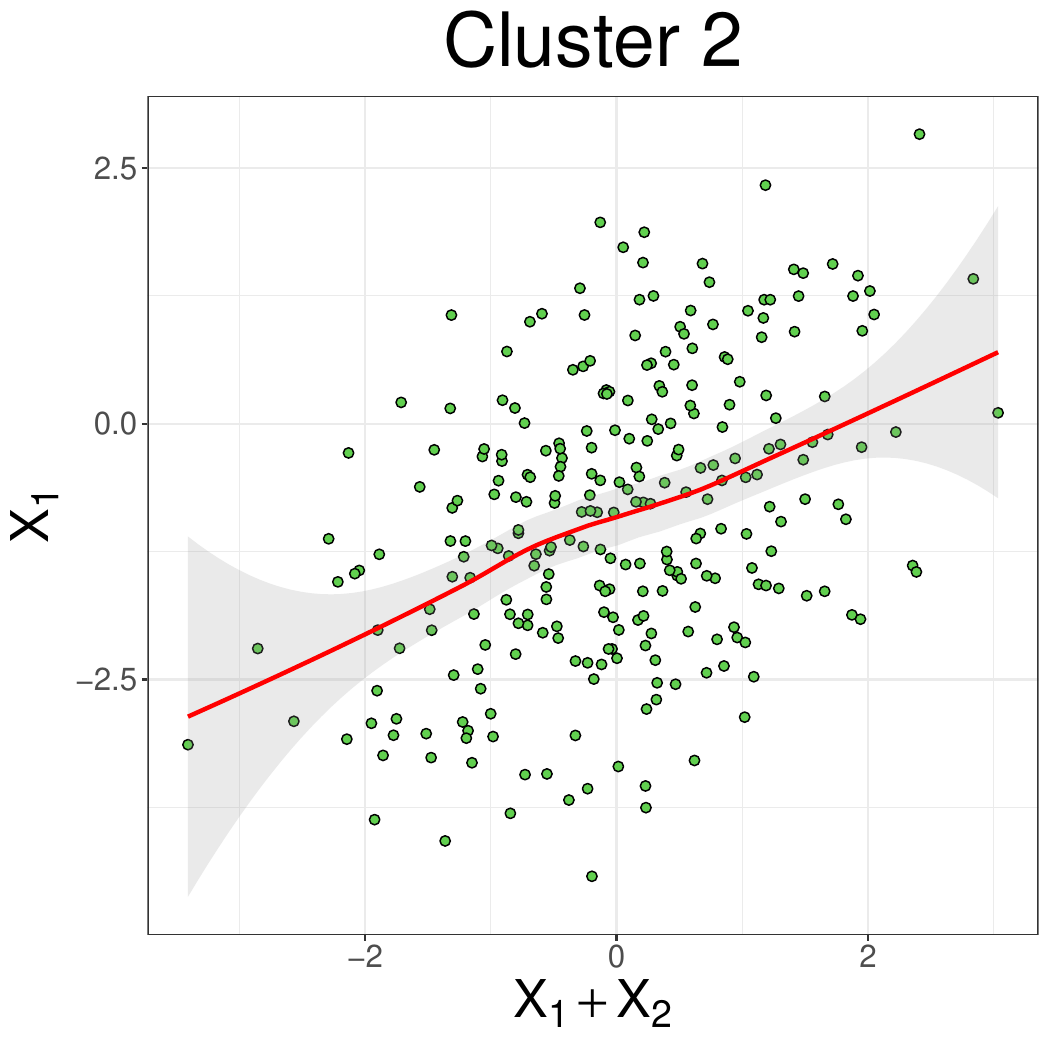}
\caption{\small In the upper panels are the scatter plots of the fitted clusters of Model 4 with $\widehat q$ forced at two, the left two plots specified for each cluster; in the lower panels are the scatter plots that illustrate the loess fit of $E(X | \beta\lo 0\trans X)$ within each cluster, where the x-axis is changed to $\beta\lo 0\trans X = X\lo 1 + X\lo 2$, and the y-axis is changed to $X\lo 1$.}
\label{figure: model 4 underest q}
\end{center}
\end{figure}

When $q$ is estimated by a consistent order-determination method, it may still be misspecified in practice. However, unlike the case of arbitrarily fixed $\widehat q$ above, an underestimation of $q$ in this case is plausibly a sign for severely overlapping mixture components in the data set; that is, two or more mixture components have similar distributions and thus can be regarded as one mixture component that approximately falls in the working parametric family ${\mathcal {F}}$ (e.g. multivariate normal). Thus, the distribution of $X$ can still be approximated by the working mixture model, indicating the effectiveness of the proposed SDR methods. Together with the robustness of the proposed methods to the choice of $q$ discussed above, a misspecification of $q$ should not be worrisome in practice.

Next, we assess the impact of the estimation error of $\widehat \pi$'s, $\widehat \mu\lo i$'s, and $\widehat \Sigma\lo i$'s on the proposed $\msir$ and $\rmsir$. To this end, we implement both SDR methods in the oracle scenario that the mixture model of $X$ is completely known for each model in Subsection 6.1 of the main text, and compare their performance (based on $200$ independent runs) with that in Table \ref{tab: ss SIR}. The results are summarized in Table \ref{tab: ss SIR oracle}. Generally, the estimation error in fitting the mixture model of $X$ has a negligible impact on both of the proposed methods. A relative asymptotic study is deferred to future.

\begin{center}
\begin{table}[t]
\centering
\caption{Comparison of the proposed methods with their oracle counterparts}
\label{tab: ss SIR oracle}
\begin{threeparttable}
\begin{tabular}[l]{l|cccc}
\hline \hline
\empty  & Model 1  & Model 2  & Model 3  & Model 4       \\ 
\hline
$\msir$  &.590(.151) &.436(.094) 
        &.553(.111)  &.374(.105) \\
$\rmsir$ &.555(.146) &.402(.085) 
        &.502(.109)  &.338(.098) \\
oracle-$\msir$  &.616(.158) &.469(.113) 
                &.544(.128) &.401(.105)\\
oracle-$\rmsir$ &.538(.138) &.314(.064) 
                &.476(.116) &.326(.078)\\
\hline
\end{tabular}
\begin{tablenotes}
\footnotesize
\item The methods ``oracle-$\msir$" and ``oracle-$\rmsir$" refer to $\msir$ and $\rmsir$ with the mixture model of $X$ completely known, respectively. The meanings of numbers in each cell of Columns $2-5$ follow those in Table~\ref{tab: ss SIR}.
\end{tablenotes}
\end{threeparttable}
\end{table}
\end{center}

	

\bibliographystyle{imsart-nameyear}
\bibliography{mixnormal}

\begin{thebibliography}{40}

\bibitem[\protect\citeauthoryear{Biernacki, Celeux and
  Govaert}{2000}]{biernacki2000assessing}
\begin{barticle}[author]
\bauthor{\bsnm{Biernacki},~\bfnm{Christophe}\binits{C.}},
  \bauthor{\bsnm{Celeux},~\bfnm{Gilles}\binits{G.}} \AND
  \bauthor{\bsnm{Govaert},~\bfnm{G{\'e}rard}\binits{G.}}
(\byear{2000}).
\btitle{Assessing a mixture model for clustering with the integrated completed
  likelihood}.
\bjournal{IEEE transactions on pattern analysis and machine intelligence}
\bvolume{22}
\bpages{719--725}.
\end{barticle}
\endbibitem

\bibitem[\protect\citeauthoryear{Cai, Ma and Zhang}{2019}]{cai2019}
\begin{barticle}[author]
\bauthor{\bsnm{Cai},~\bfnm{T~Tony}\binits{T.~T.}},
  \bauthor{\bsnm{Ma},~\bfnm{Jing}\binits{J.}} \AND
  \bauthor{\bsnm{Zhang},~\bfnm{Linjun}\binits{L.}}
(\byear{2019}).
\btitle{CHIME: Clustering of high-dimensional Gaussian mixtures with EM
  algorithm and its optimality}.
\bjournal{The Annals of Statistics}
\bvolume{47}
\bpages{1234--1267}.
\end{barticle}
\endbibitem

\bibitem[\protect\citeauthoryear{Chen, Zou and Cook}{2010}]{chen2010}
\begin{barticle}[author]
\bauthor{\bsnm{Chen},~\bfnm{X.}\binits{X.}},
  \bauthor{\bsnm{Zou},~\bfnm{C.}\binits{C.}} \AND
  \bauthor{\bsnm{Cook},~\bfnm{R.~D.}\binits{R.~D.}}
(\byear{2010}).
\btitle{Coordinate-independent sparse suffcient dimension reduction and
  variable selection}.
\bjournal{The Annals of Statistics}
\bvolume{6}
\bpages{3696--3723}.
\end{barticle}
\endbibitem

\bibitem[\protect\citeauthoryear{Chiaromonte, Cook and
  Li}{2002}]{chiaromonte2002sufficient}
\begin{barticle}[author]
\bauthor{\bsnm{Chiaromonte},~\bfnm{Francesca}\binits{F.}},
  \bauthor{\bsnm{Cook},~\bfnm{R~Dennis}\binits{R.~D.}} \AND
  \bauthor{\bsnm{Li},~\bfnm{Bing}\binits{B.}}
(\byear{2002}).
\btitle{Sufficient dimension reduction in regressions with categorical
  predictors}.
\bjournal{Annals of Statistics}
\bpages{475--497}.
\end{barticle}
\endbibitem

\bibitem[\protect\citeauthoryear{Cook}{1998}]{cook1998}
\begin{bbook}[author]
\bauthor{\bsnm{Cook},~\bfnm{R.~D.}\binits{R.~D.}}
(\byear{1998}).
\btitle{Regression Graphics}.
\bpublisher{Wiley, New York}.
\end{bbook}
\endbibitem

\bibitem[\protect\citeauthoryear{Cook}{2007}]{cook2007fisher}
\begin{barticle}[author]
\bauthor{\bsnm{Cook},~\bfnm{R~Dennis}\binits{R.~D.}}
(\byear{2007}).
\btitle{Fisher lecture: Dimension reduction in regression}.
\end{barticle}
\endbibitem

\bibitem[\protect\citeauthoryear{Cook and Forzani}{2008}]{cook2008pfc}
\begin{barticle}[author]
\bauthor{\bsnm{Cook},~\bfnm{R.~Dennis}\binits{R.~D.}} \AND
  \bauthor{\bsnm{Forzani},~\bfnm{Liliana}\binits{L.}}
(\byear{2008}).
\btitle{{Principal Fitted Components for Dimension Reduction in Regression}}.
\bjournal{Statistical Science}
\bvolume{23}
\bpages{485 -- 501}.
\end{barticle}
\endbibitem

\bibitem[\protect\citeauthoryear{Cook and Li}{2002}]{cook2002}
\begin{barticle}[author]
\bauthor{\bsnm{Cook},~\bfnm{R.~D.}\binits{R.~D.}} \AND
  \bauthor{\bsnm{Li},~\bfnm{B.}\binits{B.}}
(\byear{2002}).
\btitle{Dimension reduction for conditional mean in regression}.
\bjournal{The Annals of Statistics}
\bvolume{30}
\bpages{455--474}.
\end{barticle}
\endbibitem

\bibitem[\protect\citeauthoryear{Cook and Weisberg}{1991}]{cook1991}
\begin{barticle}[author]
\bauthor{\bsnm{Cook},~\bfnm{R.~D.}\binits{R.~D.}} \AND
  \bauthor{\bsnm{Weisberg},~\bfnm{S.}\binits{S.}}
(\byear{1991}).
\btitle{Discussion of ``{S}liced inverse regression for dimension reduction"}.
\bjournal{Journal of the American Statistical Association}
\bvolume{86}
\bpages{316--342}.
\end{barticle}
\endbibitem

\bibitem[\protect\citeauthoryear{Dong and Li}{2010}]{dong2010}
\begin{barticle}[author]
\bauthor{\bsnm{Dong},~\bfnm{Y.}\binits{Y.}} \AND
  \bauthor{\bsnm{Li},~\bfnm{B.}\binits{B.}}
(\byear{2010}).
\btitle{Dimension reduction for non-elliptically distributed predictors:
  second-order methods}.
\bjournal{Biometrika}
\bvolume{97}
\bpages{279--294}.
\end{barticle}
\endbibitem

\bibitem[\protect\citeauthoryear{Everitt}{2013}]{everitt2013finite}
\begin{bbook}[author]
\bauthor{\bsnm{Everitt},~\bfnm{Brian}\binits{B.}}
(\byear{2013}).
\btitle{Finite mixture distributions}.
\bpublisher{Springer Science \& Business Media}.
\end{bbook}
\endbibitem

\bibitem[\protect\citeauthoryear{Fertl and Bura}{2022}]{fertl2022ensemble}
\begin{barticle}[author]
\bauthor{\bsnm{Fertl},~\bfnm{Lukas}\binits{L.}} \AND
  \bauthor{\bsnm{Bura},~\bfnm{Efstathia}\binits{E.}}
(\byear{2022}).
\btitle{The ensemble conditional variance estimator for sufficient dimension
  reduction}.
\bjournal{Electronic Journal of Statistics}
\bvolume{16}
\bpages{1595--1634}.
\end{barticle}
\endbibitem

\bibitem[\protect\citeauthoryear{Guan, Xie and Zhu}{2017}]{guan2017}
\begin{barticle}[author]
\bauthor{\bsnm{Guan},~\bfnm{Yu}\binits{Y.}},
  \bauthor{\bsnm{Xie},~\bfnm{Chuanlong}\binits{C.}} \AND
  \bauthor{\bsnm{Zhu},~\bfnm{Lixing}\binits{L.}}
(\byear{2017}).
\btitle{Sufficient dimension reduction with mixture multivariate
  skew-elliptical distributions}.
\bjournal{Statistica Sinica}
\bpages{335--355}.
\end{barticle}
\endbibitem

\bibitem[\protect\citeauthoryear{Hall and Li}{1993}]{hall1993}
\begin{barticle}[author]
\bauthor{\bsnm{Hall},~\bfnm{P.}\binits{P.}} \AND
  \bauthor{\bsnm{Li},~\bfnm{K.~C.}\binits{K.~C.}}
(\byear{1993}).
\btitle{On almost linearity of low dimensional projections from high
  dimensional data}.
\bjournal{The Annals of Statistics}
\bvolume{47}
\bpages{867--889}.
\end{barticle}
\endbibitem

\bibitem[\protect\citeauthoryear{Li}{1991}]{li1991}
\begin{barticle}[author]
\bauthor{\bsnm{Li},~\bfnm{K.~C.}\binits{K.~C.}}
(\byear{1991}).
\btitle{Sliced inverse regression for dimension reduction (with discussion)}.
\bjournal{Journal of the American Statistical Association}
\bvolume{86}
\bpages{316--342}.
\end{barticle}
\endbibitem

\bibitem[\protect\citeauthoryear{Li}{1992}]{li1992}
\begin{barticle}[author]
\bauthor{\bsnm{Li},~\bfnm{K.}\binits{K.}}
(\byear{1992}).
\btitle{On principal {H}essian directions for data visualization and dimension
  reduction: another application of {S}tein's lemma}.
\bjournal{Journal of the American Statistical Association}
\bvolume{87}
\bpages{1025--1039}.
\end{barticle}
\endbibitem

\bibitem[\protect\citeauthoryear{Li}{2007}]{li2007sparse}
\begin{barticle}[author]
\bauthor{\bsnm{Li},~\bfnm{Lexin}\binits{L.}}
(\byear{2007}).
\btitle{Sparse sufficient dimension reduction}.
\bjournal{Biometrika}
\bvolume{94}
\bpages{603--613}.
\end{barticle}
\endbibitem

\bibitem[\protect\citeauthoryear{Li and Dong}{2009}]{li2009}
\begin{barticle}[author]
\bauthor{\bsnm{Li},~\bfnm{B.}\binits{B.}} \AND
  \bauthor{\bsnm{Dong},~\bfnm{Y.}\binits{Y.}}
(\byear{2009}).
\btitle{Dimension reduction for nonelliptically distributed predictors}.
\bjournal{The Annals of Statistics}
\bpages{1272--1298}.
\end{barticle}
\endbibitem

\bibitem[\protect\citeauthoryear{Li and Duan}{1989}]{li1989}
\begin{barticle}[author]
\bauthor{\bsnm{Li},~\bfnm{K.~C.}\binits{K.~C.}} \AND
  \bauthor{\bsnm{Duan},~\bfnm{N.}\binits{N.}}
(\byear{1989}).
\btitle{Regression analysis under link violation}.
\bjournal{The Annals of Statistics}
\bpages{1009--1052}.
\end{barticle}
\endbibitem

\bibitem[\protect\citeauthoryear{Li and Wang}{2007}]{li2007}
\begin{barticle}[author]
\bauthor{\bsnm{Li},~\bfnm{B.}\binits{B.}} \AND
  \bauthor{\bsnm{Wang},~\bfnm{S.}\binits{S.}}
(\byear{2007}).
\btitle{On directional regression for dimension reduction}.
\bjournal{Journal of the American Statistical Association}
\bvolume{35}
\bpages{2143--2172}.
\end{barticle}
\endbibitem

\bibitem[\protect\citeauthoryear{Lin, Zhao and Liu}{2016}]{lin2016sparse}
\begin{barticle}[author]
\bauthor{\bsnm{Lin},~\bfnm{Qian}\binits{Q.}},
  \bauthor{\bsnm{Zhao},~\bfnm{Zhigen}\binits{Z.}} \AND
  \bauthor{\bsnm{Liu},~\bfnm{Jun~S}\binits{J.~S.}}
(\byear{2016}).
\btitle{Sparse sliced inverse regression for high dimensional data}.
\bjournal{arXiv preprint arXiv:1611.06655}.
\end{barticle}
\endbibitem

\bibitem[\protect\citeauthoryear{Lin, Zhao and Liu}{2019}]{lin2019}
\begin{barticle}[author]
\bauthor{\bsnm{Lin},~\bfnm{Qian}\binits{Q.}},
  \bauthor{\bsnm{Zhao},~\bfnm{Zhigen}\binits{Z.}} \AND
  \bauthor{\bsnm{Liu},~\bfnm{Jun~S}\binits{J.~S.}}
(\byear{2019}).
\btitle{Sparse sliced inverse regression via lasso}.
\bjournal{Journal of the American Statistical Association}
\bvolume{114}
\bpages{1726--1739}.
\end{barticle}
\endbibitem

\bibitem[\protect\citeauthoryear{Lindsay}{1995}]{lindsay1995mixture}
\begin{binproceedings}[author]
\bauthor{\bsnm{Lindsay},~\bfnm{Bruce~G}\binits{B.~G.}}
(\byear{1995}).
\btitle{Mixture models: theory, geometry, and applications}.
\bpublisher{Ims}.
\end{binproceedings}
\endbibitem

\bibitem[\protect\citeauthoryear{Luo}{2018}]{luo2018}
\begin{barticle}[author]
\bauthor{\bsnm{Luo},~\bfnm{Wei}\binits{W.}}
(\byear{2018}).
\btitle{On the second-order inverse regression methods for a general type of
  elliptical predictors}.
\bjournal{STATISTICA SINICA}
\bvolume{28}
\bpages{1415--1436}.
\end{barticle}
\endbibitem

\bibitem[\protect\citeauthoryear{Luo}{2022}]{luo2022determine}
\begin{barticle}[author]
\bauthor{\bsnm{Luo},~\bfnm{Wei}\binits{W.}}
(\byear{2022}).
\btitle{Determine the number of clusters by data augmentation}.
\bjournal{Electronic Journal of Statistics}
\bvolume{16}
\bpages{3910--3936}.
\end{barticle}
\endbibitem

\bibitem[\protect\citeauthoryear{Luo, Li and Yin}{2014}]{luo2014}
\begin{barticle}[author]
\bauthor{\bsnm{Luo},~\bfnm{W.}\binits{W.}},
  \bauthor{\bsnm{Li},~\bfnm{B.}\binits{B.}} \AND
  \bauthor{\bsnm{Yin},~\bfnm{X.}\binits{X.}}
(\byear{2014}).
\btitle{On efficient dimension reduction with respect to a statistical
  functional of interest}.
\bjournal{The Annals of Statistics}
\bvolume{42}
\bpages{382--412}.
\end{barticle}
\endbibitem

\bibitem[\protect\citeauthoryear{Luo and Li}{2016}]{luo2016ladle}
\begin{barticle}[author]
\bauthor{\bsnm{Luo},~\bfnm{W.}\binits{W.}} \AND
  \bauthor{\bsnm{Li},~\bfnm{B.}\binits{B.}}
(\byear{2016}).
\btitle{Combining eigenvalues and variation of eigenvectors for order
  determination}.
\bjournal{Biometrika}
\bvolume{103}
\bpages{875--887}.
\end{barticle}
\endbibitem

\bibitem[\protect\citeauthoryear{Luo and Li}{2021}]{luo2021pae}
\begin{barticle}[author]
\bauthor{\bsnm{Luo},~\bfnm{Wei}\binits{W.}} \AND
  \bauthor{\bsnm{Li},~\bfnm{Bing}\binits{B.}}
(\byear{2021}).
\btitle{{On order determination by predictor augmentation}}.
\bjournal{Biometrika}
\bvolume{108}
\bpages{557-574}.
\end{barticle}
\endbibitem

\bibitem[\protect\citeauthoryear{Ma and Zhu}{2012}]{ma2012}
\begin{barticle}[author]
\bauthor{\bsnm{Ma},~\bfnm{Yanyuan}\binits{Y.}} \AND
  \bauthor{\bsnm{Zhu},~\bfnm{Liping}\binits{L.}}
(\byear{2012}).
\btitle{A Semiparametric Approach to Dimension Reduction}.
\bjournal{Journal of the American Statistical Association}
\bvolume{107}
\bpages{168-179}.
\end{barticle}
\endbibitem

\bibitem[\protect\citeauthoryear{Ma and Zhu}{2013}]{ma2013}
\begin{barticle}[author]
\bauthor{\bsnm{Ma},~\bfnm{Y.}\binits{Y.}} \AND
  \bauthor{\bsnm{Zhu},~\bfnm{L.}\binits{L.}}
(\byear{2013}).
\btitle{Efficient estimation in sufficient dimension reduction}.
\bjournal{Annals of statistics}
\bvolume{41}
\bpages{250}.
\end{barticle}
\endbibitem

\bibitem[\protect\citeauthoryear{Marin, Mengersen and
  Robert}{2005}]{marin2005bayesian}
\begin{barticle}[author]
\bauthor{\bsnm{Marin},~\bfnm{Jean-Michel}\binits{J.-M.}},
  \bauthor{\bsnm{Mengersen},~\bfnm{Kerrie}\binits{K.}} \AND
  \bauthor{\bsnm{Robert},~\bfnm{Christian~P}\binits{C.~P.}}
(\byear{2005}).
\btitle{Bayesian modelling and inference on mixtures of distributions}.
\bjournal{Handbook of statistics}
\bvolume{25}
\bpages{459--507}.
\end{barticle}
\endbibitem

\bibitem[\protect\citeauthoryear{McLachlan, Lee and
  Rathnayake}{2019a}]{mclachlan2019finite}
\begin{barticle}[author]
\bauthor{\bsnm{McLachlan},~\bfnm{Geoffrey~J}\binits{G.~J.}},
  \bauthor{\bsnm{Lee},~\bfnm{Sharon~X}\binits{S.~X.}} \AND
  \bauthor{\bsnm{Rathnayake},~\bfnm{Suren~I}\binits{S.~I.}}
(\byear{2019}a).
\btitle{Finite mixture models}.
\bjournal{Annual review of statistics and its application}
\bvolume{6}
\bpages{355--378}.
\end{barticle}
\endbibitem

\bibitem[\protect\citeauthoryear{McLachlan, Lee and
  Rathnayake}{2019b}]{mclachlan2019}
\begin{barticle}[author]
\bauthor{\bsnm{McLachlan},~\bfnm{Geoffrey~J.}\binits{G.~J.}},
  \bauthor{\bsnm{Lee},~\bfnm{Sharon~X.}\binits{S.~X.}} \AND
  \bauthor{\bsnm{Rathnayake},~\bfnm{Suren~I.}\binits{S.~I.}}
(\byear{2019}b).
\btitle{Finite Mixture Models}.
\bjournal{Annual Review of Statistics and Its Application}
\bvolume{6}
\bpages{355-378}.
\end{barticle}
\endbibitem

\bibitem[\protect\citeauthoryear{Tan et~al.}{2018}]{tan2018}
\begin{barticle}[author]
\bauthor{\bsnm{Tan},~\bfnm{Kean~Ming}\binits{K.~M.}},
  \bauthor{\bsnm{Wang},~\bfnm{Zhaoran}\binits{Z.}},
  \bauthor{\bsnm{Zhang},~\bfnm{Tong}\binits{T.}},
  \bauthor{\bsnm{Liu},~\bfnm{Han}\binits{H.}} \AND
  \bauthor{\bsnm{Cook},~\bfnm{R~Dennis}\binits{R.~D.}}
(\byear{2018}).
\btitle{A convex formulation for high-dimensional sparse sliced inverse
  regression}.
\bjournal{Biometrika}
\bvolume{105}
\bpages{769--782}.
\end{barticle}
\endbibitem

\bibitem[\protect\citeauthoryear{Titterington, Smith and
  Makov}{1985}]{titterington1985}
\begin{binproceedings}[author]
\bauthor{\bsnm{Titterington},~\bfnm{D.~Michael}\binits{D.~M.}},
  \bauthor{\bsnm{Smith},~\bfnm{Adrian F.~M.}\binits{A.~F.~M.}} \AND
  \bauthor{\bsnm{Makov},~\bfnm{Udi~E.}\binits{U.~E.}}
(\byear{1985}).
\btitle{Statistical analysis of finite mixture distributions}.
\bpublisher{John Wiley \& Sons}.
\end{binproceedings}
\endbibitem

\bibitem[\protect\citeauthoryear{Wang et~al.}{2020}]{wang2020}
\begin{barticle}[author]
\bauthor{\bsnm{Wang},~\bfnm{Qin}\binits{Q.}},
  \bauthor{\bsnm{Yin},~\bfnm{Xiangrong}\binits{X.}},
  \bauthor{\bsnm{Li},~\bfnm{Bing}\binits{B.}} \AND
  \bauthor{\bsnm{Tang},~\bfnm{Zhihui}\binits{Z.}}
(\byear{2020}).
\btitle{On aggregate dimension reduction}.
\bjournal{Statistica Sinica}
\bvolume{30}
\bpages{1027--1048}.
\end{barticle}
\endbibitem

\bibitem[\protect\citeauthoryear{Xia et~al.}{2002}]{xia2002}
\begin{barticle}[author]
\bauthor{\bsnm{Xia},~\bfnm{Yingcun}\binits{Y.}},
  \bauthor{\bsnm{Tong},~\bfnm{Howell}\binits{H.}},
  \bauthor{\bsnm{Li},~\bfnm{WK}\binits{W.}} \AND
  \bauthor{\bsnm{Zhu},~\bfnm{Li-Xing}\binits{L.-X.}}
(\byear{2002}).
\btitle{An adaptive estimation of dimension reduction space}.
\bjournal{Journal of the Royal Statistical Society: Series B (Statistical
  Methodology)}
\bvolume{64}
\bpages{363--410}.
\end{barticle}
\endbibitem

\bibitem[\protect\citeauthoryear{Yin, Li and Cook}{2008}]{yin2008}
\begin{barticle}[author]
\bauthor{\bsnm{Yin},~\bfnm{X.}\binits{X.}},
  \bauthor{\bsnm{Li},~\bfnm{B.}\binits{B.}} \AND
  \bauthor{\bsnm{Cook},~\bfnm{R.~D.}\binits{R.~D.}}
(\byear{2008}).
\btitle{Successive direction extraction for estimating the central subspace in
  a multiple-index regression}.
\bjournal{Journal of Multivariate Analysis}
\bvolume{99}
\bpages{1733--1757}.
\end{barticle}
\endbibitem

\bibitem[\protect\citeauthoryear{Zeng, Mai and Zhang}{2022}]{zeng2022}
\begin{barticle}[author]
\bauthor{\bsnm{Zeng},~\bfnm{Jing}\binits{J.}},
  \bauthor{\bsnm{Mai},~\bfnm{Qing}\binits{Q.}} \AND
  \bauthor{\bsnm{Zhang},~\bfnm{Xin}\binits{X.}}
(\byear{2022}).
\btitle{Subspace Estimation with Automatic Dimension and Variable Selection in
  Sufficient Dimension Reduction}.
\bjournal{Journal of the American Statistical Association}
\bvolume{0}
\bpages{1-13}.
\end{barticle}
\endbibitem

\bibitem[\protect\citeauthoryear{Zhao, Hautamaki and
  Fr{\"a}nti}{2008}]{zhao2008knee}
\begin{binproceedings}[author]
\bauthor{\bsnm{Zhao},~\bfnm{Qinpei}\binits{Q.}},
  \bauthor{\bsnm{Hautamaki},~\bfnm{Ville}\binits{V.}} \AND
  \bauthor{\bsnm{Fr{\"a}nti},~\bfnm{Pasi}\binits{P.}}
(\byear{2008}).
\btitle{Knee point detection in BIC for detecting the number of clusters}.
In \bbooktitle{Advanced Concepts for Intelligent Vision Systems: 10th
  International Conference, ACIVS 2008, Juan-les-Pins, France, October 20-24,
  2008. Proceedings 10}
\bpages{664--673}.
\bpublisher{Springer}.
\end{binproceedings}
\endbibitem

\end{thebibliography}

\end{document}